\documentclass[reqno,12pt,a4paper]{article}

\usepackage{amsfonts,amssymb,amsthm, 
amsmath,amscd, bbm, bm, mathabx, 
mathrsfs,delarray,subfigure}
\usepackage[dvipsnames]{xcolor}
\usepackage{hyperref}
\usepackage{cite}
\usepackage{xypic}
\usepackage{sectsty}

\usepackage{accents}

\sectionfont{\fontsize{12}{15}\selectfont}
\subsectionfont{\fontsize{12}{15}\selectfont}

\usepackage{enumitem}

\DeclareMathOperator{\ad}{ad}

\DeclareMathOperator{\id}{id}

\DeclareMathOperator{\Hom}{Hom}
\DeclareMathOperator{\iHom}{\sf Hom}
\DeclareMathOperator{\Obj}{Obj}
 
\DeclareMathOperator{\Fun}{Fun}

\DeclareMathOperator{\iMap}{\sf Map}

\DeclareMathOperator{\Vect}{Vect}

\DeclareMathOperator{\End}{End}
\DeclareMathOperator{\iEnd}{\sf End}

\DeclareMathOperator{\grtr}{grtr}

\DeclareMathOperator{\ee}{e}

\DeclareMathOperator{\Texp}{Texp}
\DeclareMathOperator{\zz}{\sf z}

\numberwithin{equation}{subsection} 
\numberwithin{subsection}{section} 

\newcommand{\ceqref}[1]{{\textcolor{blue}{\eqref{#1}}}}
\newcommand{\cref}[1]{{\textcolor{blue}{\ref{#1}}}}
\newcommand{\ccite}[1]{{\textcolor{blue}{\!\cite{#1}}}}

\newcommand{\sss}{{\hbox{$\sum$}}}
\newcommand{\ppp}{{\hbox{$\prod$}}}

\newcommand{\ms}{\scriptscriptstyle}
\newcommand{\bcdot}{{\ms\bullet}}

\newcommand{\mathsans}[1]{{{\sf #1}}}

\font\euler=eusm10 at 12.8 truept
\font\scripteuler=eusm7
\font\scriptscripteuler=eusm5 
\newtheorem{defi}{{\sf Definition}}[section]
\newtheorem{prop}{{\sf Proposition}}[section]

\newtheorem{lemma}{{\sf Lemma}}[section]
\newtheorem{exa}{{\sf Example}}[section]

\begin{document}

\vskip1.5cm
\begin{large}
{\flushleft\textcolor{blue}{\sffamily\bfseries Exact renormalization group and effective action:}}  
{\flushleft\textcolor{blue}{\sffamily\bfseries a Batalin--Vilkovisky algebraic formulation}}  
\end{large}
\vskip1.3cm
\hrule height 1.5pt
\vskip1.3cm
{\flushleft{\sffamily \bfseries Roberto Zucchini}\\
\it Dipartimento di Fisica ed Astronomia,\\
Universit\`a di Bologna,\\
I.N.F.N., sezione di Bologna,\\
viale Berti Pichat, 6/2\\
Bologna, Italy\\
Email: \textcolor{blue}{\tt \href{mailto:roberto.zucchini@unibo.it}{roberto.zucchini@unibo.it}}, 
\textcolor{blue}{\tt \href{mailto:zucchinir@bo.infn.it}{zucchinir@bo.infn.it}}}


\vskip.7cm
\vskip.6cm 
{\flushleft\sc
Abstract:} 
In the present paper, which is a mathematical follow--up of \ccite{Zucchini:2017irg}
taking inspiration from \ccite{Costello:2007ei}, 
we present an abstract formulation of exact renormalization group 
(RG) in the framework of Batalin--Vilkovisky (BV) algebra theory. 
In the first part, we work out a general algebraic and geometrical
theory of BV algebras, canonical maps, flows and flow stabilizers.
In the second part, relying on this formalism, we build a BV algebraic 
theory of the RG. In line with the graded geometric outlook of our approach, 
we adjoin the RG scale with an odd parameter and analyse in depth  
the implications of the resulting RG supersymmetry and find that the RG equation (RGE)
takes Polchinski's form \ccite{Polchinski:1983gv}.
Finally, we study abstract purely algebraic
odd symplectic free models of RG flow and effective action (EA) and the perturbation theory thereof
to illustrate and exemplify the general theory.
\vspace{2mm}

\vfil\eject

\tableofcontents

\vfil\eject

\section{\textcolor{blue}{\sffamily Introduction}}\label{sec:intro}


Renormalization group (RG) theory \ccite{Kadanoff:1966wm, Wilson:1974mb} is a powerful 
scheme designed for the study of a quantum field theory as an effective field theory (EFT) of field modes 
below a given energy scale. The theory's invariant physical content is by itself indifferent to the presence 
of the scale, which is introduced only for analysis purposes. Scale independence is encoded 
in the RG equation (RGE). There are several versions of RG theory. 
In this paper, we shall deal only with the so called exact 
RG \ccite{Polchinski:1983gv,Wetterich:1992yh} (to be called simply RG in the following), 
which is concerned mainly with the RG flow of the effective action (EA). 
See ref. \ccite{Rosten:2010vm} for a recent review. 

Batalin--Vilkovisky (BV) theory \ccite{BV1,BV2} is a very general 
quantization scheme for classical field theories characterized by gauge symmetries 
in the broadest sense. In ordinary gauge theory, it reproduces the results of 
Becchi--Rouet--Stora--Tyutin 
theory. It applies however also to more general field
theories whose gauge symmetries are neither freely acting nor involutive such as those encountered 
e. g. in supergravity and string theory. See ref. \ccite{Gomis:1994he} for an introduction to 
this topic. 

RG theory and its ramifications have been studied in a BV framework both in a 
field theoretic \ccite{Igarashi:2001mf,Igarashi:2001ey} and a physical mathematical
\ccite{Costello:2007ei,Costello:2011ams,Costello:2016faqft,Gwilliam:2016faqft,Mnev:2008sa} 
perspective. In ref. \ccite{Zucchini:2017irg}, we laid the foundations 
for a general BV RG theory. In the present paper, building upon the results obtained there, 
we present a mathematical formulation of BV RG theory in the framework of BV algebra and geometry. 
(See refs. \ccite{Elliott:2017ahb,Li:2017bv,Gwilliam:2017axm} for three recent mathematical studies similar in spirit 
to ours). The rest of this section will be devoted to introducing the reader to the subject 
and to motivate our approach to it.

\subsection{\textcolor{blue}{\sffamily Abstract BV RG theory}}\label{subsec:absbvrgth}

In BV theory, each field is paired to an antifield of opposite statistics
and the total field space is built as an odd phase space with fields 
and antifields as its canonically conjugate coordinates and momenta.
Classical Hamiltonian notions such as those of 
symplectic Laplacian, Poisson bracket and canonical map have in this way BV analogs. 
Gauge fixing is implemented by restricting to an appropriate Lagrangian submanifold of the total 
field space. Independence of the gauge fixed  theory from the choice of the submanifold requires 
the BV master action (MA) $S$ to obey the BV quantum master equation (ME) 
\begin{equation}
\varDelta S+\frac{1}{2}(S,S)=0. 
\label{ibvfix1}
\end{equation}

In a BV EFT, the whole BV structure acquires a dependence on the 
underlying energy scale $t$. The BV Laplacian and bracket as well  as 
the BV MA $S_t$ depend on $t$ and the quantum ME \ceqref{ibvfix1} holds identically in $t$,
\begin{equation}
\varDelta_tS_t+\frac{1}{2}(S_t,S_t)_t=0. \vphantom{\bigg]}
\label{ibvfix7}
\end{equation}
$S_t$ is to be identified with the BV EA. The RG flow
that governs its dependence on $t$ must be compatible with \ceqref{ibvfix7}.
The flow maps must therefore constitute a group $\varphi_{t,s}$ of 
BV canonical maps, whose pull--back action switches the BV Laplacian and bracket structure at
the scale $s$ to that at the scale $t$ and suitably relates the EA at those scales.
A closer analysis shows that 
\begin{equation}
S_t=\varphi_{t,s}{}^*S_s+r_{\varphi_{t,s}}, \vphantom{\bigg]}
\label{ibvfix7/1}
\end{equation}
where $r_{\varphi_{t,s}}=\frac{1}{2}\ln J_{\varphi_{t,s}}$ is the logarithmic Jacobian of $\varphi_{t,s}$.
$S_t$ so obeys infinitesimally an RGE of the form  \hphantom{xxxxxxxxxx}
\begin{equation}
\frac{dS_t}{dt}=\varphi^\bcdot{}_tS_t+r^\bcdot{}_{\varphi t}.
\label{ibvfix19}
\end{equation}
We justified this approach to the RG by a revisitation of BV theory 
in ref. \ccite{Zucchini:2017irg}, to which the reader is referred for 
a more comprehensive discussion of the field theoretic aspects of this matter. 


The mathematical structure of BV theory has a very elegant formulation in the language of 
BV geometry as shown originally by Schwarz in ref. \ccite{Schwarz:1992nx} 
and reviewed e. g. in \ccite{Mnev:2017oko}. Going one step further, 
one can leave aside the datum itself of a 
BV manifold and consider instead of the algebra of functions on the manifold
with its BV Laplacian and bracket structure a generic graded commutative algebra endowed 
with a formal odd Laplacian and bracket structure. 
Much of the theory can then be formulated in an abstract purely algebraic setting. \pagebreak 
This leads one to the realm of BV algebras. It is also the point of view taken in this paper.

\vfil 
It is important to stress that a purely algebraic formulation of BV theory 
is not a mere academic exercise. It makes it possible to highlight those features 
of the theory which are purely algebraic in nature and distinguish them form those
which are essentially field theoretic. Since analysis and computation 
in field theory are typically much more involved than in algebra, this constitutes a 
definite advantage.

\vfil 
In sect. \cref{sec:bvalg}, we work out in detail a general algebraic and geometrical
theory of BV algebras, canonical maps, flows and flow stabilizers. A BV algebra 
$X$ is a graded commutative algebra endowed with a degree $1$ nilpotent Laplacian $\varDelta_X$ 
and a degree $1$ graded antisymmetric, Jacobi and derivative bilinear bracket $(-,-)_X$ 
compatible in an appropriate sense (cf. defs. \cref{def:bvalg}, \cref{def:bvbrac}). 
A canonical map $\alpha:X\rightarrow Y$ of BV algebras is defined as a graded algebra isomorphism 
for which there exists a degree $0$ logarithmic Jacobian $r_\alpha\in Y$ measuring its failure to intertwine 
between the Laplacians $\varDelta_X$, $\varDelta_Y$ (cf. def. \cref{def:canmap}). A BV flow is a 
family $\chi_{t,s}:X_s\rightarrow X_t$ of canonical maps of BV algebras 
parametrized by $s,t\in T$, where $T$ is some graded manifold, and closed under 
composition (cf. def. \cref{def:bvflow}). 
A BV flow stabilizer encodes the symmetries of a BV flow (cf. def\ \cref{def:stab}).

\vfil
The definitions of BV canonical map and flow assumed in the present paper are, to the best of our knowledge, 
original. The seemingly unconventional notion of canonical map generalizes the customary one of 
bracket preserving map (and is in fact equivalent to the latter in standard BV theory). 
The notion of BV flow is obviously inspired by the analogous one of Hamiltonian mechanics.

\vfil
The formal structure described in sect. \cref{sec:bvalg} is tailored made 
for the algebraic formulation of BV RG theory detailed in sect. \cref{sec:bvrenorm}. 
In a BV algebra $X$, a BV MA is simply a degree $0$ element $S\in X$ 
satisfying the abstract version of the ME \ceqref{ibvfix1} (cf. def. \cref{def:bvqma}). 
In a parametrized family $X_t$, $t\in T$, of BV algebras all equal to the same graded commutative algebra $X$, 
a BV EA is a parametrized family of degree $0$ elements $S_t\in X$, $t\in T$, satisfying the abstract version of 
the ME \ceqref{ibvfix7} and admitting a BV flow $\chi_{t,s}:X_s\rightarrow X_t$ such that the abstract version of 
the flow relation \ceqref{ibvfix7/1} holds (cf. defs. \cref{def:bvrgfl}, \cref{def:bvrgea}). 

\vfill\eject

In  sect. \cref{sec:bvrenorm}, further developing the investigation initiated in ref.  
\ccite{Zucchini:2017irg}, we analyse in depth the conditions under which the general properties 
exhibited by the RG flow in field theory continue to hold when more general manifolds than just 
the real line $\mathbb{R}$ are used as energy scale space. The graded algebraic and geometric 
context in which we formulate BV theory suggests promoting the scale space to a graded manifold. 
Indeed, new insight can be achieved by employing the shifted tangent space $T[1]\mathbb{R}$ of 
$\mathbb{R}$ rather $\mathbb{R}$ itself for this role. This leads to an RG set--up, which we call ``extended''
to distinguish it from the customary ``basic'' set--up. 

The manifold $T[1]\mathbb{R}$ is coordinatized 
by a degree $0$ parameter $t$, identified as the usual RG energy scale, and a further degree $1$ parameter 
$\theta$. In the extend set--up, so, the BV Laplacian and bracket depend on both $t$ and $\theta$. 
The BV MA $S_{t\theta}$ also does and satisfies the extended version of the BV ME, 
\begin{equation}
\varDelta_{t\theta}S_{t\theta}+\frac{1}{2}(S_{t\theta},S_{t\theta})_{t\theta}=0. 
\label{ibveffact25/1}
\end{equation}
Analogously to the basic case, $S_{t\theta}$ is the extended BV EA.  Its RG flow is governed by a group 
$\varphi_{t\theta,s\zeta}$ of canonical maps turning the BV Laplacian and bracket structure at an extended 
scale $s,\zeta$ to that at another scale $t,\theta$ and relating the EA at those scales according to 
\hphantom{xxxxxxxxxxxxxxxx}
\begin{equation}
S_{t\theta}=\varphi_{t\theta,s\zeta}{}^*S_{s\zeta}+r_{\varphi_{ t\theta,s\zeta}}. 
\label{ibveffact25/2}
\end{equation}
At the infinitesimal level, once \ceqref{ibveffact25/1} is taken into account, 
\ceqref{ibveffact25/2} yields a more structured RGE than \ceqref{ibvfix19}. Since $\theta$
is an odd parameter, we have $\varDelta_{t\theta}=\varDelta_t+\theta \varDelta^\star {}_t$
and $(-,-)_{t\theta}=(-,-)_t\pm\theta (-,-)^\star{}_t$, where $\varDelta_t$ is the basic BV Laplacian, 
$\varDelta^\star {}_t$ is a degree $0$ Laplacian, $(-,-)_t$ is the basic BV bracket
and $(-,-)^\star{}_t$ is a degree $0$ graded symmetric bracket. 
Further, $S_{t\theta}=S_t+\theta S^\star{}_t$, where $S_t$ is to be identified with the basic BV EA and $S^\star{}_t$
is a degree $-1$ EA. By the extended RGE, $S^\star{}_t$ gets expressed in terms of $S_t$ and 
the RGE for $S_t$ takes the form 
\begin{equation}
\frac{dS_t}{dt}=\varDelta^\star{}_tS_t+\frac{1}{2}(S_t,S_t)^\star{}_t
+\bar\varphi^\bcdot{}_tS_t+\bar r^\bcdot{}_{\varphi t}, 
\label{ibveffact27}
\end{equation}
the last two summands in the right hand side being RG  ``seed'' terms. The RGE \ceqref{ibveffact27}, which was
derived by purely algebraic and geometric means, has the characteristic form of Polchinski's \ccite{Polchinski:1983gv}.  
Adding to the RG scale an odd parameter has moreover unveiled the existence of a novel RG supersymmetry. 
The above is just one example of the information that can be extracted by considering a judiciously chosen
scale manifold extending $\mathbb{R}$. The ramifications of this type of approach, 
especially in perturbation theory, are largely unexplored. 

In ref. \ccite{Zucchini:2017irg}, we illustrated a non trivial free model based on the $\mathfrak{gl}(1|1)$ 
degree $-1$ symplectic framework originally developed by Costello in ref. \ccite{Costello:2007ei}, in which the 
extended set--up described above is implemented. In  sect. \cref{sec:models}, 
we provide the proofs of the mathematical statements given there. 
We derive the expressions of the free EAs $S^0{}_t$ and $S^{0\star}{}_t$ and investigate the 
associated perturbation theory by expanding the full EAs $S_t$ and $S^\star{}_t$ as 
\begin{equation}
S_t=S^0{}_t+I_t, \qquad S^\star{}_t=S^{0\star}{}_t+I^\star{}_t, 
\label{}
\end{equation}
where $I_t$, $I^\star{}_t$ are interaction actions expressed as formal power series of $\hbar$. 
We rederive the ME and RGE of $I_t$ originally written down by Costello and obtain a further ME involving 
simultaneously $I_t$, $I^\star{}_t$. We also show that, through this latter, the RGE of $I_t$ 
can be put in a form analogous to Polchinsi's. Finally, we show that the interaction action 
$I^\star{}_t$ measures perturbatively the difference between the interacting RG flow and its free analogue. 

\vspace{-.1mm}

\subsection{\textcolor{blue}{\sffamily Outlook: beyond abstract RG theory}}\label{subsec:beyond}

In this paper, we formulate an abstract theory of BV RG flow and RGE 
from a very general perspective carrying out a thorough algebraic and geometric investigation 
of these topics based on the theory of BV algebras and their canonical isomorphisms. 
An approach of this kind has of course strengths and weaknesses.
On one hand, algebra and geometry are no substitute for quantum field theory. On the other,
an algebraic and geometric perspective can disclose novel potentially fruitful 
ways of analyzing quantum field theory. 
We already argued in favour of this point of view in the previous subsection. 
Still, further discussion about its eventual usefulness is required. 
The main challenge ahead now is testing the theory in interesting models \pagebreak
in field theory and explore possible ramifications in algebra and geometry. 

The weak coupling limit of the RG flow of 
two--dimensional non linear sigma models, originally studied by Friedan in 
\ccite{Friedan:1980th,Friedan:1980jf,Friedan:1980jm},
and the Ricci flow, introduced by Hamilton in  \ccite{Hamilton:1982rc}, 
are now known to be intimately related \ccite{Carfora:2010iz}. The former is one of the 
earliest instances where the interplay of quantum and geometric features 
show up distinctly in field theory \ccite{Grady:2017ns}.  
The latter is central in two important developments of modern geometric analysis:
Perelman's proof of Thurston's geometrization program 
\ccite{Thurston:1982zz,Thurston:1997bk} for three--manifolds and the related Poincar\'e conjecture 
\ccite{Perelman:2006un,Perelman:2006up,Perelman:2003uq}. 
The  RGE of selected field theories 
may therefore be of considerable mathematical interest. 

The Alexandrov--Kontsevich--Schwartz--Zaboronsky (AKSZ) formulation of BV theory
\ccite{Alexandrov:1995kv} is a general framework for the construction of the BV MA
of a broad class of sigma models. It has a wide range of 
field theoretic applications on one hand and important mathematical implications on the other. 
(See ref. \ccite{Ikeda:2012pv} for a recent review of this subject). In AKSZ theory, as a rule,
the BV ME is obeyed by the MA of a sigma model only when its target manifold has special 
geometric properties. Conversely, AKSZ theory can be used in principle 
to construct for any geometry a canonical sigma model with that target geometry and the
corresponding BV MA action. A suitably designed BV RGE for the 
model may conceivably result in a flow equation useful for the analysis of relevant 
geometrical and topological issues of the target manifold analogously to Ricci flow. 
This is a possibility worthy to be explored.

\vfil\eject

\section{\textcolor{blue}{\sffamily Batalin--Vilkovisky algebras and flows}}\label{sec:bvalg}

In this section, we develop the theory of BV algebras, canonical maps, flows and flow stabilizers.  
A part of the material is just a review of known results, the rest is original to the best of our knowledge.  
Our treatment is essentially mathematical, the material presented below providing mainly the formal foundations 
of the BV theory of the RG presented in sects. \cref{sec:bvrenorm} and \cref{sec:models}. 

The notion of BV flow is a formalization of that of BV RG flow and so plays a basic role in the BV RG theory of 
sect. \cref{sec:bvrenorm}. The notion of BV stabilizers emerges naturally enough within BV flow theory, but so 
far we have found no application of it in BV RG theory. We elaborate on it anyway for its potential 
in the analysis of the symmetries of the RGE. 

We assume that the reader is acquainted with the basics of abstract BV theory and graded geometry. For background, 
we refer him/her again to \ccite{Mnev:2017oko} and also to ref. \ccite{Cattaneo:2010re}.
Unless otherwise stated, we shall deal mostly with real unital associative graded commutative algebras
or graded commutative algebras for short. Commutators are always assumed to be graded.


\subsection{\textcolor{blue}{\sffamily BV algebras}}\label{subsec:bvalg}

In this subsection, we review the basic notions and the main 
properties of BV algebras. This material is standard and is presented only to set our notation 
and for later reference. 

We introduce BV structures through BV Laplacians.

\begin{defi} \label{def:bvalg}
A BV algebra is a graded commutative algebra $X$ equipped with a BV Laplacian $\varDelta_X$, 
that is a linear endomorphism $\varDelta_X:X\rightarrow X$ enjoying the following properties:
\begin{align}
&|\varDelta_Xf|=|f|+1,
\vphantom{\Big]}
\label{bvalg1}
\\
&\varDelta_X(fgh)=-\varDelta_Xfgh-(-1)^{|f|}f\varDelta_Xgh-(-1)^{|f|+|g|}fg\varDelta_Xh
\vphantom{\Big]}
\label{bvalg2}
\\
&\hspace{3cm}+\varDelta_X(fg)h+(-1)^{(|f|+1)|g|}g\varDelta_X(fh)+(-1)^{|f|}f\varDelta_X(gh)
\vphantom{\Big]}
\nonumber
\end{align}
for $f,g,h\in X$ and \hphantom{xxxxxxxxxxxxxxxxxxxx}
\begin{align}
&\varDelta_X{}^2=0,
\vphantom{\Big]}
\label{bvalg3}
\\
&\varDelta_X1_X=0.
\vphantom{\Big]}
\label{bvalg4}
\end{align}
\end{defi}

\noindent
Stated in words, a BV algebra is graded commutative algebra endowed with a 
degree $1$, second order differential operator $\varDelta_X$ squaring to $0$
and annihilating the unit $1_X$ of $X$. 

\noindent
By \ceqref{bvalg3}, if $X$ is BV algebra, $(X,\varDelta_X)$ is a cochain complex. To this there 
is attached a cohomology $H^*(X,\varDelta_X)$, the BV algebra cohomology.

From now on, we consider a fixed BV algebra $X$ with Laplacian $\varDelta_X$.

\begin{defi}  \label{def:bvbrac}
The BV bracket of $X$ is the bilinear map $(-,-)_X:X\times X\rightarrow X$ defined by \hphantom{xxxxxxxxxx}
\begin{equation}
(f,g)_X=(-1)^{|f|}\big(\varDelta_X(fg)-\varDelta_Xfg-(-1)^{|f|}f\varDelta_Xg\big) \vphantom{\bigg]}
\label{bvalg5}
\end{equation}
for $f,g\in X$. 
\end{defi}

\noindent The BV bracket of $X$ is therefore completely determined by 
the BV Laplacian of $X$ as the failure of the latter to be a degree $1$ derivation of $X$. 

The following well--known proposition is the result of straightforward calculations based 
on relation \eqref{bvalg5}. 

\begin{prop} \label{prop:bvbracgerst}
The bracket $(-,-)_X$ of the BV algebra $X$ obeys the following relations, 
\begin{align}
&|(f,g)_X|=|f|+|g|+1,
\vphantom{\Big]}
\label{bvalg6}
\\
&(f,g)_X+(-1)^{(|f|+1)(|g|+1)}(g,f)_X=0,
\vphantom{\Big]}
\label{bvalg7}
\\
&(-1)^{(|h|+1)(|f|+1)}(f,(g,h)_X)_X
+(-1)^{(|f|+1)(|g|+1)}(g,(h,f)_X)_X
\vphantom{\Big]}
\label{bvalg8}
\\
&\hspace{5cm}+(-1)^{(|g|+1)(|h|+1)}(h,(f,g)_X)_X=0,
\vphantom{\Big]}
\nonumber
\\
&(f,gh)_X-(f,g)_Xh-(-1)^{(|f|+1)|g|}g(f,h)_X=0, 
\vphantom{\Big]}
\label{bvalg9}
\\
&(f,1_X)_X=0 
\vphantom{\Big]}
\label{bvalg10}
\end{align}
\vspace{-.8cm}\eject\noindent
for $f,g,h\in X$. 
\end{prop}

\noindent
Hence, $(-,-)_X$ has degree $1$, 
is shifted graded antisymmetric and obeys the shifted graded Jacobi identity. 
Further, $\ad_Xf:=(f,-)_X$ is a degree $|f|+1$ derivation of $X$. 
This makes any BV algebra a Gerstenhaber algebra, an odd version of a Poisson algebra. 
Many notions and constructions of Poisson algebra theory can in this way be extended to BV algebras. 

The BV Laplacian $\varDelta_X$ differentiates the BV brackets.

\begin{prop} One has 
\begin{equation}
\varDelta_X(f,g)_X=(\varDelta_Xf,g)_X+(-1)^{|f|+1}(f,\varDelta_Xg)_X 
\label{bvalg11}
\end{equation}
for $f,g\in X$. 
\end{prop}

The BV algebra $X$ is characterized by its BV center $C_X$, the subalgebra formed by all elements 
$c\in X$ such that $(f,c)_X=0$ for all $f\in X$. We note that by \ceqref{bvalg10} $\mathbb{R}1_X\subset C_X$. 

\begin{defi}\label{def:bvnonsing}
The BV algebra $X$ is said non singular  if $C_X=\mathbb{R}1_X$. In that case, the underlying BV Laplacian 
$\varDelta_X$ is also called non singular. 
\end{defi}

In this paper, we shall consider exclusively non singular BV algebras. 


\subsection{\textcolor{blue}{\sffamily BV canonical maps}}\label{subsec:bvcanon}

In this subsection, we introduce the notion of BV canonical map in
a way that is original to the best of our knowledge.

Let $X$, $Y$ be non singular BV algebras (cf. defs. \cref{def:bvalg}, \cref{def:bvnonsing}). 

\begin{defi} \label{def:canmap}
A canonical map from $X$ to $Y$ is an invertible graded commutative algebra morphism $\alpha:X\rightarrow Y$
with the following property. There exists an element $r_\alpha\in Y$ such that \hphantom{xxxxxxx}
\begin{align}
&|r_\alpha|=0,
\vphantom{\Big]}
\label{bvcanon1}
\\
&\varDelta_Y\alpha-\alpha\varDelta_X+\ad_Yr_\alpha\alpha=0.
\vphantom{\Big]}
\label{bvcanon2}
\end{align}
\vspace{-1cm}\eject\noindent
$r_\alpha$ is called the logarithmic Jacobian of $\alpha$.
\end{defi}

\noindent The name given to $r_\alpha$ derives from fact that in the functional 
formulations of BV theory $r_\alpha$ is half the logarithm of the Jacobian Berezinian of $\alpha$ 
\ccite{Zucchini:2017irg}. We note that in the present formulation, $r_\alpha$ is defined 
only modulo $C_Y=\mathbb{R}1_Y$. 

\begin{prop} \label{prop:qmera}
The Jacobian $r_\alpha$ of a canonical map $\alpha:X\rightarrow Y$ satisfies 
\begin{equation}
\varDelta_Yr_\alpha+\frac{1}{2}(r_\alpha,r_\alpha)_Y=0.
\label{bvcanon3}
\end{equation}
\end{prop}

\noindent
Eq. \ceqref{bvcanon3} 
has the form of a quantum ME. 

\begin{proof} 
From \ceqref{bvcanon1} and \ceqref{bvcanon2}, using \ceqref{bvalg3} and \ceqref{bvalg11}, we find 
\begin{align}
0&=\alpha\varDelta_X\varDelta_Xf
\vphantom{\Big]}
\label{}
\\
&=\varDelta_Y\alpha\varDelta_Xf+(r_\alpha,\alpha\varDelta_Xf)_Y
\vphantom{\Big]}
\nonumber
\\
&=\varDelta_Y(\varDelta_Y\alpha f+(r_\alpha,\alpha f)_Y)+(r_\alpha,\varDelta_Y\alpha f+(r_\alpha,\alpha f)_Y)_Y
\vphantom{\Big]}
\nonumber
\\
&=(\varDelta_Yr_\alpha+(r_\alpha,r_\alpha)_Y/2,\alpha f)_Y,
\vphantom{\Big]}
\nonumber
\end{align}
where $f\in X$. So, $\varDelta_Yr_\alpha+(r_\alpha,r_\alpha)_Y/2\in C_Y=\mathbb{R}1_Y$. 
As $|\varDelta_Yr_\alpha+(r_\alpha,r_\alpha)_Y/2|=1$, \ceqref{bvcanon3} holds. 
\end{proof}

Canonical maps preserve BV brackets. 

\begin{prop} \label{prop:canmapbvbrac}
Let $\alpha:X\rightarrow Y$ be a canonical map. Then, for $f,g\in X$, one has 
\begin{equation}
\alpha(f,g)_X=(\alpha f, \alpha g)_Y. \vphantom{\bigg]}
\label{bvcanon4}
\end{equation}
\end{prop}

\begin{proof} By virtue of \ceqref{bvalg5}, using \ceqref{bvcanon1}, \ceqref{bvcanon2}, for $f,g\in X$
\begin{align}
\alpha(f,g)_X
&=(-1)^{|f|}\big(\alpha\varDelta_X(fg)-\alpha\varDelta_Xf\alpha g-(-1)^{|f|}\alpha f\alpha\varDelta_Xg\big)
\vphantom{\Big]}
\label{}
\\
&=(-1)^{|f|}\big(\varDelta_Y(\alpha f\alpha g)+(r_\alpha, \alpha f\alpha g)_Y
\vphantom{\Big]}
\nonumber
\\
&\hspace{1.5cm}-(\varDelta_Y\alpha f+(r_\alpha,\alpha f)_Y)\alpha g 
-(-1)^{|f|}\alpha f(\varDelta_Y\alpha g+(r_\alpha,\alpha g)_Y\big)
\vphantom{\Big]}
\nonumber
\\
&=(\alpha f, \alpha g)_Y,
\vphantom{\Big]}
\nonumber
\end{align}
as claimed. 
\end{proof}

\noindent We have not been able to show that every BV bracket preserving graded algebra isomorphism
is canonical in the sense of def. \cref{def:canmap}. Since in the analysis of the BV formulation 
of the RG the logarithmic Jacobian plays a basic role, however, the definition 
of canonical map given in  \cref{def:canmap} is definitely the most natural. 

If $X,Y,Z$ are BV algebras
and $\alpha:X\rightarrow Y$, $\beta:Y\rightarrow Z$ are canonical maps, $\beta\alpha: X\rightarrow Z$ also is, the Jacobians
of $\alpha$, $\beta$ and $\beta\alpha$ being related as 
\begin{equation}
r_{\beta\alpha}=r_\beta+\beta r_\alpha \qquad \text{mod $\mathbb{R}1_Z$}. 
\label{bvcanon5}
\end{equation}
If $X$ is a BV algebra, the identity $\id_X:X\rightarrow X$ of $X$ is canonical with Jacobian
\begin{equation}
r_{\id_X}=0 \qquad \text{mod $\mathbb{R}1_X$}. 
\label{bvcanon6}
\end{equation}
If $X,Y$ are BV algebras and $\alpha:X\rightarrow Y$ is a canonical map, so is $\alpha^{-1}:Y\rightarrow X$, 
the Jacobians of $\alpha$, $\alpha^{-1}$ being related as 
\begin{equation}
r_{\alpha^{-1}}=-\alpha^{-1}r_\alpha \qquad \text{mod $\mathbb{R}1_X$}. 
\label{bvcanon7}
\end{equation}
By \ceqref{bvcanon5}--\ceqref{bvcanon7}, BV algebras and BV algebra canonical maps form a groupoid 
$\mathsans{BVAlg}$. $\mathsans{BVAlg}$ is a proper subcategory  of the category $\mathsans{grcAlg}$ of 
graded commutative algebras and algebra morphisms. There is thus 
a canonical action of $\mathsans{BVAlg}$ on the linear fibered 
set $\mathcal{X}_{\mathsans{BVAlg}}=\coprod_{X\in\Obj_{\mathsans{BVAlg}}}X/\mathbb{R}1_X$.
\ceqref{bvcanon5} entails that the logarithmic Jacobian map 
$r:\Hom_{\mathsans{BVAlg}}\rightarrow \mathcal{X}_{\mathsans{BVAlg}}$ 
is a $1$--cocycle of $\mathsans{BVAlg}$ valued in $\mathcal{X}_{\mathsans{BVAlg}}$. 

Let $X$ be a non singular BV algebra. 

\begin{defi} \label{def:infbvcan}
An infinitesimal canonical map of $X$ is a degree $0$ derivation $\xi$ 
of the graded commutative algebra $X$ with the following property. There exists an element $e_\xi\in X$ 
such that \hphantom{xxxxxxx}
\begin{align}
&|e_\xi|=0,
\vphantom{\Big]}
\label{bvcanon8}
\\
&\varDelta_X\xi-\xi\varDelta_X+\ad_X e_\xi=0.
\vphantom{\Big]}
\label{bvcanon9}
\end{align}
$e_\xi$ is called the logarithmic Jacobian of $\xi$.
\end{defi}

\noindent
In formal terms, an infinitesimal canonical map $\xi$ describes a finite canonical map 
$\id_X+\xi$ infinitely close to the identity $\id_X$, \ceqref{bvcanon9} being the linearization
of \ceqref{bvcanon2}.
Again, $e_\xi$ is defined only modulo $\mathbb{R}1_X$. 

\begin{prop} \label{prop:qmeex}
The Jacobian $e_\xi$ of an infinitesimal canonical map $\xi$ of $X$ obeys 
\begin{equation}
\varDelta_Xe_\xi=0.
\label{bvcanon10}
\end{equation}
\end{prop}

\begin{proof} From \ceqref{bvcanon8} and \ceqref{bvcanon9}, using \ceqref{bvalg3} and \ceqref{bvalg11}, we find 
\begin{align}
0&=\xi\varDelta_X\varDelta_Xf
\vphantom{\Big]}
\label{}
\\
&=\varDelta_X\xi\varDelta_Xf+(e_\xi,\varDelta_X f)_X
\vphantom{\Big]}
\nonumber
\\
&=\varDelta_X(\varDelta_X\xi f+(e_\xi,f)_X)+(e_\xi,\varDelta_X f)_X
\vphantom{\Big]}
\nonumber
\\
&=(\varDelta_Xe_\xi,f)_X,
\vphantom{\Big]}
\nonumber
\end{align}
where $f\in X$. Hence, $\varDelta_Xe_\xi\in C_X=\mathbb{R}1_X$. As  $|\varDelta_Xe_\xi|=1$
\ceqref{bvcanon10} holds.
\end{proof}

\noindent
As expected, \ceqref{bvcanon10}  is the linearized version of \ceqref{bvcanon3}. 

Infinitesimal canonical maps preserve BV brackets in the appropriate sense. 

\begin{prop} \label{prop:infcanmapbvbrac}
Let $\xi$ be an infinitesimal canonical map of $X$. For any $f,g\in X$, 
\begin{equation}
\xi(f,g)_X-(\xi f,g)_X-(f,\xi g)_X=0. 
\label{bvcanon11}
\end{equation}
\end{prop}

\begin{proof} By virtue of \ceqref{bvalg5}, using \ceqref{bvcanon8}, \ceqref{bvcanon9}, for $f,g\in X$
\begin{align}
\xi(f,g)_X
&=(-1)^{|f|}\big(\xi\varDelta_X(fg)-\xi\varDelta_Xfg 
\vphantom{\Big]}
\label{}
\\
&\hspace{2cm}-\varDelta_Xf\xi g
-(-1)^{|f|}\xi f\varDelta_Xg-(-1)^{|f|}f\xi\varDelta_Xg\big)
\vphantom{\Big]}
\nonumber
\\
&=(-1)^{|f|}\big(\varDelta_X(\xi fg+f\xi g)+(e_\xi,fg)_X
\vphantom{\Big]}
\nonumber
\\
&\hspace{2cm}-(\varDelta_X\xi f+(e_\xi,f)_X)g-\varDelta_Xf\xi g
\vphantom{\Big]}
\nonumber
\\
&\hspace{2cm}-(-1)^{|f|}\xi f\varDelta_Xg-(-1)^{|f|}f(\varDelta_X\xi g+(e_\xi,g)_X)\big)
\vphantom{\Big]}
\nonumber
\\
&=(\xi f,g)_X+(f,\xi g)_X,
\vphantom{\Big]}
\nonumber
\end{align}
\vspace{-.8cm}\eject\noindent
as claimed. 
\end{proof}

\noindent
Again, as expected, \ceqref{bvcanon11} is the linearized version of \ceqref{bvcanon4}. 

The infinitesimal canonical maps of $X$ form a Lie algebra $\mathfrak{can}_X$, which is a proper Lie
subalgebra of the Lie algebra $\mathfrak{der}_{X0}$ of degree $0$ derivations of the graded commutative algebra 
$X$. Indeed, as is easily verified, if $\xi,\eta$ are 
infinitesimal canonical maps of $X$, so is $[\xi,\eta]$, its Jacobian being 
\begin{equation}
e_{[\xi,\eta]}=\xi e_\eta-\eta e_\xi \qquad\text{mod $\mathbb{R}1_X$}. 
\label{bvcanon12}
\end{equation}

There is an important subset of infinitesimal canonical maps, those of the form
$\ad_X x$ with $x\in X$, $|x|=-1$. 
Indeed, it is straightforward to check from \ceqref{bvalg11} that $\ad_X x$  satisfies the conditions
\ceqref{bvcanon8}, \ceqref{bvcanon9} with Jacobian 
\begin{equation}
e_{\ad_X x}=-\varDelta_X x \qquad\text{mod $\mathbb{R}1_X$}. 
\label{bvcanon14}
\end{equation}
The infinitesimal canonical maps of the form $\ad_X x$ form a Lie subalgebra 
$\mathfrak{can}^0{}_X$ of the infinitesimal canonical map Lie algebra $\mathfrak{can}_X$.


\subsection{\textcolor{blue}{\sffamily Families of  BV structures}}\label{subsec:bvfam}

In this subsection, we introduce and study families of BV structures. 

The theory developed below involves differentiation of functions defined on a parameter graded 
manifold and valued in a graded commutative algebra. In order not to commit ourselves with any definite notion 
of differentiability and remain as general as possible, we proceed in a completely formal way.

Let $M$ be a graded manifold and $A$ a graded commutative algebra.

\begin{defi} \label{def:formdif}
A formal differentiation structure of $A$ over $M$ consists of the following data.
\begin{enumerate}

\item A graded Lie algebra $\Vect(M)$ of vector fields over $M$.

\item 
A graded $\Vect(M)$--module $\iMap(M,\iHom_{\mathsans{grVct}}(A^{\otimes k},A))$ of 
maps of $M$ into $\iHom_{\mathsans{grVct}}(A^{\otimes k},A)$ containing 
all constant maps for each integer $k\geq 0$. 

\end{enumerate}
It is further required that the operation of morphism nested composition 
\begin{align}
\iHom_{\mathsans{grVct}}(A^{\otimes k},A)\times \ppp_{r=1}^k
\iHom_{\mathsans{grVct}}(A^{\otimes k_r},A)\rightarrow \iHom_{\mathsans{grVct}}(A^{\otimes \Sigma_{r=1}^kk_r},A)
\vphantom{\bigg]}
\label{nbvfam1}
\end{align}
induces an operation of morphism valued map nested composition 
\begin{align}
&\iMap(M,\iHom_{\mathsans{grVct}}(A^{\otimes k},A))\times \ppp_{r=1}^k
\iMap(M,\iHom_{\mathsans{grVct}}(A^{\otimes k_r},A))
\label{nbvfam2}
\\
&\hspace{7cm}\rightarrow \iMap(M,\iHom_{\mathsans{grVct}}(A^{\otimes \Sigma_{r=1}^kk_r},A))
\nonumber
\end{align}
and that the Lie module action of $\Vect(M)$ is graded Leibniz with respect to it. Finally, it is requested that
$\Vect(M)$ acts trivially on constant maps.
\end{defi} 

\noindent
Above, $\mathsans{grVct}$ denotes the category of graded vector spaces and 
$\iHom_{\mathsans{grVct}}(A^{\otimes k},A)$ the graded vector space
internal morphism space of $A^{\otimes k}$ into $A$. By a 
map of $M$ into $\iHom_{\mathsans{grVct}}(A^{\otimes k},A)$ we mean one of the local form
$\varPhi(m)=\sss_R m_s{}^R\varPhi_R(m_b)$, where $m_b$ and $m_s$ are the body and soul local coordinates 
of $M$, $R$ is a soul multiindex, $m_s{}^R$ is the corresponding soul coordinate multipower and the $\varPhi_R$
are maps  valued in $\iHom_{\mathsans{grVct}}(A^{\otimes k},A)$ non zero only for finitely many values of $R$. 

For every $k$, $\iHom_{\mathsans{grVct}}(A^{\otimes k},A)$ is contained in $\iMap(M,\iHom_{\mathsans{grVct}}(A^{\otimes k},A))$ 
as its subspace of constant maps. In particular, $k$--fold multiplication of $A$ is in 
$\iMap(M,\iHom_{\mathsans{grVct}}(A^{\otimes k},A))$. The morphism nested composition \ceqref{nbvfam1} 
is the mapping associating the composite $T\circ_{\mathsans{grVct}}(T_1,\dots,T_k)$ with $(T,T_1,\dots,T_k)$.

For $k=0$, 
one has $\iHom_{\mathsans{grVct}}(A^{\otimes k},A)=A$ and $\iMap(M,\iHom_{\mathsans{grVct}}(A^{\otimes k},A))=\iMap(M,A)$. 
Further, by the nested composition operation \ceqref{nbvfam2} with $k_r=0$, 
pointwise multiplication renders $\iMap(M,A)$ a graded commutative algebra of maps of $M$ into $A$. 

For any vector field $X\in\Vect(M)$ and  map $\varPhi\in\iMap(M,\iHom_{\mathsans{grVct}}(A^{\otimes k},A))$,
a map $X\varPhi\in\iMap(M,\iHom_{\mathsans{grVct}}(A^{\otimes k},A))$ is defined, that is linear in $X$ and $\varPhi$, 
has the property that $X\varPhi=0$ when $\varPhi$ is constant 
and satisfies $[X,X']\varPhi=X(X'\varPhi)-(-1)^{|X||X'|}X'(X\varPhi)$. The graded Leibniz requirement can be stated as 
\begin{equation}
X(\varPhi u)=(X\varPhi) u+(-1)^{|X||\varPhi|}\varPhi(Xu) \vphantom{\bigg]}
\label{nbvfam3}
\end{equation}
for $\varPhi\in\iMap(M,\iEnd_{\mathsans{grVct}}(A))$ and $u\in\iMap(M,A)$ for $k=1$, \pagebreak 
\begin{equation}
X(\varPhi(u,v))=(X\varPhi)(u,v)+(-1)^{|X||\varPhi|}\varPhi(Xu,v)+(-1)^{|X|(|\varPhi|+|u|)}\varPhi(u,Xv)
\label{nbvfam4}
\end{equation}
for $\varPhi\in\iMap(M,\iHom_{\mathsans{grVct}}(A^{\otimes 2},A))$ and $u,v\in\iMap(M,A)$ for $k=2$, etc. 
When $\varPhi$ is the $2$--fold multiplication map, we recover the standard Leibniz relation
\begin{equation}
X(uv)=Xuv+(-1)^{|X||u|}uXv. \vphantom{\bigg]}
\label{nbvfam5}
\end{equation}  

It is useful to adjoin to the formal differentiation structure  the 
vector spaces $\iMap(M^p,\iHom_{\mathsans{grVct}}(A^{\otimes k},A))$ of all 
maps $\varPhi$ of $M^p$ into $\iHom_{\mathsans{grVct}}(A^{\otimes k},A)$ with the property
that all the maps obtained by fixing  $p-1$ of the their $p$ arguments belong to 
$\iMap(M,\iHom_{\mathsans{grVct}}(A^{\otimes k},A))$. Notice that when $p=1$ we recover the 
space $\iMap(M,\iHom_{\mathsans{grVct}}(A^{\otimes k},A))$ as defined previously. 

Although the above formulation of differentiation structure of $A$ over $M$ 
is rigorous and economical in terms of number of assumptions made,
it is somewhat notationally and terminologically unwieldy. In what follows, we think of a map 
$\varPhi\in\iMap(M,\iHom_{\mathsans{grVct}}(A^{\otimes k},A))$ as a parametrized collection of morphisms
$\varPhi_m:A^{\otimes k}\rightarrow A$, $m\in M$, and call it a family of morphism of 
$A^{\otimes k}$ into $A$ over $M$. Similarly, we think of a map $\varPhi\in\iMap(M^p,\iHom_{\mathsans{grVct}}(A^{\otimes k},A))$ 
as a parametrized collection of morphisms $\varPhi_{m_1,\ldots,m_p}:A^{\otimes k}\rightarrow A$, $m_1,\ldots,m_p\in M$, 
and call it a family of morphism of $A^{\otimes k}$ into $A$ over $M^p$. We shall also often omit 
to indicate the range of variation of $m$, respectively $m_1,\ldots,m_p$, when no confusion can arise. 

%
%
%
%
%
%

Below,  we consider a graded commutative algebra $X$,
a parameter graded manifold $T$ and a formal differentiation structure 
of $X$ over $T$ and families parametrized by $t\in T$. 
In the analysis of the RGE in BV theory, $X$ will be 
the BV field space and $T$ will be the energy scale parameter manifold. Albeit
there is no a priori
restriction to the generality of $T$, it may help the reader to keep in mind the 
main examples treated in this paper, viz $T=\mathbb{R}$ and $T=T[1]\mathbb{R}$. 

\begin{defi} \label{def:tfamD}
A non singular BV Laplacian family over $T$ is a family $\varDelta_t:X\rightarrow$ $X$
of vector space endomorphisms of $X$ over $T$ such that, for each $t$, 
$\varDelta_t$ is a non singular BV Laplacian  (cf. defs. \cref{def:bvalg} and \cref{def:bvnonsing}). 
\end{defi}

\noindent
For fixed $t$, we shall denote by $X_t$ the BV algebra resulting from endowing the graded algebra $X$ 
with the BV Laplacian $\varDelta_t$. 

Let $\varDelta_t$ be a non singular BV Laplacian family over $T$. 

\begin{defi} \label{def:tfambvbrac}
The BV bracket family over $T$ corresponding to $\varDelta_t$ is the family 
$(-,-)_t:X\times X\rightarrow X$ of bilinear brackets over $T$ defined according 
to \ceqref{bvalg5}.
\end{defi}

\noindent 
Explicitly, one has 
\begin{equation}
(f,g)_t=(-1)^{|f|}\big(\varDelta_t(fg)-\varDelta_tfg-(-1)^{|f|}f\varDelta_tg\big),
\label{bvfam5}
\end{equation}
where $f,g\in X$. $(-,-)_t$ is the BV bracket of the BV algebra $X_t$.

It is possible to use a vector field $D$ of $T$ to probe the BV Laplacian family 
$\varDelta_t$. In BV RG theory, this generates the differential operators 
appearing in the relevant RGEs. In the main examples considered in this paper, 
$T=\mathbb{R}$ and $T[1]\mathbb{R}$, coordinatized with the standard coordinates $t$ and 
$t,\theta$, we have $D=d/dt$ and $D=\partial/\partial t,\partial/\partial\theta$, 
respectively. Again, other cases may be envisaged. 

Let $D\in\Vect(T)$ be a fixed vector field of $T$. 

\begin{defi} \label{def:DdertfamD}
The $D$--derived BV Laplacian family over $T$ of $\varDelta_t$ is the family 
$\varDelta^D{}_t:X\rightarrow X$ of vector space endomorphisms of $X$ over $T$ defined by 
\begin{equation}
\varDelta^D{}_t=D_t\varDelta_t.  
\label{bvfam0}
\end{equation}
\end{defi}


\noindent 
The basic  relations \ceqref{bvalg1} --\ceqref{bvalg4} satisfied by $\varDelta_t$
imply the following. 

\begin{prop} \label{prop:dbvdeltabasic}
For each $t$, $\varDelta^D{}_t$ obeys 
\begin{align}
&|\varDelta^D{}_tf|=|f|+|D|+1,
\vphantom{\Big]}
\label{bvfam1}
\\
&\varDelta^D{}_t(fgh)=-\varDelta^D{}_tfgh-(-1)^{(|D|+1)|f|}f\varDelta^D{}_tgf
\vphantom{\Big]}
\label{bvfam2}
\\
&\hspace{2.4cm}-(-1)^{(|D|+1)(|f|+|g|)}fg\varDelta^D{}_th
+\varDelta^D{}_t(fg)h
\vphantom{\Big]}
\nonumber
\\
&\hspace{2.4cm}+(-1)^{(|f|+|D|+1)|g|}g\varDelta^D{}_t(fh)+(-1)^{(|D|+1)|f|}f\varDelta^D{}_t(gh),
\vphantom{\Big]}
\nonumber
\end{align}
where $f,g,h\in X$ and \hphantom{xxxxxxxxxxxxxxx} 
\begin{align}
&[\varDelta_t,\varDelta^D{}_t]=0,
\vphantom{\Big]}
\label{bvfam3}
\\
&\varDelta^D{}_t1_X=0.
\vphantom{\Big]}
\label{bvfam4}
\end{align}
\end{prop}

\noindent
So, for fixed $t$, $\varDelta^D_t$ is a degree $|D|+1$
second order differential operator commuting with $\varDelta_t$ and annihilating $1_X$. 

We can use the vector field $D$ to probe in similar way also the BV bracket family  
$(-,-)_t$.

\begin{defi} \label{def:Ddertfambvbrac}
The  $D$--derived BV bracket family over $T$ associated with  $(-,-)_t$ is the family
$(-,-)^D{}_t:X\times X\rightarrow X$ of bilinear brackets over $T$ defined by 
\begin{equation}
(f,g)^D{}_t=(-1)^{|D||f|}D_t(f,g)_t 
\label{bvfam6}
\end{equation}
for $f,g\in X$ 
\end{defi}


\noindent
$(-,-)^D{}_t$ obeys certain relations following from 
the basic relations \ceqref{bvalg6}-- \ceqref{bvalg10} satisfied by $(-,-)_t$.

\begin{prop} \label{prop:bvbracdprop}
For each  $t$, one has 
\begin{align}
&|(f,g)^D{}_t|=|f|+|g|+|D|+1,
\vphantom{\Big]}
\label{bvfam8}
\\
&(f,g)^D{}_t+(-1)^{(|f|+|D|+1)(|g|+|D|+1)+|D|}(g,f)^D{}_t=0,
\vphantom{\Big]}
\label{bvfam9}
\\
&(-1)^{(|h|+|D|+1)(|f|+1)}\big((f,(g,h)_t)^D{}_t+(-1)^{|D|(|g|+1)}(f,(g,h)^D{}_t)_t\big)
\vphantom{\Big]}
\label{bvfam10}
\\
&+(-1)^{(|f|+|D|+1)(|g|+1)}\big((g,(h,f)_t)^D{}_t+(-1)^{|D|(|h|+1)}(g,(h,f)^D{}_t)_t\big)
\vphantom{\Big]}
\nonumber
\\
&+(-1)^{(|g|+|D|+1)(|h|+1)}\big((h,(f,g)_t)^D{}_t+(-1)^{|D|(|f|+1)}(h,(f,g)^D{}_t)_t\big)=0,
\vphantom{\Big]}
\nonumber
\\
&(f,gh)^D{}_t-(f,g)^D{}_th-(-1)^{(|f|+|D|+1)|g|}g(f,h)^D{}_t=0, 
\vphantom{\Big]}
\label{bvfam11}
\\
&(f,1_X)^D{}_t=0,
\vphantom{\Big]}
\label{bvfam12}
\end{align} 
where $f,g,h\in X$,
\end{prop}

\noindent
Hence, $(-,-)^D{}_t$ has degree $|D|+1$. Furthermore,
$(-,-)^D{}_t$ is shifted graded antisymmetric (graded symmetric) when $D$ is even (odd)
and obeys an appropriate generalization of the (shifted) graded Jacobi identity. Finally,
$(f,-)^D{}_t$ is a degree $|f|+|D|+1$ derivation of $X$. 

The results illustrated above find immediate application in the theory of BV flows 
of next subsection. 


\subsection{\textcolor{blue}{\sffamily BV flows}}\label{subsec:bvflow}

In this subsection, we introduce BV flows and study their properties in depth. 
The notion of BV flow is a formalization of that of RG 
flow in BV theory studied in the next section. For this reason it is of great salience for us. 

As in the previous subsection, we consider a graded commutative algebra $X$ and  
a parameter graded manifold $T$ and select a formal differentiation structure of $X$ over $T$
(cf. def. \cref{def:formdif}). 
Further, we assume that a non singular BV Laplacian family 
$\varDelta_t$ over $T$ and its appended  BV bracket family $(-,-)_t$ are given 
(cf. defs. \cref{def:tfamD}, \cref{def:tfambvbrac}) making $X$ a BV algebra $X_t$ 
for each $t$. 

A BV flow relates 
the BV structures corresponding to different points of the 
parameter manifold $T$ by means of canonical maps. Again, for the sake of concreteness, 
it may be helpful to keep in mind our main examples: $T=\mathbb{R}$ and $T=T[1]\mathbb{R}$.

\begin{defi} \label{def:bvflow}
A BV flow along $T$ (relative to the BV Laplacian family $\varDelta_t$) is a 
family $\chi_{t,s}:X\rightarrow X$ of algebra automorphisms of $X$ over $T^2$ obeying  
\begin{align}
&\chi_{u,s}=\chi_{u,t}\chi_{t,s},
\vphantom{\Big]}
\label{bvflow1}
\\
&\chi_{s,t}=\chi_{t,s}{}^{-1},
\vphantom{\Big]}
\label{bvflow2}
\\
&\chi_{s,s}=\id_X
\vphantom{\Big]}
\label{bvflow3}
\end{align}
and such that $\chi_{t,s}:X_s\rightarrow X_t$ is a canonical map 
(cf. def. \cref{def:canmap}). It is further required that 
the $\mathbb{R}1_X$ indeterminacy of the logarithmic Jacobians $r_{\chi t,s}$ of the $\chi_{t,s}$ 
can be fixed for each $s,t\in T$ in such a way that the resulting collection $r_{\chi t,s}$ 
is a family of degree $0$ elements of $X$ over $T^2$ satisfying 
\begin{align}
&r_{\chi u,s}=r_{\chi u,t}+\chi_{u,t}r_{\chi t,s},
\vphantom{\Big]}
\label{bvflow1/1}
\\
&r_{\chi s,t}=-\chi_{s,t}r_{\chi t,s},
\vphantom{\Big]}
\label{bvflow2/1}
\\
&r_{\chi s,s}=0.
\vphantom{\Big]}
\label{bvflow5}
\end{align}
\end{defi}

\noindent 
\ceqref{bvflow2}, \ceqref{bvflow3} are implied by \ceqref{bvflow1}
and are shown here for their usefulness. Similarly, \ceqref{bvflow2/1}, \ceqref{bvflow5}
follow from \ceqref{bvflow1/1}. \pagebreak  \ceqref{bvflow1}--\ceqref{bvflow3} and 
\ceqref{bvflow1/1}--\ceqref{bvflow5} are compatible with relations \ceqref{bvcanon5}--\ceqref{bvcanon7},
as is easily verified.

$r_{\chi t,s}$ is the logarithmic Jacobian family of the BV flow $\chi_{t,s}$. 
$r_{\chi t,s}$ is determined  only up to addition an $\mathbb{R}1_X$ valued family over $T^2$
obeying  conditions \ceqref{bvflow1/1}--\ceqref{bvflow5}. 

\begin{defi} \label{def:cenfam}
A degree $p$ central logarithmic family over $T^2$ is a family of degree $p$ elements of $X$ over $T^2$
of the form $\rho_{t,s}1_X$, where $\rho$ is a degree $p$ 
function on $T^2$ satisfying the relation \hphantom{xxxxxxxxxxxxxx}
\begin{equation}
\rho_{u,s}=\rho_{u,t}+\rho_{t,s}.
\label{bvflow/00}
\end{equation}
\end{defi}

\noindent
We denote by $\zz_p(T^2,X)\subset \iMap(T^2,X)$ the vector space of all such families. The logarithmic Jacobian family 
$r_{\chi t,s}$ is then defined mod $\zz_0(T^2,X)$. 
In what follows, we assume that a choice of $r_{\chi t,s}$ has been made once and for all. 

$\chi_{t,s}$ and $r_{\chi t,s}$ fit by \ceqref{bvcanon2} into the equation 
\begin{equation}
\varDelta_t\chi_{t,s}-\chi_{t,s}\varDelta_s+\ad_tr_{\chi t,s}\,\chi_{t,s}=0. 
\label{bvflow4}
\end{equation}
Further, by \ceqref{bvcanon3},  $r_{\chi t,s}$ satisfies  the equation 
\begin{equation}
\varDelta_tr_{\chi t,s}+\frac{1}{2}(r_{\chi t,s},r_{\chi t,s})_t=0. 
\label{bvflow6}
\end{equation}

A BV flow $\chi$ along $T$ can be regarded as a groupoid $\mathsans{G}_\chi$ such that $\Obj_{\mathsans{G}_\chi}=T$
and $\Hom_{\mathsans{G}_\chi}(s,t)=\{\chi_{t,s}\}$ for any pair $s,t\in \Obj_{\mathsans{G}_\chi}$. 
$\mathsans{G}_\chi$ is trivial in the sense that 
\begin{equation}
\chi_{t,s}=\chi_{t,o}\chi_{s,o}{}^{-1}, 
\label{zbvflow1}
\end{equation}
where $o\in T$ is a fixed chosen origin of $T$. 

The groupoid $\mathsans{G}_\chi$ acts on the linear fibered space 
$\mathcal{X}_{\mathsans{G}_\chi}=\coprod_{t\in \Obj_{\mathsans{G}_\chi}}X_t$
in the ob\-vious manner. By  \ceqref{bvflow1/1}, the logarithmic Jacobian 
$r_\chi:\Hom_{\mathsans{G}_\chi}\rightarrow \mathcal{X}_{\mathsans{G}_\chi}$
of $\chi$ defines a $1$--cocycle of $\mathsans{G}_\chi$ valued in $\mathcal{X}_{\mathsans{G}_\chi}$.
This cocycle is trivial, since by \ceqref{bvflow1/1} and \ceqref{zbvflow1} $r_\chi$ can be expressed as 
\begin{equation}
r_{\chi t,s}=r_{\chi t,o}-\chi_{t,s}r_{\chi s,o}. 
\label{zbvflow2}
\end{equation}

\vspace{-.3cm}
\pagebreak 

As for families of BV structures,  it is possible to probe a BV flow $\chi_{t,s}$  over $T$
using a vector field $D$ of $T$. 

\begin{defi} \label{def:chidrd}
The $D$--infinitesimal generator family over $T$  of the BV flow $\chi_{t,s}$ is the 
family $\chi^D{}_t:X\rightarrow X$ of vector space endomorphisms over $T$ 
\begin{equation}
\chi^D{}_t=D_t\chi_{t,s}\big|_{s=t}. 
\label{bvflow7}
\end{equation}
The  $D$--infinitesimal logarithmic Jacobian family over $T$ of the BV flow is the family 
of elements $r^D{}_{\chi t}\in X$ over $T$ defined by 
\begin{equation}
r^D{}_{\chi t}=D_tr_{\chi t,s}\big|_{s=t}. 
\label{bvflow8}
\end{equation}
\end{defi}

\begin{prop}
$\chi^D{}_t$ is degree $|D|$ derivation of $X$. 
\end{prop}

\begin{proof}
Let $f,g\in X$. 
Acting with $D_t$ on the identity $\chi_{t,s}(fg)-\chi_{t,s}f\chi_{t,s}g=0$, 
setting $s=t$ and then using relations \ceqref{bvflow3} and \ceqref{bvflow7}, 
we obtain 
\begin{align}
0&=\big[D_t\chi_{t,s}(fg)-D_t\chi_{t,s}f\chi_{t,s}g-(-1)^{|D||f|}\chi_{t,s}fD_t\chi_{t,s}g\big]\big|_{s=t}
\vphantom{\Big]}
\label{bvflow9}
\\
&=\chi^D{}_t(fg)-\chi^D{}_tfg-(-1)^{|D||f|}f\chi^D{}_tg.
\vphantom{\Big]}
\nonumber
\end{align}
Hence, $\chi^D{}_t$ is degree $|D|$ derivation of $X$. 
\end{proof}

\begin{prop}
The equation
\begin{equation}
\varDelta^D{}_t-[\chi^D{}_t,\varDelta_t]+\ad_tr^D{}_{\chi t}=0 
\label{bvflow10}
\end{equation}
holds, where the $D$--derived BV Laplacian $\varDelta^D{}_t$ is defined by \ceqref{bvfam0}.
\end{prop}

\noindent 
Eq. \ceqref{bvflow10} is the $D$--evolution equation of the BV flow $\chi_{t,s}$.  

\begin{proof}
Acting with $D_t$ on the identity \ceqref{bvflow4}, setting $s=t$ and then using re\-lations \ceqref{bvflow3},
\ceqref{bvflow5} and \ceqref{bvflow7}, \ceqref{bvflow8}, we obtain 
\begin{align}
0&=\big[D_t\varDelta_t\chi_{t,s}f+(-1)^{|D|}\varDelta_tD_t\chi_{t,s}f-D_t\chi_{t,s}\varDelta_sf
\vphantom{\Big]}
\label{bvflow11}\\
&\hspace{.66cm}+(D_tr_{\chi t,s},\chi_{t,s}f)_t+(-1)^{|D|}(r_{\chi t,s},D_t\chi_{t,s}f)_t
+D_t(r_{\chi u,s},\chi_{u,s}f)_t\big|_{u=t}\big]\big|_{s=t}
\vphantom{\Big]}
\nonumber
\end{align}
\vspace{-1cm}\eject\noindent
\begin{align}
&=\varDelta^D{}_tf+(-1)^{|D|}\varDelta_t\chi^D{}_tf-\chi^D{}_t\varDelta_tf+(r^D{}_{\chi t},f)_t,
\vphantom{\Big]}
\nonumber
\end{align}
where $f\in X$. This shows \ceqref{bvflow10}. 
\end{proof}

\begin{prop}
$r^D{}_{\chi t}$ has degree $D$ and satisfies the equation  
\begin{equation}
\varDelta_tr^D{}_{\chi t}=0. \vphantom{\bigg]}
\label{bvflow12}
\end{equation}
\end{prop}

\noindent
Eq. \ceqref{bvflow12} is the $D$--Jacobian equation of the BV flow $\chi_{t,s}$.   

\begin{proof} 
From the definition \ceqref{bvflow8}, it is clear that $|r^D{}_{\chi t}|=|D|$.  

Acting with $D_t$ on the identity \ceqref{bvflow6}, setting $s=t$ and using relation \ceqref{bvflow5}
and again \ceqref{bvflow8}, we find 
\begin{align}
0&=\big[D_t\varDelta_tr_{\chi t,s}+(-1)^{|D|}\varDelta_tD_tr_{\chi t,s} \hspace{3.5cm}
\vphantom{\Big]}
\label{bvflow13}
\\
&\hspace{3cm}+(D_tr_{\chi t,s},r_{\chi t,s})_t+\frac{1}{2}D_t(r_{\chi u,s},r_{\chi u,s})_t\big|_{u=t}\big]\big|_{s=t}
\vphantom{\Big]}
\nonumber
\\
&=(-1)^{|D|}\varDelta_tr^D{}_{\chi t}, 
\vphantom{\Big]}
\nonumber
\end{align}
showing \ceqref{bvflow12}. 
\end{proof}

Since $\chi_{t,s}$ is canonical, one has 
\begin{equation}
\chi_{t,s}(f,g)_s=(\chi_{t,s}f,\chi_{t,s}g)_t \vphantom{\bigg]}
\label{bvflow14/1}
\end{equation}
for $f,g\in X$, by \ceqref{bvcanon4}.
This entails the following property.

\begin{prop}
For $f,g\in X$, the relation 
\begin{equation}
\chi^D{}_t(f,g)_t=(\chi^D{}_tf,g)_t+(-1)^{|D|(|f|+1)}(f,\chi^D{}_tg)_t+(-1)^{|D||f|}(f,g)^D{}_t
\label{bvflow15/1}
\end{equation}
holds, where $(f-,-)^D{}_t$ is the $D$--derived BV bracket defined in \ceqref{bvfam6}.
\end{prop}

\begin{proof}
Acting with $D_t$ on the identity \ceqref{bvflow14/1}, setting $s=t$ and using relations \ceqref{bvflow3},
and \ceqref{bvflow7}, we obtain 
\begin{align}
0&=\big[D_t\chi_{t,s}(f,g)_s-(D_t\chi_{t,s}f,\chi_{t,s}g)_t
\vphantom{\Big]}
\label{bvflow16/1}
\\
&\hspace{2cm}-(-1)^{|D|(|f|+1)}(\chi_{t,s}f,D_t\chi_{t,s}g)_t
-D_t(\chi_{u,s}f,\chi_{u,s}g)_t\big|_{u=t}\big]\big|_{s=t}
\vphantom{\Big]}
\nonumber
\end{align}
\vspace{-1cm}\eject\noindent
\begin{align}
&=\chi^D{}_t(f,g)_t-(\chi^D{}_tf,g)_t-(-1)^{|D|(|f|+1)}(f,\chi^D{}_tg)_t-(-1)^{|D||f|}(f,g)^D{}_t,
\vphantom{\Big]}
\nonumber
\end{align}
showing \ceqref{bvflow15/1}. 
\end{proof}

In the last part of this subsection, we show how a certain set of natural data on a graded algebra 
can be used to assemble a BV Laplacian family and flow.  

Below $X$ is a graded commutative algebra, $T$ is a graded manifold and $X$ is equipped with a
formal differentiation structure over $T$. 

\begin{prop} \label{prop:bvdbuild}
Assume the following data are given.
\begin{enumerate}


\item A fixed origin $o\in T$. 

\item A non singular BV Laplacian $\varDelta_o$ on $X$. 

\item A family $\chi_{t,s}:X\rightarrow X$ of algebra automorphisms of $X$
over $T^2$ satisfying relation \ceqref{bvflow1}.
 
\item A family $r_{\chi t,s}\in X$ of degree $0$ elements of $X$ over $T^2$ satisfying 
\ceqref{bvflow1/1} and \ceqref{bvflow6} with $t=o$. 

\end{enumerate}
Set for $t\in T$, \hphantom{xxxxxxxxxxxxxxxxxxx}
\begin{equation}
\varDelta_{\chi t}=\chi_{t,o}(\varDelta_o+\ad_o r_{\chi o,t})\chi_{t,o}{}^{-1}. \vphantom{\bigg]}
\label{xbvflow13}
\end{equation}
Then, $\varDelta_{\chi t}$ is a non singular BV Laplacian family over $T$, $\chi_{t,s}$ is a BV flow
relative to it and $r_{\chi t,s}$ is the associated logarithmic Jacobian family. 
\end{prop}

\noindent
The proof 
is a bit involved. It follows from the following five lemmas. 

\begin{lemma} \label{lemma:1}
Let $\varDelta_0$ be a BV Laplacian on $X$ and 
let $a\in X$ with $|a|=0$ satisfying 
\begin{equation}
\varDelta_0a+\frac{1}{2}(a,a)_0=0. \vphantom{\bigg]}
\label{xbvflow1}
\end{equation}
Then, the operator \hphantom{xxxxxxxxxxxxxxxxx}
\begin{equation}
\varDelta_a=\varDelta_0+\ad_0 a \vphantom{\bigg]}
\label{xbvflow2}
\end{equation}
is a BV Laplacian on $X$. 
\end{lemma}

\begin{proof} 
We have to demonstrate that $\varDelta_a$ satisfies conditions \ceqref{bvalg1}--\ceqref{bvalg4} 
of def. \cref{def:bvalg}.

For $f\in X$, we have $|\varDelta_0f|=|\ad_0f|=|f|+1$. Hence, $|\varDelta_af|=|f|+1$.
Thus, $\varDelta_a$ satisfies condition \ceqref{bvalg1}. 

$\varDelta_0$ satisfies condition \ceqref{bvalg2} by assumption. 
By repeated application of property \ceqref{bvalg9}, one can verify that the relation  
\begin{align}
\ad_0a(fgh)
&=-\ad_0afgh-(-1)^{|f|}f\ad_0agh-(-1)^{|f|+|g|}fg\ad_0ah
\vphantom{\Big]}
\label{xbvflow5}
\\
&+\ad_0a(fg)h+(-1)^{(|f|+1)|g|}g\ad_0a(fh)+(-1)^{|f|}f\ad_0a(gh)
\vphantom{\Big]}
\nonumber
\end{align}  
holds for $f,g,h\in X$. Hence,  by \ceqref{xbvflow2}, $\varDelta_a$ also obeys 
condition \ceqref{bvalg2}. 

By \ceqref{xbvflow2}, we have 
\begin{equation}
\varDelta_a{}^2=\varDelta_0{}^2+\varDelta_0\ad_0a+\ad_0a\varDelta_0+(\ad_0a)^2. 
\label{xbvflow6}
\end{equation}
Identity \ceqref{bvalg11} with $f=a$ entails that 
\begin{equation}
\varDelta_0\ad_0a+\ad_0a\varDelta_0=\ad_0(\varDelta_0a). 
\label{xbvflow7}
\end{equation}
Further identity \ceqref{bvalg8} with $f=g=a$ implies that 
\begin{equation}
(\ad_0a)^2=\frac{1}{2}\ad_0(a,a)_0. 
\label{xbvflow8}
\end{equation}
It follows that \hphantom{xxxxxxxxxxxx}
\begin{equation}
\varDelta_a{}^2=\ad_0\left(\varDelta_0a+\frac{1}{2}(a,a)_0\right)=0
\label{xbvflow9}
\end{equation}
by \ceqref{xbvflow1}. So, $\varDelta_a$ satisfies also condition \ceqref{bvalg3}. 

$\varDelta_0$ obeys condition \ceqref{bvalg4} by assumption. 
By virtue of \ceqref{bvalg10}, we have $\ad_01_X=0$.  
Hence, by \ceqref{xbvflow2}, $\varDelta_a$ also satisfies condition \ceqref{bvalg4}. 
\end{proof}

\begin{lemma} \label{lemma:2}
If $\alpha:X\rightarrow X$ is a graded algebra automorphism, then the operator 
\begin{equation}
\varDelta_{\alpha a}=\alpha\varDelta_a\alpha^{-1} 
\label{xbvflow11}
\end{equation}
with $\varDelta_a$ given by \ceqref{xbvflow2} is a BV Laplacian on $X$. 
\end{lemma}

\begin{proof}
Conditions \ceqref{bvalg1}--\ceqref{bvalg4} of def. \cref{def:bvalg}, required for an 
operator to be a BV Laplacian, are manifestly invariant under conjugation by a graded algebra 
automorphism. Thus, since $\varDelta_a$ is a BV Laplacian by lemma \cref{lemma:1} and $\alpha$
is an automorphism, $\varDelta_{\alpha a}$ is a BV Laplacian as well.
\end{proof}

\begin{lemma} \label{lemma:3}
Let $(-,-)_{\alpha a}$ be the BV bracket associated with the BV Laplacian $\varDelta_{\alpha a}$
by \ceqref{bvalg5}. Then, for $u\in X$ with $|u|=0$, one has
\begin{equation}
\ad_{\alpha a}u=\alpha\ad_0(\alpha^{-1}u)\alpha^{-1}. 
\label{xbvflow14}
\end{equation}
\end{lemma}

\begin{proof}
The expression \ceqref{bvalg5}  of the BV bracket associated with a BV Laplacian 
is manifestly invariant under conjugation by graded algebra automorphism. So, 
\begin{equation}
(f,g)_{\alpha a}=\alpha(\alpha^{-1}f,\alpha^{-1}g)_a 
\label{xbvflow15}
\end{equation}
for $f,g\in X$, where $(-,-)_a$ is the BV bracket associated with the BV Laplacian $\varDelta_a$.
$(-,-)_a$ is also given by \ceqref{bvalg5}. As $\ad_0a$ is a degree $1$ derivation of $X$, 
\begin{equation}
\ad_0(fg)-\ad_0fg-(-1)^{|f|}f\ad_0g=0. 
\label{xbvflow1s}
\end{equation}
From \ceqref{bvalg5} and \ceqref{xbvflow2}, it follows then that $\ad_0a$ gives a vanishing contribution to $(-,-)_a$ 
so that $(-,-)_a=(-,-)_0$. Hence, \ceqref{xbvflow15} gets updated to 
\begin{equation}
(f,g)_{\alpha a}=\alpha(\alpha^{-1}f,\alpha^{-1}g)_0. 
\label{xbvflow16}
\end{equation}
From this relation, \ceqref{xbvflow14} is obvious. 
\end{proof} 

\begin{lemma} \label{lemma:5}
$\varDelta_{\alpha a}$ is non singular. 
\end{lemma}

\begin{proof}
This follows readily from \ceqref{xbvflow16} and the non singulararity of $\varDelta_0$. 
\end{proof}

\noindent 
Denote by $X_{\alpha a}$ the BV algebra consisting in the graded algebra $X$ 
equipped with the BV Laplacian $\varDelta_{\alpha a}$.

\begin{lemma} \label{lemma:4}
Let $a_1,a_2\in X$ with $|a_1|=|a_2|=0$, $\alpha_1,\alpha_2:X\rightarrow X$ be graded algebra automorphisms
and $\beta=\alpha_2\alpha_1{}^{-1}$. Then, $\beta:X_{\alpha_1a_1}\rightarrow X_{\alpha_2a_2}$ is canonical with 
associated logarithmic Jacobian \pagebreak 
\begin{equation}
r_\beta=-\alpha_2(a_2-a_1)\qquad \text{\rm mod $\mathbb{R}1_X$}. 
\label{xbvflow17}
\end{equation}
\end{lemma}

\begin{proof}
Using \ceqref{xbvflow2},  \ceqref{xbvflow11} and \ceqref{xbvflow14}, we find 
\begin{align}
\varDelta_{\alpha_2a_2}\beta-\beta\varDelta_{\alpha_1a_1}
&=\alpha_2(\varDelta_{a_2}-\varDelta_{a_1})\alpha_1{}^{-1}
\vphantom{\Big]}
\label{}
\\
&=\alpha_2\ad_0(a_2-a_1)\alpha_1{}^{-1}
\vphantom{\Big]}
\nonumber
\\
&=\ad_{\alpha_2a_2}(\alpha_2(a_2-a_1))\beta. 
\vphantom{\Big]}
\nonumber
\end{align}
Now, $\varDelta_{\alpha a}$ is non singular by lemma  \cref{lemma:5}. The statement follows. 
\end{proof}

\noindent
{\it Proof of prop. \cref{prop:bvdbuild}}.
$\varDelta_{\chi t}$ is of the form \ceqref{xbvflow11} with $\alpha=\chi_{t,o}$ and $a=r_{\chi o,t}$.
By lemma \cref{lemma:2}, so, $\varDelta_{\chi t}$ is a BV Laplacian. By lemma \cref{lemma:5},
$\varDelta_{\chi t}$ is non singular. The fact that $\chi_{t,o}$ and $r_{\chi o,t}$ 
are families of endomorphisms of $X$ of the given formal differentiation structure of $X$ over $T$ 
ensures that $\varDelta_{\chi t}$ is a family of the same kind 
as well (cf. def. \cref{def:formdif}). In this way,  $\varDelta_{\chi t}$ constitutes a non singular 
BV Laplacian family. 

Let $\alpha_1=\chi_{s,o}$, $\alpha_2=\chi_{t,o}$, $\beta=\chi_{t,o}\chi_{s,o}{}^{-1}=\chi_{t,s}$.
By lemma \cref{lemma:4}, then, $\chi_{t,s}:X_s\rightarrow X_t$ is canonical.
Furthermore, its logarithmic Jacobian, computed according to \ceqref{xbvflow17}, is given by 
\begin{align}
r_{\chi_{t,s}}&=-\chi_{t,o}(r_{\chi o,t}-r_{\chi o,s}) \qquad \text{mod $\mathbb{R}1_X$}
\vphantom{\Big]}
\label{}
\\
&=-\chi_{t,o}\chi_{o,s}r_{\chi s,t}
\vphantom{\Big]}
\nonumber
\\
&=r_{\chi t,s},
\vphantom{\Big]}
\nonumber
\end{align}
where \ceqref{bvflow1/1} has been used. Hence the $\mathbb{R}1_X$ indeterminacy of $r_{\chi_{t,s}}$ 
can be fixed so that $r_{\chi_{t,s}}$ is a family of degree $0$ elements 
of $X$ over $T^2$, since  
$r_{\chi t,s}$ is one such family by assumption. It follows so from def. \cref{def:bvflow} that $\chi_{t,s}$ is a BV flow 
relative to the Laplacian family $\varDelta_{\chi t}$. 
\hfill {\small $\Box$}

Given a non singular BV Laplacian family $\varDelta_t$ of $X$ over $T$  
and a BV flow $\chi_{t,s}$ relative to $\varDelta_t$ with logarithmic Jacobian family $r_{\chi t,s}$, we 
can use $\varDelta_o$, $\chi_{t,s}$ and $r_{\chi t,s}$ to construct the  BV Laplacian family
$\varDelta_{\chi t}$  having $\chi_{t,s}$ as a BV flow using \ceqref{xbvflow13}
as stated in prop. \cref{prop:bvdbuild}. One should keep in mind, however, that in general
$\varDelta_{\chi t}$ needs not equal $\varDelta_t$.


\vfil\eject

\subsection{\textcolor{blue}{\sffamily BV flow stabilizers}}\label{subsec:bvflowstab}

In this subsection, we introduce the notion of BV flow 
stabilizer which is expected be relevant in the study of the symmetry of BV flows. 

Let $X$ and $T$ be as in subsect. \cref{subsec:bvflow} above and let $\varDelta_t$ be 
a non singular BV Laplacian family over $T$ and $\chi_{t,s}$ a BV flow along $T$ relative to it
(cf. defs. \cref{def:tfamD} and \cref{def:bvflow}). 

It is possible to generate new BV flows from $\chi_{t,s}$ by conjugation.

\begin{defi}
A canonical map family over $T$ (relative to the BV Laplacian family $\varDelta_t$) 
is a family $\gamma_t:X\rightarrow X$ of algebra automorphisms of $X$ over $T$ 
such that $\gamma_t:X_t\rightarrow X_t$ is a canonical map for every $t$ (cf. def. \cref{def:canmap}).  
It is further required that the $\mathbb{R}1_X$ indeterminacy of the logarithmic Jacobians $r_{\gamma t}$ 
of the $\gamma_t$ 
can be fixed for each $t$ in such a way that the collection $r_{\gamma t}$ 
is a family of degree $0$ elements of $X$ over $T$.
\end{defi}

\noindent 
$r_{\gamma t}$ is the logarithmic Jacobian family of the canonical map family $\gamma_t$. 
$r_{\gamma t}$ is unique only up to addition of an $\mathbb{R}1_X$ valued family over $T$.

\begin{defi} \label{def:cenfam1}
A degree $p$ central logarithmic family over $T$ is a family of degree $p$ elements of $X$ over $T$
of the form $\rho_t1_X$, where $\rho$ is a degree $p$ 
function on $T$.
\end{defi}

\noindent
We denote by $\zz_p(T,X)\subset \iMap(T,X)$ the vector space of all such families. The logarithmic Jacobian family 
$r_{\gamma t}$ is then defined mod $\zz_0(T,X)$. In what follows, we assume that a choice of $r_{\gamma t}$ 
has been made once and for all. 

Let $\gamma_t$ be a canonical map family. 

\begin{defi} \label{def:gammaconj}
The $\gamma$--conjugate of the BV flow $\chi_{t,s}$ is the family 
${}^\gamma\chi_{t,s}:X\rightarrow X$ of algebra automorphisms of $X$ over $T^2$ defined by 
\begin{equation}
{}^\gamma\chi_{t,s}=\gamma_t\chi_{t,s}\gamma_s{}^{-1}. \vphantom{\bigg]}
\label{bvflow14}
\end{equation}
\end{defi}

\begin{prop}
The $\gamma$--conjugate ${}^\gamma\chi_{t,s}$ of $\chi_{t,s}$ is a BV flow. 
\end{prop}

\begin{proof}
${}^\gamma\chi_{t,s}$ is \pagebreak a family of algebra automorphisms of $X$ over $T^2$, as by \ceqref{bvflow14} 
it is a composition of the families ${}^\gamma\chi_{t,s}$ and $\gamma_t$. ${}^\gamma\chi_{t,s}$ enjoys the 
properties \ceqref{bvflow1}-- \ceqref{bvflow3}, since those relations are manifestly 
invariant under conjugation. Further, ${}^\gamma\chi_{t,s}:X_s\rightarrow X_t$ is canonical, as it is a 
composition of canonical maps. Next, by \ceqref{bvflow14} again, and the identities \ceqref{bvcanon5}, 
\ceqref{bvcanon7},  it follows that the logarithmic Jacobian of ${}^\gamma\chi_{t,s}$ is  
\begin{equation}
r_{{}^\gamma\chi_{t,s}}=\gamma_tr_{\chi t,s}+r_{\gamma t}-{}^\gamma\chi_{t,s}r_{\gamma s}
\qquad \text{mod $\mathbb{R}1_X$}. 
\label{bvflow14/2}
\end{equation}
The $\mathbb{R}1_X$ indeterminacy of $r_{{}^\gamma\chi_{t,s}}$ can be fixed in such a way that 
$r_{{}^\gamma\chi_{t,s}}$ is a family of degree $0$ elements of $X$ over $T^2$, 
as the right hand side 
of \ceqref{bvflow14/2} is built of the families $\chi_{t,s}$, $r_{\chi t,s}$ and $\gamma_t$,
$r_{\gamma t}$.  The statement follows. 
\end{proof}

\begin{prop} \label{prop:rgchits}
The logarithmic Jacobian family $r_{{}^\gamma\chi t,s}$ 
of the $\gamma$--conjugate flow ${}^\gamma\chi_{t,s}$ can be chosen to be given by 
\begin{equation}
r_{{}^\gamma\chi t,s}=\gamma_tr_{\chi t,s}+r_{\gamma t}-{}^\gamma\chi_{t,s}r_{\gamma s}.
\label{bvflow18}
\end{equation}
\end{prop}

\begin{proof}
This follows readily from \ceqref{bvflow14/2}. 
\end{proof}

\noindent
Recall that any other picking differs from the above one by an element of  
the degree $0$ central logarithmic family space 
$\zz_0(T^2,X)$ (cf. subsect. \cref{subsec:bvflow}, def. \cref{def:cenfam}). 
Below, we adhere to the choice \ceqref{bvflow18}. 

The stabilizers of a BV flow are special conjugations leaving the flow unchanged.
They so encode the flow's symmetry. 


\begin{defi} \label{def:stab}
The canonical map family $\gamma_t$ is a stabilizer of the BV flow $\chi_{t,s}$ if
\begin{equation}
{}^\gamma\chi_{t,s}=\chi_{t,s}. \vphantom{\bigg]}
\label{bvflow20}
\end{equation}
\end{defi}

Combining def. \cref{def:gammaconj} and prop. \cref{prop:rgchits}, we obtain 
the following result. 

\begin{prop} \label{prop:stabrgchits}
Let the canonical map family $\gamma_t$ be a stabilizer of the BV flow $\chi_{t,s}$. Then, $\chi_{t,s}$ 
satisfies the identity \hphantom{xxxxxxxxxxxxxxxx}
\begin{equation}
\chi_{t,s}=\gamma_t\chi_{t,s}\gamma_s{}^{-1}, 
\label{bvflow10/1}
\end{equation}
\vspace{-1cm}\eject\noindent
while the logarithmic Jacobian $r_{\chi t,s}$ family of $\chi_{t,s}$ obeys 
\begin{equation}
r_{\chi t,s}=\gamma_tr_{\chi t,s}+r_{\gamma t}-\chi_{t,s}r_{\gamma s}
\qquad \text{\rm mod $\zz_0(T^2,X)$}. \vphantom{\bigg]}
\label{bvflow10/2}
\end{equation}
\end{prop}

\noindent
The identity \ceqref{bvflow10/2} holds only mod $\zz_0(T^2,X)$ because the choice of 
$r_{\chi t,s}$ and $r_{{}^\gamma\chi t,s}$ fixes their respective $\zz_0(T^2,X)$ 
indeterminacy in a generally independent manner. 

Let $D\in\Vect(M)$ be a vector field. 

\begin{prop} \label{prop:conjchi}
Under $\gamma$--conjugation, the $D$--infinitesimal generator 
$\chi^D{}_t$ of the BV flow $\chi_{t,s}$ (cf. 
eq. \ceqref{bvflow7}) transforms as
\begin{equation}
{}^\gamma\chi^D{}_t=D_t\gamma_t\gamma_t{}^{-1}+\gamma_t\chi^D{}_t\gamma_t{}^{-1}. \vphantom{\bigg]}
\label{bvflow15}
\end{equation}
Moreover, the $D$--infinitesimal logarithmic Jacobian $r^D{}_{{}^\gamma\chi t}$ 
(cf. 
eq. \ceqref{bvflow8}) gets  
\begin{equation}
r^D{}_{{}^\gamma\chi t}=\gamma_tr^D{}_{\chi t}+D_tr_{\gamma t}-{}^\gamma\chi^D{}_tr_{\gamma t}. \vphantom{\bigg]}
\label{bvflow16}
\end{equation}
\end{prop}
 
\begin{proof}
By relation \ceqref{bvflow7} defining $\chi^D{}_t$, 
acting with $D_t$ on both sides of the identity \ceqref{bvflow14}, setting $s=t$ and 
using \ceqref{bvflow3}, we find 
\begin{align}
{}^\gamma\chi^D{}_t&=\big[D_t\gamma_t\chi_{t,s}\gamma_s{}^{-1}
+\gamma_tD_t\chi_{t,s}\gamma_s{}^{-1}\big]\big|_{s=t}
\vphantom{\Big]}
\label{bvflow17}
\\
&=D_t\gamma_t\gamma_t{}^{-1}+\gamma_t\chi^D{}_t\gamma_t{}^{-1}.\hspace{1.5cm}
\vphantom{\Big]}
\nonumber
\end{align}
This shows \ceqref{bvflow15}. 

By relation \ceqref{bvflow8} defining $r^D{}_{{}^\gamma\chi t}$, making 
$D_t$ act on both sides of the identi\-ty \ceqref{bvflow18}, setting $s=t$ and 
using \ceqref{bvflow5} and \ceqref{bvflow7}, we obtain
\begin{align}
r^D{}_{{}^\gamma\chi t}&
=\big[D_t\gamma_tr_{\chi t,s}+\gamma_tD_tr_{\chi t,s}+D_tr_{\gamma t}
-D_t{}^\gamma\chi_{t,s}r_{\gamma s}\big]\big|_{s=t}
\vphantom{\Big]}
\label{bvflow19}
\\
&=\gamma_tr^D{}_{\chi t}+D_tr_{\gamma t}-{}^\gamma\chi^D{}_tr_{\gamma t}.
\vphantom{\Big]}
\nonumber
\end{align}
This shows \ceqref{bvflow16}. 
\end{proof}

From prop. \cref{prop:conjchi}, we obtain immediately the following result. 

\begin{prop} \label{prop:stabchi}
If the canonical map family $\gamma_t$ is a stabilizer of the BV flow $\chi_{t,s}$,
the $D$--infinitesi\-mal generator $\chi^D{}_t$ of $\chi_{t,s}$ satisfies 
\begin{equation}
\chi^D{}_t=D_t\gamma_t\gamma_t{}^{-1}+\gamma_t\chi^D{}_t\gamma_t{}^{-1},
\label{bvflow21}
\end{equation}
while the $D$--infinitesimal logarithmic Jacobian $r^D{}_{\chi t}$ obeys 
\begin{equation}
r^D{}_{\chi t}=\gamma_tr^D{}_{\chi t}+D_tr_{\gamma t}-\chi^D{}_tr_{\gamma t}
\qquad \text{\rm mod $\zz_{|D|}(T,X)$}.
\label{bvflow22}
\end{equation}
\end{prop}
 
\begin{proof}
Only \ceqref{bvflow22} needs to be commented.
\ceqref{bvflow22} holds only mod $\zz_{|D|}(T,X)$ because \ceqref{bvflow10/2} holds mod 
$\zz_0(T^2,X)$ and the action of $D_t$ on a degree $0$ central logarithmic family 
$\rho_{t,s}1_X$ over $T^2$ at $s=t$ yields the degree $D$ central logarithmic family 
$D_t\rho_{t,s}|_{s=t}1_X$ over $T$.  
\end{proof}

The infinitesimal versions of 
conjugation and stabilizer turn out to be useful.

\begin{defi}
An infinitesimal canonical map family over $T$  (relative to the BV Laplacian family $\varDelta_t$) 
is a family $\epsilon_t:X\rightarrow X$, $t\in T$, of degree $0$ derivations of $X$ over $T$
such that $\epsilon_t:X_t\rightarrow X_t$ is an infinitesimal canonical map for every $t$
(cf. def. \cref{def:infbvcan}). It is further required that the $\mathbb{R}1_X$ 
indeterminacy of the logarithmic Jacobians $e_{\epsilon t}$ 
of the $\epsilon_t$ 
can be fixed in such a way that the collection $e_{\epsilon t}$, $t\in T$, 
is a family of degree $0$ elements of $X$ over $T$.
\end{defi}

\noindent 
$e_{\epsilon t}$ is the logarithmic Jacobian family of the infinitesimal canonical map family $\epsilon_t$
and, as in the finite case, is defined only mod $\zz_0(T,X)$. Again, we assume that a choice of $e_{\epsilon t}$
has been made once and for all below. 

Let $\epsilon_t$ be an infinitesimal canonical map family.

\begin{defi} \label{def:infvarchi}
The infinitesimal $\epsilon$--conjugate of the BV flow $\chi_{t,s}$ is the 
family  $\delta_\epsilon\chi_{t,s}:X\rightarrow X$ of endomorphisms of $X$ over $T^2$ 
defined by 
\begin{equation}
\delta_\epsilon\chi_{t,s}=\epsilon_t\chi_{t,s}-\chi_{t,s}\epsilon_s.
\label{bvflow23}
\end{equation}
\end{defi}

\begin{prop}
$\delta_\epsilon\chi_{t,s}$
is a degree $0$ derivation of $X$ over $\chi_{t,s}$, that is  
\begin{equation}
\delta_\epsilon\chi_{t,s}(fg)
=\delta_\epsilon\chi_{t,s}f\chi_{t,s}g+\chi_{t,s}f\delta_\epsilon\chi_{t,s}g
\label{bvflow24}
\end{equation}
for all $f,g\in X$.
\end{prop}

\begin{proof} Let $f,g\in X$. Then, using that $\chi_{t,s}$ is an algebra endomorphism and $\epsilon_t$ 
is a degree $0$ derivation of $X$, we have \pagebreak 
\begin{align}
\delta_\epsilon\chi_{t,s}(fg)&=\epsilon_t(\chi_{t,s}f\chi_{t,s}g)-\chi_{t,s}(\epsilon_sfg+f\epsilon_sg)
\vphantom{\Big]}
\label{bvflow25}
\\
&=\epsilon_t\chi_{t,s}f\chi_{t,s}g+\chi_{t,s}f\epsilon_t\chi_{t,s}g
-\chi_{t,s}\epsilon_sf\chi_{t,s}g-\chi_{t,s}f\chi_{t,s}\epsilon_sg
\vphantom{\Big]}
\nonumber
\\
&=\delta_\epsilon\chi_{t,s}f\chi_{t,s}g+\chi_{t,s}f\delta_\epsilon\chi_{t,s}g,
\vphantom{\Big]}
\nonumber
\end{align}
as claimed.
\end{proof}

\noindent
$\delta_\epsilon$ behaves as a variation operation as appears from 
comparing \ceqref{bvflow4} and the following proposition. 

\begin{prop} \label{prop:dechi}
The relation 
\begin{equation}
\varDelta_t\delta_\epsilon\chi_{t,s}-\delta_\epsilon\chi_{t,s}\varDelta_s
+\ad_tr_{\chi t,s}\delta_\epsilon\chi_{t,s}+\ad_t\delta_\epsilon r_{\chi t,s}\chi_{t,s}=0
\label{bvflow26}
\end{equation}
holds, where $\delta_\epsilon r_{\chi t,s}$ is the family of degree $0$ elements of $X$ over $T^2$  
\begin{equation}
\delta_\epsilon r_{\chi t,s}=\epsilon_tr_{\chi t,s}+e_{\epsilon t}-\chi_{t,s}e_{\epsilon s}.
\label{bvflow27}
\end{equation}
\end{prop}

\begin{proof}
Since $\chi_{t,s}$ is canonical, \ceqref{bvflow4} holds. Furthermore, as 
$\epsilon_t$ is infinitesimal canonical, we have \hphantom{xxxxxxxxxxxxxxxxx}
\begin{equation}
\varDelta_t\epsilon_t-\epsilon_t\varDelta_t+\ad_te_{\epsilon t}=0
\label{bvflow28}
\end{equation}
by \ceqref{bvcanon9}. Then, by \ceqref{bvflow24} and \ceqref{bvcanon4}, \ceqref{bvcanon11}, 
\begin{align}
\varDelta_t\delta_\epsilon\chi_{t,s}f
&=\varDelta_t\epsilon_t\chi_{t,s}f-\varDelta_t\chi_{t,s}\epsilon_sf
\vphantom{\Big]}
\label{bvflow29}
\\
&=\epsilon_t\varDelta_t\chi_{t,s}f-(e_{\epsilon t},\chi_{t,s}f)_t
-\chi_{t,s}\varDelta_s\epsilon_sf+(r_{\chi t,s},\chi_{t,s}\epsilon_sf)_t
\vphantom{\Big]}
\nonumber
\\
&=\epsilon_t\chi_{t,s}\varDelta_sf-\epsilon_t(r_{\chi t,s},\chi_{t,s}f)_t-(e_{\epsilon t},\chi_{t,s}f)_t
\vphantom{\Big]}
\nonumber
\\
&\hspace{1.5cm}-\chi_{t,s}\epsilon_s\varDelta_sf+\chi_{t,s}(e_{\epsilon s},f)_s+(r_{\chi t,s},\chi_{t,s}\epsilon_sf)_t
\vphantom{\Big]}
\nonumber
\\
&=\delta_\epsilon\chi_{t,s}\varDelta_sf-(r_{\chi t,s},\delta_\epsilon\chi_{t,s}f)_t
-(\delta_\epsilon r_{\chi t,s},\chi_{t,s}f)_t
\vphantom{\Big]}
\nonumber
\end{align}
for $f\in X$. This shows the proposition.  
\end{proof}

\noindent
We note that \ceqref{bvflow27} is compatible with \ceqref{bvflow1/1} in the sense that the  identity 
\begin{equation}
\delta_\epsilon r_{\chi u,s}=\delta_\epsilon r_{\chi u,t}+\chi_{u,t}\delta_\epsilon r_{\chi t,s}+\delta_\epsilon\chi_{u,t}r_{\chi t,s}
\label{bvflow34/0}
\end{equation}
holds, as is straightforward to verify. We also observe that \ceqref{bvflow27} may be altered as usual 
by adding to the right hand side an arbitrary degree $0$ central logarithmic family of $\zz_0(T^2,X)$. 
The modification preserves \ceqref{bvflow34/0}. In what follows, we conform to the choice \ceqref{bvflow27}.  

The infinitesimal stabilizers of a BV flow encode its infinitesimal symmetry. 

\begin{defi} \label{def:infstab}
The infinitesimal canonical map family $\epsilon_t$ is an infinitesimal stabilizer 
of the BV flow $\chi_{t,s}$ if \hphantom{xxxxxxxxxxxxxxxxxx}
\begin{equation}
\delta_\epsilon\chi_{t,s}=0. 
\label{bvflow34}
\end{equation}
\end{defi}

From def. \cref{def:infvarchi} and prop. \cref{prop:dechi}, we obtain immediately the following result.

\begin{prop} \label{prop:infstabrgchits}
Let the infinitesimal canonical map family $\epsilon_t$ be an infinitesimal  stabilizer of the BV flow $\chi_{t,s}$. 
Then, the BV flow $\chi_{t,s}$ satisfies the identity 
\begin{equation}
\epsilon_t\chi_{t,s}-\chi_{t,s}\epsilon_s=0, 
\label{bvflow34/1}
\end{equation}
while the logarithmic Jacobian $r_{\chi t,s}$ of $\chi_{t,s}$ obeys 
\begin{equation}
\epsilon_tr_{\chi t,s}+e_{\epsilon t}-\chi_{t,s}e_{\epsilon s}=0 \qquad \text{\rm mod $\zz_0(T^2,X)$}.
\label{bvflow34/2}
\end{equation}
\end{prop}

\noindent
\ceqref{bvflow34/2} holds mod $\zz_0(T^2,X)$ for reasons analogous to those why 
\ceqref{bvflow10/2} does

The variation operator $\delta$ commutes with the action of vector fields of $T$, as we show next. 

Let $D\in\Vect(M)$ be a vector field. 

\begin{prop} \label{prop:infdchidrd}
Under infinitesimal $\epsilon$--conjugation, the variation of the $D$--in\-finitesimal generator 
$\chi^D{}_t$ of the BV flow $\chi_{t,s}$ is given by  
\begin{equation}
\delta_\epsilon\chi^D{}_t=D_t\epsilon_t-[\chi^D{}_t,\epsilon_t].
\label{bvflow30}
\end{equation}
Moreover, the variation of the $D$--infinitesimal logarithmic Jacobian reads as
\begin{equation}
\delta_\epsilon r^D{}_{\chi t}=\epsilon_tr^D{}_{\chi t}+D_te_{\epsilon t}-\chi^D{}_te_{\epsilon t}.
\label{bvflow31}
\end{equation}
\end{prop}

\begin{proof} 
By the defining relation \ceqref{bvflow7}, acting with $D_t$ on both sides 
of the identity \ceqref{bvflow23}, setting $s=t$ and using \ceqref{bvflow3}, we find 
\begin{align}
\delta_\epsilon\chi^D{}_t
&=D_t\delta_\epsilon\chi_{t,s}\big|_{s=t}
\vphantom{\Big]}
\label{bvflow32}
\\
&=\big[D_t\epsilon_t\chi_{t,s}+\epsilon_tD_t\chi_{t,s}-D_t\chi_{t,s}\epsilon_s\big]\big|_{s=t}
\vphantom{\Big]}
\nonumber
\\
&=D_t\epsilon_t+\epsilon_t\chi^D{}_t-\chi^D{}_t\epsilon_t.
\vphantom{\Big]}
\nonumber
\end{align}
This shows \ceqref{bvflow30}. 

By the defining relation \ceqref{bvflow8}, acting with $D_t$ on both sides of 
the identity \ceqref{bvflow27}, setting $s=t$ and using \ceqref{bvflow5}, we obtain
\begin{align}
\delta_\epsilon r^D{}_{\chi t}
&=D_t\delta_\epsilon r_{\chi t,s}\big|_{s=t}
\vphantom{\Big]}
\label{bvflow33}
\\
&=\big[D_t\epsilon_tr_{\chi t,s}+\epsilon_tD_tr_{\chi t,s}+D_te_{\epsilon t}
-D_t\chi_{t,s}e_{\epsilon s}\big]\big|_{s=t}
\vphantom{\Big]}
\nonumber
\\
&=\epsilon_tr^D{}_{\chi t}+D_te_{\epsilon t}-\chi^D{}_te_{\epsilon t}. \hspace{3cm}
\vphantom{\Big]}
\nonumber
\end{align}
\ceqref{bvflow31} is so proven. 
\end{proof}

From prop. \cref{prop:infdchidrd}, we obtain immediately the following result. 

\begin{prop} \label{prop:infstabchi}
If the infinitesimal canonical map family $\epsilon_t$ is an infinitesimal stabilizer of the BV flow $\chi_{t,s}$, 
the $D$--infinitesimal generator $\chi^D{}_t$ of $\chi_{t,s}$ satisfies 
\begin{equation}
D_t\epsilon_t-[\chi^D{}_t,\epsilon_t]=0, 
\label{bvflow35}
\end{equation}
while the $D$--infinitesimal logarithmic Jacobian correspondingly obeys 
\begin{equation}
\epsilon_tr^D{}_{\chi t}+D_te_{\epsilon t}-\chi^D{}_te_{\epsilon t}=0
\qquad \text{\rm mod $\zz_{|D|}(T,X)$}. 
\label{bvflow36}
\end{equation}
\end{prop}

\noindent
\ceqref{bvflow22} holds only mod $\zz_{|D|}(T,X)$ for the same reasons indicated in 
the proof of relation \ceqref{bvflow22}.

The above analysis applies in particular for a family of infinitesimal canonical maps of the adjoint 
form $\ad_tx_t$, where $x_t$ is a family of degree $-1$ elements $X$ over $T$. 

\begin{prop} \label{prop:adcarchi}
Under infinitesimal $\ad x$--conjugation, \pagebreak the variation of the BV flow $\chi_{t,s}$ is given by
\hphantom{xxxxxxxxxx}
\begin{equation}
\delta_{\ad x}\chi_{t,s}=\ad_t(x_t-\chi_{t,s}x_s)\chi_{t,s}. \vphantom{\bigg]}
\label{bvflow37}
\end{equation}
\end{prop}

\begin{proof}
By \ceqref{bvflow23}, we have 
\begin{align}
\delta_{\ad x}\chi_{t,s}f&=(x_t,\chi_{t,s}f)-\chi_{t,s}(x_s,f)_s
\vphantom{\Big]}
\label{bvflow38}
\\
&=(x_t-\chi_{t,s}x_s,\chi_{t,s}f)_t,
\vphantom{\Big]}
\nonumber
\end{align}
where $f\in X$, showing \ceqref{bvflow37}.
\end{proof}

\begin{prop}
Under infinitesimal $\ad x$--conjugation, the variation of the $D$--infinitesimal generator 
$\chi^D{}_t$ is 
\begin{equation}
\delta_{\ad x}\chi^D{}_t=\ad_t(D_tx_t-\chi^D{}_tx_t), \vphantom{\bigg]}
\label{bvflow39}
\end{equation}
while that of the $D$--infinitesimal logarithmic Jacobian reads as 
\begin{equation}
\delta_{\ad x}r^D{}_{\chi t}=-(-1)^{|D|}\varDelta_t(D_tx_t-\chi^D{}_tx_t). \vphantom{\bigg]}
\label{bvflow40}
\end{equation}
\end{prop}

\begin{proof}
By \ceqref{bvflow30} with $\epsilon_t=\ad_tx_t$, \ceqref{bvfam6} and relation \ceqref{bvflow15/1}, we have  
\begin{align}
\delta_{\ad x}\chi^D{}_tf
&=D_t(x_t,f)_t-\chi^D{}_t(x_t,f)_t+(x_t,\chi^D{}_tf)_t
\vphantom{\Big]}
\label{bvflow41}
\\
&=(D_tx_t,f)_t+D_t(x_u,f)_t\big|_{u=t}-(\chi^D{}_tx_t,f)_t
\vphantom{\Big]}
\nonumber
\\
&\hspace{2cm}-(x_t,\chi^D{}_tf)_t-(-1)^{|D|}(x_t,f)^D{}_t+(x_t,\chi^D{}_tf)_t
\vphantom{\Big]}
\nonumber
\\
&=(D_tx_t-\chi^D{}_tx_t,f)_t
\vphantom{\Big]}
\nonumber
\end{align}
for $f\in X$. This shows \ceqref{bvflow41}. 

From \ceqref{bvcanon14}, it follows that $e_{\ad x t}=-\varDelta_tx_t$. 
By \ceqref{bvflow31} with $\epsilon_t=\ad_tx_t$, the evolution equation \ceqref{bvflow10} 
and the above relation, we have 
\begin{align}
\delta_{\ad x}r^D{}_{\chi t}
&=(x_t,r^D{}_{\chi t})_t-D_t\varDelta_tx_t-(-1)^{|D|}\varDelta_tD_tx_t+\chi^D{}_t\varDelta_tx_t
\vphantom{\Big]}
\label{bvflow42}
\\
&=-(r^D{}_{\chi t},x_t)_t-\varDelta^D{}_tx_t+\chi^D{}_t\varDelta_tx_t-(-1)^{|D|}\varDelta_tD_tx_t
\vphantom{\Big]}
\nonumber
\\
&=-(-1)^{|D|}\varDelta_t(D_tx_t-\chi^D{}_tx_t).
\vphantom{\Big]}
\nonumber
\end{align}
\vspace{-.9cm}\eject\noindent
This shows \ceqref{bvflow40}. 
\end{proof}

Prop. \cref{prop:adcarchi} immediately leads to the following. 

\begin{prop}
If $\ad_tx_t$ is an infinitesimal stabilizer of the BV flow $\chi_{t,s}$, then 
\begin{equation}
x_t-\chi_{t,s}x_s=0 \qquad \text{\rm mod $\zz_{-1}(T,X)$}.
\label{bvflow43}
\end{equation}
\end{prop}

\noindent
It is simple to see that \ceqref{bvflow43} implies that 
\begin{equation}
D_tx_t-\chi^D{}_tx_t=0 \qquad \text{\rm mod $\zz_{|D|-1}(T,X)$}
\label{bvflow44}
\end{equation}
and so the vanishing of the variations \ceqref{bvflow39} and \ceqref{bvflow40}.

\vfil\eject

\section{\textcolor{blue}{\sffamily Batalin--Vilkovisky  renormalization group theory}}\label{sec:bvrenorm}

In this section, we expound an axiomatic formulation of the BV theory
of RG using the algebraic and geometric framework of sect. \cref{sec:bvalg}.
The reader is referred to ref. \ccite{Zucchini:2017irg} 
for physical motivation and a more comprehensive exposition of the field theoretic 
aspect of this topic. 

The basic notions we introduce are those of BV RG flow and BV EA. The BV RGE is derived
in its more general form for a generic scale parameter space. More specific results are 
obtained when the parameter space is the real line $\mathbb{R}$ or its shifted tangent bundle
$T[1]\mathbb{R}$. In the latter case, an RG supersymmetry emerges that shapes the structure of the 
RGE in a form closely related to Polchinski's  \ccite{Polchinski:1983gv}. We consider this one of the  
main results of the present paper. 


\subsection{\textcolor{blue}{\sffamily BV quantum MAs}}\label{subsec:bvact}

As reviewed in \ccite{Zucchini:2017irg}, the BV quantum MA and ME
are central in BV theory. Here, we present an algebraic theory of them, concentrating 
on their covariance under canonical maps. 

Let $X$ be a non singular BV algebra with BV Laplacian $\varDelta_X$.

\begin{defi} \label{def:bvqma}
A BV quantum MA of $X$ is a special element $S\in X$ mod $\mathbb{R}1_X$ 
meeting the conditions 
\begin{align}
&|S|=0,
\vphantom{\Big]}
\label{bvact1}
\\
&\varDelta_XS+\frac{1}{2}(S,S)_X=0. 
\vphantom{\Big]}
\label{bvact2}
\end{align}
\end{defi}

\noindent 
\ceqref{bvact2} is called (abstract) BV quantum ME. 

With any MA, there is associated a variation operator.

\begin{defi}
The BV variation operator of a BV quantum MA $S$
is the linear endomorphism $\varDelta_{XS}:X\rightarrow X$ given by 
\begin{equation}
\varDelta_{XS}f=\varDelta_Xf+(S,f)_X
\label{bvact3}
\end{equation}
with $f\in X$. 
\end{defi}


\begin{prop} $\varDelta_{XS}$ is degree $1$ and nilpotent,
\begin{align}
&|\varDelta_{XS}|=1,
\vphantom{\Big]}
\label{bvact4}
\\
&\varDelta_{XS}{}^2=0. 
\vphantom{\Big]}
\label{bvact45}
\end{align}
\end{prop}

\begin{proof}
We provide the proof of these facts for the sake of completeness, albeit it is well--known.
The property \ceqref{bvact4} is obvious. We have to show only \ceqref{bvact45}. 
Let $f\in X$. Using \ceqref{bvalg3}, \ceqref{bvalg8} and \ceqref{bvalg11}, we find
\begin{align}
\varDelta_{XS}\varDelta_{XS}f
&=\varDelta_X(\varDelta_Xf+(S,f)_X)+(S,\varDelta_Xf+(S,f)_X)_X
\vphantom{\Big]}
\label{bvact6}
\\
&=(\varDelta_XS+(S,S)_X/2,f)_X=0,
\vphantom{\Big]}
\nonumber
\end{align}
where in the last step we used \ceqref{bvact2}. \ceqref{bvact45} follows. 
\end{proof}

\noindent
Hence, $(X,\varDelta_{XS})$ is a cochain complex to which there is attached a cohomology $H^*(X,\varDelta_{XS})$,
the BV MA cohomology of $S$.

Let $X$, $Y$ be non singular 
BV algebras and $\alpha:X\rightarrow Y$ be a canonical map (cf. def. \cref{def:canmap}). 

\begin{defi} \label{def:qmactcan}
The $\alpha$--transform of a BV quantum MA $S$ in $X$ is 
\begin{equation}
\hat\alpha S=\alpha S+r_\alpha, 
\label{bvact7}
\end{equation}
where $r_\alpha$ is the logarithmic Jacobian of $\alpha$. 
\end{defi}

\begin{prop} \label{prop:qmecan}
$\hat\alpha S$ is a BV quantum MA in $Y$. 
\end{prop}

\begin{proof}
By \ceqref{bvact7}, using \ceqref{bvcanon2}, \ceqref{bvcanon3} and \ceqref{bvcanon4}, we have  
\begin{align}
\varDelta_Y\hat\alpha S+\frac{1}{2}(\hat\alpha S,\hat\alpha S)_Y
&=\varDelta_Y\alpha S+\varDelta_Y r_\alpha 
\vphantom{\Big]}
\label{bvact8}
\\
&\hspace{1cm}+\frac{1}{2}(\alpha S,\alpha S)_Y
+(r_\alpha, \alpha S)_Y+\frac{1}{2}(r_\alpha,r_\alpha)_Y
\vphantom{\Big]}
\nonumber
\\
&=\alpha\bigg(\varDelta_XS+\frac{1}{2}(S,S)_X\bigg)=0, 
\vphantom{\Big]}
\nonumber
\end{align}
where in the last step we used \ceqref{bvact2} again. This shows the statement. 
\end{proof}

Canonical transformation \pagebreak is compatible with the compositional structure of canonical maps. 
If $X,Y,Z$ are non singular BV algebras and $\alpha:X\rightarrow Y$, $\beta:Y\rightarrow Z$ are canonical 
maps, then one has \hphantom{xxxxxxxx}
\begin{equation}
\widehat{\beta\alpha}S=\hat\beta\hat\alpha S
\label{bvact15}
\end{equation}
for any BV quantum MA $S$ in $X$,
as follows readily from \ceqref{bvact7} by means of  a simple application of 
the Jacobian relation \ceqref{bvcanon5}.

The BV variation operator of the BV quantum MA behaves covariantly under canonical 
transformation. 

\begin{prop}
For any canonical map $\alpha:X\rightarrow Y$ of BV algebras and BV quantum MA $S$, one 
has \hphantom{xxxxxxxxxxxxx}
\begin{equation}
\varDelta_{Y\hat\alpha S}=\alpha\varDelta_{XS}\alpha^{-1}.
\label{bvact9}
\end{equation}
\end{prop}

\begin{proof}
Let $f\in X$. Then, by \ceqref{bvact7} and \ceqref{bvcanon2},
\begin{align}
\varDelta_{Y\hat\alpha S}f
&=\varDelta_Yf+(\alpha S+r_\alpha,f)_Y
\vphantom{\Big]}
\label{bvact10}
\\
&=\alpha\varDelta_X\alpha^{-1}f-(r_\alpha,f)_Y+\alpha(S,\alpha^{-1}f)_X+(r_\alpha,f)_Y
\vphantom{\Big]}
\nonumber
\\
&=\alpha\varDelta_{XS}\alpha^{-1}f. 
\vphantom{\Big]}
\nonumber
\end{align}
This proves \ceqref{bvact9}
\end{proof}

\noindent
From \ceqref{bvact9}, it follows that the BV MA cohomologies of $S$ and $\hat\alpha S$ are isomorphic.
MA cohomology is so a canonical transformation invariant. 

Let $X$ be BV a non singular algebra and $\xi:X\rightarrow X$ be an infinitesimal canonical map (cf. def. \cref{def:infbvcan}). 

\begin{defi} \label{def:qmactinfcan}
The $\xi$--variation of a BV quantum MA $S$ in $X$ is 
\begin{equation}
\hat\delta_\xi S=\xi S+e_\xi,
\label{bvact11}
\end{equation}
where $e_\xi$ is the logarithmic Jacobian of $\xi$. 
\end{defi}

\begin{prop} \label{prop:infqmecan}
$\hat\delta_\xi S$ is a $0$--cocycle of $\varDelta_{XS}$,
\begin{equation}
\varDelta_{XS}\hat\delta_\xi S=0.
\label{bvact12}
\end{equation}
\end{prop}

\begin{proof} By \ceqref{bvact11}, \ceqref{bvact3}, using \ceqref{bvalg7}, \ceqref{bvcanon9}, 
\ceqref{bvcanon10} and \ceqref{bvcanon11}, we have  
\begin{align}
\varDelta_{XS}\hat\delta_\xi S
&=\varDelta_X(\xi S+e_\xi)+(S,\xi S+e_\xi)_X
\vphantom{\Big]}
\label{bvact 13}
\\
&=\xi\varDelta_XS-(e_\xi,S)_X+(S,\xi S)_X+(S,e_\xi)_X
\vphantom{\Big]}
\nonumber
\\
&=\xi\bigg(\varDelta_XS+\frac{1}{2}(S,S)_X\bigg)=0,
\vphantom{\Big]}
\nonumber
\end{align}
where in the last step we used \ceqref{bvact2} once more. This shows \ceqref{bvact12}. 
\end{proof}

\noindent
In particular, when $\xi=\ad_X x$ with $x\in X$, $|x|=-1$, \ceqref{bvact12} yields 
\begin{equation}
\hat\delta_{\ad_X x}S=-\varDelta_Xx-(S,x)_X=-\varDelta_{XS}x
\label{bvact14}
\end{equation}
by \ceqref{bvcanon14}. $\hat\delta_{\ad_X x} S$ is so a $0$--coboundary of $\varDelta_{XS}$. 

Infinitesimal canonical transformation is compatible with the Lie bracketing
structure of infinitesimal canonical maps. 
If $X$ is a non singular BV algebra and $\xi,\eta:X\rightarrow X$ are infinitesimal canonical 
maps, then,
\begin{equation}
\hat\delta_{[\xi,\eta]}S=[\hat\delta_{\xi},\hat\delta_{\eta}]S
\label{bvact16}
\end{equation}
for any BV quantum MA $S$ in $X$. We get this relation immediately
from \ceqref{bvact11}, using the infinitesimal Jacobian relation \ceqref{bvcanon12},


\subsection{\textcolor{blue}{\sffamily BV RG flow and  EA}}\label{subsec:bvrenflow}

In this subsection we introduce and study RG flows and EAs 
in the algebraic BV theoretic framework of subsect. \cref{subsec:bvact}. 

We consider a graded commutative algebra $X$ and a parameter graded manifold $T$
and select a formal differentiation structure of $X$ over $T$ (cf. subsect \cref{subsec:bvfam}).
Further, we assume that a non singular BV Laplacian family $\varDelta_t$ over $T$ 
together with its associated BV bracket family  $(-,-)_t$ are given 
(cf. defs. \cref{def:tfamD}, \cref{def:tfambvbrac}) rendering $X$ a non singular BV algebra $X_t$ 
for each $t$. 

\begin{defi} \label{def:bvrgfl}
A BV RG flow along $T$ is just a BV flow $\chi_{t,s}:X_s\rightarrow X_t$ along $T$
(cf. subsect. \cref{subsec:bvflow}). 
\end{defi}

\noindent
The reference to the RG is made only to highlight its physical origin. 
Our basic examples $T=\mathbb{R}$ and $T=T[1]\mathbb{R}$ should be kept in mind. 

\begin{defi} \label{def:bvrgea}
A BV RG EA along $T$ for the RG flow $\chi_{t,s}$ is a family $S_t$ 
of degree $0$ elements of $X$ over $T$ obeying 
\begin{equation}
\varDelta_tS_t+\frac{1}{2}(S_t,S_t)_t=0 
\label{bvrenflow1}
\end{equation}
(cf. eq. \ceqref{bvact2}) and with the property that \hphantom{xxxxxx}
\begin{equation}
S_t=\hat\chi_{t,s}S_s=\chi_{t,s}S_s+r_{\chi t,s} \qquad \text{\rm mod $\zz_0(T,X)$}
\label{bvrenflow2}
\end{equation}
(cf. 
eq. \ceqref{bvact7}), 
where $r_{\chi t,s}$ is the logarithmic Jacobian of $\chi_{t,s}$.
\end{defi}

\noindent
By \ceqref{bvrenflow1}, for each $t$ the mod $\mathbb{R}1_X$ class of $S_t$ is a 
BV quantum MA. By prop. \cref{prop:qmecan}, it is sufficient that 
$S_t$ obeys \ceqref{bvrenflow1} for a single $t$ for $S_t$ doing so for all 
$t$, as the maps $\chi_{t,s}:X_s\rightarrow X_t$ are canonical. 
The family $S_t$ can be redefined by adding to it a central logarithmic family
of $\zz_0(T,X)$ without spoiling its being a RG EA provided the logarithmic Jacobian family 
$r_{\chi t,s}$ of the RG flow $\chi_{t,s}$ is redefined accordingly by a suitable
central logarithmic family of $\zz_0(T^2,X)$ as allowed. 

As in subsects. \cref{subsec:bvfam}, \cref{subsec:bvflow}, here too it is useful to probe
the parameter dependence of the BV RG flow and EA using vector fields of $T$. 

Pick a vector field $D\in \Vect(T)$. 

\begin{prop}
$S_t$ obeys the $D$--BV RGE
\begin{equation}
D_tS_t=\chi^D{}_tS_t+r^D{}_{\chi t}, 
\vphantom{\bigg]}
\label{bvrenflow3}
\end{equation}
where $\chi^D{}_t$, $r^D{}_{\chi t}$  are the $D$--infinitesimal generator and logarithmic Jacobian
of the flow $\chi_{t,s}$ (cf. eqs. \ceqref{bvflow7}, \ceqref{bvflow8})
\end{prop}

\begin{proof}
From \ceqref{bvrenflow2}, using  \ceqref{bvflow7}, \ceqref{bvflow8}, we have 
\begin{align}
D_tS_t&=\big[D_t\chi_{t,s}S_s+D_tr_{\chi t,s}\big]\big|_{s=t}
\vphantom{\Big]}
\label{bvrenflow4}
\\
&=\chi^D{}_tS_t+r^D{}_{\chi t},
\vphantom{\Big]}
\nonumber
\end{align}
showing \ceqref{bvrenflow3}. 
\end{proof}

\noindent
Eq. \ceqref{bvrenflow3} is the abstract version of the physicist's RGE as emerges in the present
formulation. Note that there is one such RGE for each choice of the vector field $D$. 

It is possible to employ stabilizers to generate new BV RG EAs from a given one. 

\begin{prop} \label{prop:stabeffact}
Let $\gamma_t:X_t\rightarrow X_t$ be a given canonical map family over $T$ 
stabilizing of the BV RG flow $\chi_{t,s}$ (cf. def. \cref{def:stab}).   
Then, if $S_t$ is a BV RG EA for $\chi_{t,s}$, then so is \hphantom{xxxxxxxxxxxxxxxxxxxxxxxxxxx}
\begin{equation}
\hat\gamma_tS_t=\gamma_tS_t+r_{\gamma t}.
\label{bvrenflow5}
\end{equation}
\end{prop}

\begin{proof}
Since $S_t$ is BV MA of $X_t$ and $\gamma_t$ is canonical, 
$\hat\gamma_tS_t$ is also a BV MA of $X_t$, by prop. \cref{prop:qmecan}. 
Furthermore, since $S_t$ satisfies relation \eqref{bvrenflow2}, we have \hphantom{xxxxxxxxxxx}
\begin{align}
\hat\gamma_tS_t
&=\gamma_t(\chi_{t,s}S_s+r_{\chi t,s})+r_{\gamma t}. 
\vphantom{\Big]}
\label{bvrenflow6}
\\
&=\gamma_t\chi_{t,s}\gamma_s{}^{-1}(\gamma_sS_s+r_{\gamma s})
-\gamma_t\chi_{t,s}\gamma_s{}^{-1}r_{\gamma s}+\gamma_tr_{\chi t,s}+r_{\gamma t}
\vphantom{\Big]}
\nonumber
\\
&={}^\gamma\chi_{t,s}\hat\gamma_sS_s+r_{{}^\gamma\chi t,s}
\vphantom{\Big]}
\nonumber
\\
&=\chi_{t,s}\hat\gamma_sS_s+r_{\chi t,s},
\vphantom{\Big]}
\nonumber
\end{align} 
where we used the identities \ceqref{bvflow14}, \ceqref{bvflow18} and relation \ceqref{bvflow20}
holding for stabilizers. Hence, $\hat\gamma_tS_t$ also satisfies \ceqref{bvrenflow2}. 
\end{proof}

The infinitesimal version of \ceqref{bvrenflow5} is 
\begin{equation}
\hat\delta_\epsilon S_t=\epsilon_tS_t+e_{\epsilon t},
\label{bvrenflow7}
\end{equation}
where $\epsilon_t$ is an infinitesimal canonical map family 
stabilizing the flow  $\chi_{t,s}$ (cf. defs. \cref{def:infstab} and eq. \ceqref{bvact11}). 
$\hat\delta_\epsilon S_t$ describes
an infinitesimal deformation of the BV RG EA $S_t$. 
Recall that, by prop. \cref{prop:infqmecan}, $\hat\delta_\epsilon S_t$ is a $0$-cocycle of $\varDelta_{tS_t}$.
When $\epsilon_t=\ad_tx_t$ for a suitable degree $-1$ family $x_t$, we have 
\begin{equation}
\hat\delta_{\ad x} S_t=-\varDelta_tx_t-(S_t,x_t)_t=-\varDelta_{tS_t}x_t.
\label{bvrenflow8}
\end{equation}
$\hat\delta_{\ad x}S_t$ is a so $0$-coboundary of $\varDelta_{tS_t}$.


\subsection{\textcolor{blue}{\sffamily Relation to standard RGEs}}\label{subsec:bveffact}

In subsect. \cref{subsec:bvrenflow}, we have provided an abstract formulation of RG flow and 
equation in BV theory. It is now time to make contact with more customary physical formulations 
of the RG. 

We assume that a graded commutative algebra $X$ and a parameter graded manifold $T$ 
together with a formal differentiation structure of $X$ over $T$ 
are given. In physical parlance, $X$ and $T$ would be the field and the scale parameter
space, respectively. The choice of $T$ determines the type of RGE we get. 

\begin{exa}\label{exa:t=r} 
$T=\mathbb{R}$, the basic RG set--up.
\end{exa} 
\vspace{-.25cm}
\noindent
In the simplest case, the parameter manifold $T$ is just $\mathbb{R}$, 
coordinatized by a degree $0$ real coordinate $t$. 
We have so a non singular BV Laplacian family $\varDelta_t$  over $\mathbb{R}$  
together with its associated BV bracket family $(-,-)_t$ 
rendering $X$ a BV algebra $X_t$ for each real $t$. 
We have furthermore a BV RG flow $\chi_{t,s}:X_s\rightarrow X_t$ along $\mathbb{R}$  
together with  its associated logarithmic Jacobian family $r_{\chi t,s}$.
Finally, a BV RG EA $S_t$ along $\mathbb{R}$ is given. In the following, we call the above the basic RG set--up. 

The main object of study in the basic RG set--up is the EA $S_t$.
This satisfies two basic equations. 
\begin{enumerate}

\item The BV ME \ceqref{bvrenflow1}, \hphantom{xxxxxxxxxxxxxxxxx}
\begin{equation}
\varDelta_tS_t+\frac{1}{2}(S_t,S_t)_t=0. 
\label{bveffact1}
\end{equation}

\item The BV RGE ensuing from a relevant choice of a vector field $D$ probing the 
parameter manifold $\mathbb{R}$. 

\end{enumerate}
\noindent
We suppose, as is natural, that the formal differentiation structure is such that 
the degree $0$ derivation $d/dt$ belongs to $\Vect(\mathbb{R})$. 
Associated with this are the $d/dt$--derived BV Laplacian \hphantom{xxxxxxxxxxxxxxxxx}
\begin{equation}
\frac{d\varDelta_t}{dt}
:=\varDelta^{d/dt}{}_t \vphantom{\bigg]^i_g}
\label{bveffact1/2}
\end{equation}
(cf. eq. \ceqref{bvfam0}) \pagebreak and the $d/dt$--infinitesimal generator and logarithmic Jacobian of the BV RG 
flow $\chi_{t,s}$ \hphantom{xxxxxxxxxxxxxxxxxxxxxx}
\begin{align}
&\chi^\bcdot{}_t:=\chi^{d/dt}{}_t,
\vphantom{\Big]}
\label{bveffact2}
\\
&r^\bcdot{}_{\chi t}:=r^{d/dt}{}_{\chi t}
\vphantom{\Big]}
\label{bveffact3}
\end{align}
(cf. eqs. \ceqref{bvflow7} and \ceqref{bvflow8}). The $d/dt$--evolution equation 
\begin{equation}
\frac{d\varDelta_t}{dt}-[\chi^\bcdot{}_t,\varDelta_t]+\ad_tr^\bcdot{}_{\chi t}=0 \vphantom{\bigg]}
\label{bveffact3/1}
\end{equation}
(cf. eq. \ceqref{bvflow10}) is a differential equation governing the $t$ dependence of $\varDelta_t$, 
since $d\varDelta_t/dt$ is the $t$ derivative of $\varDelta_t$ just as suggested by the notation. 

The $d/dt$--BV RGE \ceqref{bvrenflow3} reads 
\begin{equation}
\frac{dS_t}{dt}=\chi^\bcdot{}_tS_t+r^\bcdot{}_{\chi t}.
\label{bveffact4}
\end{equation}
This form of the RGE 
is very general and precisely for this reason not particularly 
useful. A more interesting version of the equation can be obtained by adopting a more structured 
choice of the parameter space $T$. Adding an odd parameter, in particular,
will enrich the RGE with a kind of supersymmetry. 

\begin{exa} \label{exa:t=t1r} 
$T=T[1]\mathbb{R}$, the extended RG set--up
\end{exa} 
\vspace{-.25cm}
\noindent
As anticipated in the previous paragraph, we consider next the case where
the parameter graded manifold $T$ is the shifted tangent bundle
$T[1]\mathbb{R}$ of $\mathbb{R}$ coordinatized by a degree $0$ real base coordinate $t$ and 
a degree $1$ real fiber coordinate $\theta$. 
We have in this way a non singular BV Laplacian family $\varDelta_{t\theta}$  over $T[1]\mathbb{R}$  
together with its associated BV bracket family $(-,-)_{t\theta}$ 
rendering $X$ a BV algebra $X_{t\theta}$ for each shifted real pair $t\theta$. 
We have furthermore a BV RG flow $\chi_{t\theta,s\zeta}:X_{s\zeta}\rightarrow X_{t\theta}$ along $T[1]\mathbb{R}$  
together with  its associated logarithmic Jacobian family $r_{\chi t\theta,s\zeta}$.
Finally, a BV RG EA $S_{t\theta}$ along $T[1]\mathbb{R}$ is given. In the following, we call the above the 
extended RG set--up to distinguish it from the basic set--up introduced  in the first part of this subsection.

Analogously to the basic case analyzed above, the main object of study in the extended RG set--up
is the BV RG EA $S_{t\theta}$. Again, this satisfies two basic equations. 
\begin{enumerate}

\item The BV ME \ceqref{bvrenflow1},
\begin{equation}
\varDelta_{t\theta}S_{t\theta}+\frac{1}{2}(S_{t\theta},S_{t\theta})_{t\theta}=0. 
\label{bveffact25/1}
\end{equation}

\item The BV RGE associated with a relevant choice of a vector field $D$ probing the 
parameter manifold $T[1]\mathbb{R}$. 

\end{enumerate}

We suppose that the degree $0$ and $-1$ derivations $\partial/\partial t$ and $\partial/\partial \theta$
belong to $\Vect(T[1]\mathbb{R})$ in the given formal differentiation structure of $X$.
$\partial/\partial t$, $\partial/\partial \theta$ will not however be treated on the same footing in 
the following investigation, 
because the former has an analytic significance while the latter works just as an algebraic device. 

Associated with $\partial/\partial t$ are the $\partial/\partial t$--derived BV Laplacian
\begin{equation}
\frac{\partial\varDelta_{t\theta}}{\partial t}:=\varDelta^{\partial/\partial t}{}_{t\theta}
\vphantom{\Big]}
\label{bveffact10/1}
\end{equation}
and the $\partial/\partial t$-infinitesimal generator and logarithmic Jacobian 
\begin{align}
&\chi^\bcdot{}_{t\theta}:=\chi^{\partial/\partial t}{}_{t\theta},
\vphantom{\Big]}
\label{bveffact11}
\\
&r^\bcdot{}_{\chi t\theta}:=r^{\partial/\partial t}{}_{\chi t\theta}
\vphantom{\Big]}
\label{bveffact12}
\end{align}
defined according to \ceqref{bvfam0}, \ceqref{bvflow7} and \ceqref{bvflow8}, respectively. 
The $\partial/\partial t$--evolution equation \ceqref{bvflow10} takes the form 
\begin{equation}
\frac{\partial\varDelta_{t\theta}}{\partial t}
-[\chi^\bcdot{}_{t\theta},\varDelta_{t\theta}]+\ad_{t\theta}r^\bcdot{}_{\chi t\theta}=0.
\label{bveffact14/1}
\end{equation}
\ceqref{bveffact14/1} is a differential equation governing the $t$ dependence of $\varDelta_{t\theta}$, 
since again $\partial\varDelta_{t\theta}/\partial t$ is just the $t$ derivative of $\varDelta_{t\theta}$. 

Associated similarly with $\partial/\partial \theta$ are $\partial/\partial \theta$--derived BV Laplacian
$\varDelta^\star {}_{t\theta}:=\varDelta^{\partial/\partial \theta}{}_{t\theta}$ and the 
$\partial/\partial \theta$--infinitesimal generator and logarithmic Jacobian
$\chi^\star{}_{t\theta}:=\chi^{\partial/\partial \theta}{}_{t\theta}$. 
The $\partial/\partial \theta$--evolution equation is 
\begin{equation}
\varDelta^\star {}_{t\theta}-[\chi^\star{}_{t\theta},\varDelta_{t\theta}]+\ad_{t\theta}r^\star{}_{\chi t\theta}=0. \vphantom{\bigg]}
\label{bveffact14/2}
\end{equation}
Unlike \ceqref{bveffact14/1}, \ceqref{bveffact14/2} is a mere algebraic identity, as, by the nilpotence of $\theta$, 
$\varDelta_{t\theta}$ is a degree $1$ polynomial of $\theta$ whose leading coefficient
is precisely $\varDelta^\star {}_{t\theta}=\partial\varDelta_{t\theta}/\partial \theta$. 

From \ceqref{bvrenflow3}, we can write down the $\partial/\partial t$--BV RGE
\begin{equation}
\frac{\partial S_{t\theta}}{\partial t}=\chi^\bcdot{}_{t\theta}S_{t\theta}+r^\bcdot{}_{\chi t\theta}.
\vphantom{\Big]}
\label{bveffact15}
\end{equation}
The $\partial/\partial \theta$--BV RGE takes the form 
\begin{equation}\
S^\star {}_{t\theta}=\chi^\star{}_{t\theta}S_{t\theta}+r^\star{}_{\chi t\theta},
\vphantom{\Big]}
\label{bveffact16}
\end{equation}
where $S^\star {}_{t\theta}=\partial S_{t\theta}/\partial\theta$, and is again an algebraic relation
since $S_{t\theta}$ is a degree $1$ polynomial of $\theta$ whose leading coefficient
is precisely $S^\star {}_{t\theta}$. 

The degree $-1$ EA $S^\star {}_{t\theta}$ allows us to write the infinitesimal generator and logarithmic Jacobian 
$\chi^\bcdot{}_{t\theta}$ and $r^\bcdot{}_{t\theta}$ of the BV RG flow in the reduced form
\begin{align}
&\chi^\bcdot{}_{t\theta}=-\ad_{t\theta}S^\star {}_{t\theta}+\bar\chi^\bcdot{}_{t\theta},
\vphantom{\Big]}
\label{bveffact18}
\\
&r^\bcdot{}_{\chi t\theta}=\varDelta_{t\theta}S^\star {}_{t\theta}+\bar r^\bcdot{}_{\chi t\theta},
\vphantom{\Big]}
\label{bveffact19}
\end{align}
where $\bar\chi^\bcdot{}_{t\theta}$ and $\bar r^\bcdot{}_{t\theta}$ are a degree $0$ derivation and a degree $0$ element of $X$
called reduced infinitesimal generator and Jacobian, respectively. 

\begin{prop} \label{prop:evoleqred}
The evolution equation \ceqref{bveffact14/1} can be cast in the reduced form 
\begin{equation}
\frac{\partial\varDelta_{t\theta}}{\partial t}-[\bar\chi^\bcdot{}_{t\theta},\varDelta_{t\theta}]
+\ad_{t\theta}\bar r^\bcdot{}_{\chi t\theta}=0. 
\label{bveffact19/1}
\end{equation}
\end{prop}

\begin{proof}
By virtue of \ceqref{bvalg11}, we have 
\begin{equation}
[\ad_{t\theta}S^\star {}_{t\theta},\varDelta_{t\theta}]+\ad_{t\theta}\varDelta_{t\theta}S^\star {}_{t\theta}=0.
\label{bveffact19/2}
\end{equation}
Subtracting \ceqref{bveffact19/2} from \ceqref{bveffact14/1} and then using \ceqref{bveffact18}, 
\ceqref{bveffact19}, we obtain \ceqref{bveffact19/2} readily. 
\end{proof}

\begin{prop} \label{prop:rgered}
The BV RGE \ceqref{bveffact15} can be cast as
\begin{equation}
\frac{\partial S_{t\theta}}{\partial t}=\varDelta^\star {}_{t\theta}S_{t\theta}+\frac{1}{2}(S_{t\theta},S_{t\theta})^\star {}_{t\theta}
+\bar\chi^\bcdot{}_{t\theta}S_{t\theta}+\bar r^\bcdot{}_{\chi t\theta},
\label{bveffact25}
\end{equation}
where $(-,-)^\star {}_{t\theta}:=(-,-)^{\partial/\partial \theta}{}_{t\theta}$ 
is $\partial/\partial \theta$--derived BV bracket defined according to \ceqref{bvfam6}.
Further, one has\hphantom{xxxxxxxxxxxxxx}
\begin{equation}
\varDelta_{t\theta}S^\star {}_{t\theta}+(S_{t\theta},S^\star {}_{t\theta})_{t\theta}
=\varDelta^\star {}_{t\theta}S_{t\theta}+\frac{1}{2}(S_{t\theta},S_{t\theta})^\star {}_{t\theta}.
\label{bveffact23}
\end{equation}
\end{prop}

\begin{proof} Expressing $\chi^\bcdot{}_{t\theta}$ and $r^\bcdot{}_{t\theta}$ in the reduced form 
\ceqref{bveffact18}, \ceqref{bveffact19}, the RGE \ceqref{bveffact15} becomes 
\begin{equation}
\frac{\partial S_{t\theta}}{\partial t}=\varDelta_{t\theta}S^\star {}_{t\theta}+(S_{t\theta},S^\star {}_{t\theta})_{t\theta}
+\bar\chi^\bcdot{}_{t\theta}S_{t\theta}+\bar r^\bcdot{}_{\chi t\theta}.
\label{bveffact20}
\end{equation}
Acting with $\partial/\partial\theta$ on both sides of \ceqref{bveffact25/1}, we find  
\begin{align}
0&=\frac{\partial}{\partial\theta}\varDelta_{t\theta}S_{t\theta}-\varDelta_{t\theta}\frac{\partial}{\partial\theta}S_{t\theta}
-\bigg(S_{t\theta},\frac{\partial}{\partial\theta}S_{t\theta}\bigg)_{t\theta}
+\frac{1}{2}\frac{\partial}{\partial\theta}(S_{t\zeta},S_{t\zeta})_{t\theta}\bigg|_{\zeta=\theta}
\vphantom{\Big]}
\label{bveffact24}
\\
&=\varDelta^\star {}_{t\theta}S_{t\theta}+\frac{1}{2}(S_{t\theta},S_{t\theta})^\star {}_{t\theta}-\varDelta_{t\theta}S^\star {}_{t\theta}
-(S_{t\theta},,S^\star {}_{t\theta})_{t\theta},
\vphantom{\Big]}
\nonumber
\end{align}
proving \ceqref{bveffact23}. Substituting \ceqref{bveffact23} into \ceqref{bveffact23}, we 
rewrite the RGE in the form \ceqref{bveffact25}. 
\end{proof}

As $\mathbb{R}$ can be identified with the submanifold of $T[1]\mathbb{R}=\mathbb{R}\times \mathbb{R}[1]$
defined by the condition $\theta=0$, the basic RG set--up of ex. \cref{exa:t=r} is retrievable in the extended one
by setting $\theta=0$ throughout. The import of the above results becomes apparent when we examine their implications 
for the underlying basic set--up. 


Below, we adopt the following convention. For any object $O_{t\theta}$ depending on $t$, $\theta$,
we set $O_t=O_{t0}$ and call $O_t$ the basic projection of $O_{t\theta}$ since it pertains to the basic 
RG set--up subjacent the extended one. The following result is an immediate consequence of 
props. \cref{prop:evoleqred} and \cref{prop:rgered}. 

\begin{prop}
Upon carrying out the basic projection, the reduced evolution equation \ceqref{bveffact19/1} reads 
\hphantom{xxxxxxxxx}
\begin{equation}
\frac{d\varDelta_t}{d t}-[\bar\chi^\bcdot{}_t,\varDelta_t]+\ad_t\bar r^\bcdot{}_{\chi t}=0.
\label{bveffact28}
\end{equation}
Similarly, the BV RGE \ceqref{bveffact15} furnishes 
\begin{equation}
\frac{dS_t}{dt}=\varDelta^\star {}_tS_t+\frac{1}{2}(S_t,S_t)^\star {}_t
+\bar\chi^\bcdot{}_tS_t+\bar r^\bcdot{}_{\chi t},
\label{bveffact27}
\end{equation}
where by \ceqref{bveffact23} \hphantom{xxxxxxxxxxx}
\begin{equation}
\varDelta_tS^\star {}_t+(S_t,S^\star {}_t)_t
=\varDelta^\star {}_tS_t+\frac{1}{2}(S_t,S_t)^\star {}_t.
\label{bveffact26}
\end{equation}
\end{prop}

\noindent
Above, $\varDelta^\star {}_t$ is a degree $0$ second order differential operator and 
$(-,-)^\star {}_t$ is a degree $0$ graded symmetric bracket. The first two terms of eq. \ceqref{bveffact27}, 
so, are analogous in form to those appearing in Polchinski RGE \ccite{Polchinski:1983gv}. 
Thus, the embedding of the basic RG set--up in the extended one leads to 
more structured RGEs that those one would have in an unenhanced basic RG framework. 
Under favorable conditions, the reduced infinitesimal generator and Jacobian 
$\bar\chi^\bcdot{}_t$ and $\bar r^\bcdot{}_t$ turn out 
to be simpler than their unreduced counterparts $\chi^\bcdot{}_t$ and $r^\bcdot{}_t$ 
and the RGE \ceqref{bveffact27} correspondingly is more amenable to an analytical study.

Eq. \ceqref{bveffact26} indicates that the Polchinski terms of the RGE \ceqref{bveffact27}
are trivial in the $\varDelta_{tS}$ cohomology. This is not surprising, if we recall that 
in quantum field theory the RG flow must preserve the partition function. 

The use of $T[1]\mathbb{R}$ rather than $\mathbb{R}$ as parameter space has allowed us to 
control the mutual consistence of a number of equations, which otherwise would have been unrelated 
and of dubious compatibility. Of course, the extended RG set--up contains extra structure.


\subsection{\textcolor{blue}{\sffamily Reconstruction of  BV RG flow from infinitesimal data}}\label{subsec:genres}

As explained in subsect. \ceqref{subsec:bveffact},
in the basic RG set--up, and so also in the ex\-tended set--up which contains it, 
the infinitesimal generator $\chi^\bcdot{}_t$ and logarithmic Jacobian $r^\bcdot{}_{\chi t}$ 
of a BV RG flow $\chi_{t,s}$ 
form a family of degree $0$ derivations and one of degree $0$ elements of $X$
with certain properties. When the underlying algebra $X$ supports besides the given
formal differentiation structure also a formal integration structure over $\mathbb{R}$, 
that is a notion of integration with the usual properties of Riemann's, 
it is possible to reconstruct BV RG flow $\chi_{t,s}$ 
from the infinitesimal data $\chi^\bcdot{}_t$ and $r^\bcdot{}_{\chi t}$, 
unlike what happens for a generic parameter space $T$. 

\begin{prop} \label{prop:genres1}
Assume the following data are given.
\begin{enumerate}

\item A family of degree $0$ derivations $\chi^\bcdot{}_t:X\rightarrow X$ over $\mathbb{R}$ 
such that 
\begin{equation}
\chi^\bcdot{}_t(f,g)_t=(\chi^\bcdot{}_tf,g)_t+(f,\chi^\bcdot{}_tg)_t+\frac{d}{dt}(f,g)_t
\label{genres1}
\end{equation}
for $f,g\in X$. 

\item A family of degree $0$ elements $r^\bcdot{}_{\chi t}\in X$ over $\mathbb{R}$. 

\end{enumerate}
\noindent
Assume further that \hphantom{xxxxxxxxxxxx}
\begin{equation}
\frac{d\varDelta_t}{dt}-[\chi^\bcdot{}_t,\varDelta_t]+\ad_tr^\bcdot{}_{\chi t}=0
\label{genres2}
\end{equation}
(cf. eq. \ceqref{bveffact3/1}). Then, there is a BV RG flow $\chi_{t,s}:X\rightarrow X$ having $\chi^\bcdot{}_t$ and 
\begin{equation}
r_{\chi t,s}=\int_s^t d\tau\, \chi_{t,\tau}r^\bcdot{}_{\chi \tau}
\label{genres3}
\end{equation}
as associated infinitesimal generator and infinitesimal logarithmic Jacobian, respectively. 
\end{prop}

\begin{proof}
Let $\chi_{t,s}$ be the family of algebra isomorphism of $X$ over $\mathbb{R}^2$ 
satisfying relations \ceqref{bvflow1}--\ceqref{bvflow3} generated by the degree $0$
derivation family $\chi^\bcdot{}_t$. By definition, $\chi_{t,s}$ is the unique family
of algebra isomorphism of $X$ obeying 
\begin{align}
\frac{\partial\chi_{t,s}}{\partial t}
&=\chi^\bcdot{}_t\chi_{t,s},
\vphantom{\Big]}
\label{genres4}
\\
\frac{\partial\chi_{t,s}}{\partial s}
&=-\chi_{t,s}\chi^\bcdot{}_s
\vphantom{\Big]}
\label{genres5}
\end{align}
with the initial conditions \hphantom{xxxxxxxxxxxx}
\begin{equation}
\chi_{s,s}=\id_X.
\label{genres6}
\end{equation}
Let $f,g\in X$. By virtue of \ceqref{genres1} and \ceqref{genres4}, \ceqref{genres5}, we have 
\begin{align}
\frac{\partial}{\partial t}\big(\chi_{s,t}(\chi_{t,s}f,\chi_{t,s}g)_t\big)
&=\chi_{s,t}\bigg[-\chi^\bcdot{}_t(\chi_{t,s}f,\chi_{t,s}g)_t
\vphantom{\Big]}
\label{genres7}
\\
&\hspace{-2cm}+(\chi^\bcdot{}_t\chi_{t,s}f,\chi_{t,s}g)_t
+(\chi_{t,s}f,\chi^\bcdot{}_t\chi_{t,s}g)_t+\frac{d}{dt}(\chi_{u,s}f,\chi_{u,s}g)_t\Big|_{u=t}\bigg]=0.
\vphantom{\Big]}
\nonumber
\end{align}
Then, recalling \ceqref{genres6}, we have 
\begin{align}
\chi_{s,t}(\chi_{t,s}f,\chi_{t,s}g)_t&=\chi_{s,s}(\chi_{s,s}f,\chi_{s,s}g)_s
\vphantom{\Big]}
\label{genres8}
\\
&=(f,g)_s.
\vphantom{\Big]}
\nonumber
\end{align}
The algebra isomorphisms $\chi_{t,s}$ therefore preserve BV brackets. This
suggests that they may be canonical in the sense of def. \cref{def:canmap}. 

By virtue of \ceqref{genres2} and \ceqref{genres4}, \ceqref{genres5}, we have 
\begin{align}
\frac{\partial}{\partial t}\big(\chi_{s,t}\varDelta_t\chi_{t,s}\big)
&=\chi_{s,t}\bigg[-\chi^\bcdot{}_t\varDelta_t+\varDelta_t\chi^\bcdot{}_t+\frac{d\varDelta_t}{dt}\bigg]\chi_{t,s}
\vphantom{\Big]}
\label{genres9}
\\
&=-\chi_{s,t}\ad_tr^\bcdot{}_{\chi t}\chi_{t,s}
\vphantom{\Big]}
\nonumber
\\
&=-\ad_s(\chi_{s,t}r^\bcdot{}_{\chi t}),
\vphantom{\Big]}
\nonumber
\end{align}
where in the last step \ceqref{genres8} was used. 
Then, recalling \ceqref{genres6}, we have 
\begin{align}
\chi_{s,t}\varDelta_t\chi_{t,s}-\varDelta_s
&=\chi_{s,t}\varDelta_t\chi_{t,s}-\chi_{s,s}\varDelta_s\chi_{s,s}
\vphantom{\Big]}
\label{genres10}
\\
&=-\int_s^t d\tau\,\ad_s(\chi_{s,\tau}r^\bcdot{}_{\chi \tau})
\vphantom{\Big]}
\nonumber
\\
&=-\chi_{s,t}\ad_t\bigg(\int_s^t d\tau\,\chi_{t,\tau}r^\bcdot{}_{\chi \tau}\bigg)\chi_{t,s},
\vphantom{\Big]}
\nonumber
\end{align}
where in the last step \ceqref{genres8} was used again. 
This shows that $\chi_{t,s}$ is canonical with logarithmic Jacobian 
given by \ceqref{genres3}. Thus, $\chi_{t,s}$ is a BV RG flow. 
Since by construction the infinitesimal generator of $\chi_{t,s}$ is precisely
$\chi^\bcdot{}_t$, the statement follows. 
\end{proof}

\noindent
In physical notation, $\chi_{t,s}$ would be written as 
\begin{equation}
\chi_{t,s}=\Texp\bigg(\int_s^t d\tau \chi^\bcdot{}_\tau\bigg),
\label{genres11}
\end{equation}
where $\Texp$ denotes $t$--ordered exponential. 

Prop. \cref{prop:genres1} has an immediate and useful application: it provides a way to ascertain 
when a family of BV MAs over $\mathbb{R}$ satisfying an RG like equation such as \ceqref{bveffact4} 
containing RG type infinitesimal data constitute a BV EA. 

\begin{prop} \label{prop:genres2}
Assume that $\chi^\bcdot{}_t$ and $r^\bcdot{}_{\chi t}$ are families of degree $0$ derivations 
and elements of $X$ over $\mathbb{R}$, respectively, satisfying the same conditions as in prop. \cref{prop:genres1}.
Assume further that $S_t$ is a family of BV MAs over $\mathbb{R}$ obeying the differential equation
\ceqref{bveffact4}. Then, $S_t$ is a BV RG EA. 
\end{prop}

\begin{proof}
Let $\chi_{t,s}$ be the BV RG flow associated with the data $\chi^\bcdot{}_t$, $r^\bcdot{}_{\chi t}$ 
by prop. \cref{prop:genres1}. Then, \ceqref{genres4} holds. Furthermore, by \ceqref{genres3}, \pagebreak 
\begin{align}
\frac{\partial r_{\chi t,s}}{\partial t}
&=\frac{\partial }{\partial t}\int_s^t d\tau\, \chi_{t,\tau}r^\bcdot{}_{\chi \tau}
\vphantom{\Big]}
\label{genres12}
\\
&=\chi_{t,t}r^\bcdot{}_{\chi t}+\chi^\bcdot{}_t\int_s^t d\tau\, \chi_{t,\tau}r^\bcdot{}_{\chi \tau}
\vphantom{\Big]}
\nonumber
\\
&=r^\bcdot{}_{\chi t}+\chi^\bcdot{}_tr_{\chi t,s}. 
\vphantom{\Big]}
\nonumber
\end{align}
Pick now a value $t_0$ of $t$ and set 
\begin{equation}
\tilde S_t=\hat\chi_{t,t_0}S_{t_0}=\chi_{t,t_0}S_{t_0}+r_{\chi t,t_0}
\label{genres13}
\end{equation}
(cf. eq. \ceqref{bvact7}). By \ceqref{genres4} and \ceqref{genres12}, we have 
\begin{align}
\frac{d\tilde S_t}{dt}
&=\frac{\partial\chi_{t,t_0}}{\partial t}S_{t_0}+\frac{\partial r_{\chi t,t_0}}{\partial t}
\vphantom{\Big]}
\label{genres14}
\\
&=\chi^\bcdot{}_t(\chi_{t,t_0}S_{t_0}+r_{\chi t,t_0})+r^\bcdot{}_{\chi t}
\vphantom{\Big]}
\nonumber
\\
&=\chi^\bcdot{}_t\tilde S_t+r^\bcdot{}_{\chi t}.
\vphantom{\Big]}
\nonumber
\end{align}
Hence, $\tilde S_t$ obeys the differential equations \ceqref{bveffact4}.
As $\tilde S_{t_0}=S_{t_0}$, we have $\tilde S_t=S_t$
Since $t_0$ is arbitrary, \ceqref{genres13} implies that 
\begin{equation}
S_t=\hat\chi_{t,s}S_s.
\label{genres15}
\end{equation}
This proves that $S_t$ is a BV EA as claimed. 
\end{proof}

\vfil\eject

\section{\textcolor{blue}{\sffamily Models of  Batalin--Vilkovisky  renormalization group}}\label{sec:models}

An important issue in the BV theory of the RG developed in this paper is the construction 
of non trivial models which exemplify it. A program of this scope  
certainly cannot be carried out to its full extent in the limited space of this paper, still a few 
simple but non trivial models can be built. 

In what follows, we work in the degree $-1$ symplectic framework originally developed 
by Costello in ref. \ccite{Costello:2007ei}, which has a very rich structure and lend itself 
particularly well to our task. We first review the framework to set our notation, 
presenting along the way some new results. 
We then illustrate a free model of BV RG flow and EA both in the basic and in the extended set up of subsect. 
\cref{subsec:bveffact}. 

All the results illustrated below work strictly speaking in finite dimension.  
Presumably, they can extended also to an infinite dimensional context with limited modifications.


\subsection{\textcolor{blue}{\sffamily The degree --1 symplectic set--up}}\label{subsec:sumbvrg}

In this subsection, we review the degree $-1$ symplectic framework of ref. \ccite{Costello:2007ei}.
Some original results are also reported. 

We begin by setting our notation.
Let $\mathcal{E}=\bigoplus_{\alpha\in\mathbb{Z}}\mathcal{E}_\alpha$ be a graded vector space,
$\mathcal{E}^*=\bigoplus_{\alpha\in\mathbb{Z}}\mathcal{E}^*{}_\alpha$ with $\mathcal{E}^*{}_\alpha=\mathcal{E}_{-\alpha}{}^*$
the dual vector space of $\mathcal{E}$ and $|\!-\!|$ the degree map
and $\prec\!-,-\!\succ$ 
the duality pairing of $\mathcal{E}$, $\mathcal{E}^*$. Homogeneous bases of $\mathcal{E}$, $\mathcal{E}^*$ 
come in dual pairs $a_i$, $a^{*i}$ such that $|a_i|+|a^{*i}|=0$ and $\prec a^{*i},a_j\succ=\delta^i{}_j$. 
For convenience, we set $\epsilon^i=-|a_i|=|a^{*i}|$. Any vector $e\in\mathcal{E}_\alpha$ enjoys the expansions $e=a_ie^i$,
with components $e^i$ satisfying $e^i=0$ for $\epsilon^i\not=-\alpha$. Similarly, 
any covector $l\in\mathcal{E}^*{}_\alpha$ has the expansion $l=l_ia^{*i}$ with components $l_i$
such that $l_i=0$ for $\epsilon^i\not=\alpha$. 

Next, we denote by $\End(\mathcal{E})=\bigoplus_{\alpha\in\mathbb{Z}}\End_\alpha(\mathcal{E})$ with
$\End_\alpha(\mathcal{E})=\bigoplus_{\beta,\gamma\in\mathbb{Z},\gamma-\beta=\alpha}$ 
$\Hom(\mathcal{E}_\alpha,\mathcal{E}_\beta)$, 
the internal endomorphism algebra of $\mathcal{E}$ and by $|\!-\!|$ the associated degree map. 
Given a dual basis pair $a_i$, $a^{*i}$ of $\mathcal{E}$, $\mathcal{E}^*$, 
every endomorphism $A\in\End(\mathcal{E})$ enjoys the expansion $A=a_i\otimes A^i{}_ja^{*j}$ 
with components $A^i{}_j$ such that $A^i{}_j=0$ for $\epsilon^j-\epsilon^i\not=\alpha$ when $|A|=\alpha$. 

In the following analysis, it will be necessary to complement the vector spaces 
$\mathcal{E}$, $\mathcal{E}^*$ by their full degree prolongations $E$, $E^*$. These are just the graded vector 
spaces $E=\bigoplus_{p\in\mathbb{Z}}E_p$, $E^*=\bigoplus_{p\in\mathbb{Z}}E^*{}_p$, 
whose degree $p$ subspaces are given by $E_p=\bigoplus_{\alpha\in\mathbb{Z}}\mathcal{E}_\alpha\otimes\mathbb{R}[\alpha-p]$,
$E^*{}_p=\bigoplus_{\alpha\in\mathbb{Z}}\mathcal{E}^*{}_\alpha\otimes\mathbb{R}[\alpha-p]$. 
Since $\mathcal{E}_\alpha\otimes\mathbb{R}[\alpha-p]=\mathcal{E}_\alpha[\alpha-p]$,
$\mathcal{E}^*{}_\alpha\otimes\mathbb{R}[\alpha-p]=\mathcal{E}^*{}_\alpha[\alpha-p]$,
if we regard $\mathcal{E}_\alpha$, $\mathcal{E}^*{}_\alpha$ as graded vector spaces concentrated in degree $\alpha$,
$E_p$, $E^*{}_p$ can be viewed just as $\mathcal{E}$, $\mathcal{E}^*$ with degree assignment 
reset to the uniform value $p$. It is not difficult to see that $\mathcal{E}\hookrightarrow E$,
$\mathcal{E}^*\hookrightarrow E^*$ and that $E^*$ is dual to $E$ as suggested by the notation. 
With respect to 
a dual basis pair $a_i$, $a^{*i}$ of $\mathcal{E}$, $\mathcal{E}^*$, vectors $e\in E_p$ and covectors $l\in E^*{}_p$ 
can be expressed as $e=a_ie^i$ and $l=l_ia^{*i}$ with components of degrees $|e^i|=p+\epsilon^i$ and
$|l_i|=p-\epsilon^i$, respectively, where $|\!-\!|$ denotes the degree map of $\bigoplus_{p\in\mathbb{Z}}\mathbb{R}[p]$.

The full degree prolongation $\End(E)$ of the internal endomorphism algebra $\End(\mathcal{E})$ of $\mathcal{E}$ will 
also be needed below. This is the graded algebra $\End(E)=\bigoplus_{p\in\mathbb{Z}}\End_p(E)$ where  
$\End_p(E)=\bigoplus_{\alpha,\beta\in\mathbb{Z}}\Hom(\mathcal{E}_\alpha,\mathcal{E}_\beta)\otimes \mathbb{R}[\beta-\alpha-p]$.
Again, as $\Hom(\mathcal{E}_\alpha,\mathcal{E}_\beta)\otimes \mathbb{R}[\beta-\alpha-p]=
\Hom(\mathcal{E}_\alpha,\mathcal{E}_\beta)[\beta-\alpha-p]$, $\End_p(E)$ is $\End(\mathcal{E})$ with degree 
assignment uniformly reset the value $p$. 
It is furthermore straightforward to check that $\End(E)$ is the internal endomorphism algebra of $E$ as suggested by the notation. 
Given a dual basis pair $a_i$, $a^{*i}$,  endomorphisms $A\in\End_p(E)$ can be expressed as 
$A=a_i\otimes A^i{}_ja^{*j}$ with components of degrees $|A^i{}_j|=p+\epsilon^i-\epsilon^j$.

The graded trace of a homogeneous endomorphism $A\in\End(E)$ is given by
\begin{equation}
\grtr(A)=(-1)^{(|A|+1)\epsilon^i}A^i{}_i.
\label{sumbvrg-5}
\end{equation}
It can be shown that $\grtr(A)$ does not depend on basis choices and that $|\grtr(A)|=|A|$. 
Further, one has 
$\grtr(AB)=(-1)^{|A||B|}\grtr(BA)$ for $A,B\in \End(E)$. See ref.  \ccite{Covolo:2012gt}  for a related notion. 

The Euler endomorphism is the endomorphism $F\in\End_0(E)$ given by 
\begin{equation}
F=a_i\otimes\epsilon^i a^{*i}.
\label{sumbvrg-4}
\end{equation}
$F$ is manifestly independent from basis choices. For $p\in\mathbb{Z}$, the endomorphism $(-1)^{pF}\in\End_0(E)$ is defined similarly as 
\begin{equation}
(-1)^{pF}=a_i\otimes(-1)^{p\epsilon^i}a^{*i}. 
\label{sumbvrg-4/1}
\end{equation}

\vspace{-.1cm}\pagebreak
 
We define now the basic structure underlying the degree $-1$ symplectic set--up. 

\begin{defi} \label{def:sumbvrg1}
A degree $-1$ symplectic vector space is a graded vector space $\mathcal{E}$ equipped with 
a symplectic pairing $\langle-,-\rangle_{\mathcal{E}}:\mathcal{E}\times\mathcal{E}\rightarrow\mathbb{R}$ 
such that for homogeneous $e,f\in \mathcal{E}$, $\langle e,f\rangle_{\mathcal{E}}=0$ unless $|e|+|f|=+1$.
\end{defi}

The symplectic pairing $\langle-,-\rangle_{\mathcal{E}}$ can be regarded as degree $-1$ vector space isomorphism 
$\varpi_{\mathcal{E}}:\mathcal{E}\rightarrow\mathcal{E}^*$. Hence, $\varpi_{\mathcal{E}}$ 
can be expanded as $\varpi_{\mathcal{E}}=a^{*i}\otimes\omega_{ij}a^{*j}$
with components $\omega_{ij}$ such that $\omega_{ij}=0$
for $\epsilon^i+\epsilon^j+1\not=0$ forming a non singular antisymmetric matrix $\omega$
with respect to any given basis $a^{*i}$ of $\mathcal{E}^*$. 
The inverse isomorphism $\varpi_{\mathcal{E}}{}^{-1}:\mathcal{E}^*\rightarrow\mathcal{E}$ 
can be expressed similarly as $\varpi_{\mathcal{E}}{}^{-1}=a_i\otimes \omega^{ij}a_j$
with respect to the dual basis $a_i$ of $\mathcal{E}$  
with components $\omega^{ij}$ such that $\omega^{ij}=0$
for $\epsilon^i+\epsilon^j+1\not=0$ forming the non singular antisymmetric matrix $\omega^{-1}$.

Using the matrices $\omega$ and $\omega^{-1}$, one can define the symplectic dual bases $a^i$ and $a^*{}_i$ 
of the bases $a_i$ and $a^{*i}$ of $\mathcal{E}$ and $\mathcal{E}^*$ of an (algebraically) dual pair
as the bases of $\mathcal{E}$ and $\mathcal{E}^*$ given by $a^i=a_j\omega^{ji}$ 
and $a^*{}_i=-\omega_{ij}a^{*j}$. These have degrees $|a^i|=-\epsilon_i$ and $|a^*{}_i|=\epsilon_i$, respectively,
where $\epsilon_i=-\epsilon^i-1$ for convenience, 
and obey the duality relations $\prec a^*{}_i,a^j\succ=-\delta_i{}^j$, $\prec a^{*i},a^j\succ=\omega^{ij}$, 
$\prec a^*{}_i,a_j\succ=-\omega_{ij}$. Vectors $e$, covectors $l$ and endomorphisms $A$ of $\mathcal{E}$, $\mathcal{E}^*$,
$\End(\mathcal{E})$ can be expanded with respect the 
symplectic dual bases $a^i$, $a^*{}_i$ as well. Care must be exercised in dealing with signs, e. g.
$e=-a^ie_i$, $l=-l^ia^*{}_i$, $A=a^i\otimes A_i{}^j a^*{}_j=-a_i\otimes A^i{}^j a^*{}_j=-a^i\otimes A_i{}_j a^*{}^j$
etc. The components with respect the dual bases 
are related to those with respect to the given bases by means of $\omega$ and $\omega^{-1}$ as expected,
e. g. $e_i=-\omega_{ij}e^j$, $l^i=l_j\omega^{ji}$, $A_i{}^j=-\omega_{ik}A^k{}_l\omega^{lj}$, etc. 
The same results hold for the degree prolongations $E$, $E^*$, $\End(E)$ of $\mathcal{E}$, $\mathcal{E}^*$,
$\End(\mathcal{E})$. Below, we work mainly with $E$, $E^*$, $\End(E)$. 

The symplectic pairing $\langle-,-\rangle_{\mathcal{E}}$ can be extended in obvious fashion to 
a shifted graded symplectic pairing $\langle-,-\rangle_E:E\times E\rightarrow \bigoplus_{p\in\mathbb{Z}}\mathbb{R}[p]$, that is
satisfying 
\begin{equation}
\langle e,f\rangle_E=(-1)^{(|e|+1)(|f|+1)}\langle (-1)^{|f|F}f,(-1)^{|e|F}e\rangle_E 
\label{sumbvrg-2}
\end{equation}
for homogeneous $e,f\in E$. The pairing has the property that 
$\langle-,-\rangle_E:E_p\times E_q\rightarrow \mathbb{R}[-p-q+1]$.
In particular, one has $\langle-,-\rangle_E:E_0\times E_0$ $\rightarrow \mathbb{R}[1]$.

\vspace{1mm}\pagebreak 

The transpose of an endomorphism $A\in\End(E)$ is the endomorphism $A^\sim\in\End(E)$  
uniquely characterized by the property that \hphantom{xxxxxxxxx}
\begin{equation}
\langle e,Af\rangle_E=(-1)^{(|A|+1)(|e|+|f|)+|e||f|}\langle (-1)^{|f|F}f,A^\sim (-1)^{|e|F}e\rangle_E 
\label{sumbvrg-1}
\end{equation}
for homogeneous $e,f\in E$. Explicitly, $A^\sim$ is given by 
\begin{equation}
A^\sim=a^i\otimes (-1)^{|A|(\epsilon^i+\epsilon_j)+\epsilon^i\epsilon_j}A^j{}_ia^*{}_j 
\label{sumbvrg-1/0}
\end{equation}
with respect any chosen dual basis pair. One can show that, for homogeneous $A,B\in\End(E)$,
$A^{\sim\sim}=A$ and $(AB)^\sim=(-1)^{1+|A||B|}B^\sim A^\sim$. 
Also, $1_E{}^\sim=-1_E$, $F^\sim=F+1_E$ and $(-1)^{pF}{}^\sim=-(-1)^p(-1)^{pF}$.
It can also be verified that $\grtr A^\sim=\grtr A$, so that $\grtr A=0$ whenever $A^\sim=-A$.

From now on, we concentrate on the degree $0$ subspace $E_0\subset E$. $E_0$ can be given 
a structure of graded manifold by furnishing  it with a set of globally defined graded coordinates
$x^i$. Let $a^{*i}$ be a basis of $\mathcal{E}^*$. Then, for each $i$, the coordinate $x^i$ is the 
element of $E^*$ of degree $|x^i|=\epsilon^i$ corresponding to $a^{*i}$ under the injection $\mathcal{E}^*\hookrightarrow E^*$. 
The symplectic dual coordinates $x_i=-\omega_{ij}x^j$ of degree $|x_i|=\epsilon_i$ can also be
used as coordinates of $E_0$. Employing the basis $a_i$ of $\mathcal{E}$ dual to $a^{*i}$
the coordinates $x^i$ can be conveniently assembled in a degree $0$ linear map $x:E_0\rightarrow E_0$ 
given by $x=a_i\otimes x^i$. Similarly, the dual coordinates $x_i$ can be aggregated into
a degree $-1$ linear map $x^*:E_0\rightarrow E^*{}_{-1}$ given 
by $x^*=a^{*i}\otimes x_i$. By means of $x$ and $x^*$, we can write many relations in index free form.

Since $E_0$ is a graded manifold, the graded algebra $\Fun(E_0)$ of internal smooth functions of $E_0$ 
is given. We denote by $|\!-\!|$ the degree map of $\Fun(E_0)$ as usual. 
We now show that the symplectic structure of $\mathcal{E}$ 
renders $\Fun(E_0)$ naturally a BV algebra.

\begin{defi} \label{def:sumbvrg3}
The canonical degree $-1$ symplectic form of $E_0$ is 
\begin{equation}
\omega_E=\frac{1}{2}\langle dx,dx\rangle_E. \vphantom{\bigg]}
\label{sumbvrg2}
\end{equation}
\end{defi}

\noindent
Expressed in terms of coordinates, 
\begin{equation}
\omega_E=\frac{1}{2}dx_idx^i. \vphantom{\bigg]}
\label{sumbvrg3}
\end{equation}

We denote by $\varDelta_E$ and $(-,-)_E$ the canonical BV Laplacian and bracket on $\Fun(E_0)$ 
associated with $\omega_E$. 

\begin{prop} \label{propi:sumbvrg1}
The BV Laplacian $\varDelta_E$ can be expressed as 
\begin{equation}
\varDelta_Eu=-\frac{1}{2}\grtr\big((x,\otimes \,(x^*,u)_E)_E\big)   
\label{sumbvrg4}
\end{equation}
for $u\in\Fun(E_0)$, where the tensor product operation refers to the target 
spaces $E_0$ and $E^*{}_{-1}$ of the maps $x$ and $x^*$. The BV bracket reads as 
\begin{equation}
(u,v)_E=(u,x^*)_E(x,v)_E
\label{sumbvrg5}
\end{equation}
with $u,v\in\Fun(E_0)$.
\end{prop}

\noindent 
The coordinate expression of $\varDelta_E$ and $(-,-)_E$ is often useful in computations: 
\begin{align}
&\varDelta_Eu=\frac{1}{2}(-1)^{1+\epsilon^i}(x^i,(x_i,u)_E)_E, 
\vphantom{\Big]}
\label{sumbvrg6}
\\
&(u,v)_E=(u,x_i)_E(x^i,v)_E.
\vphantom{\Big]}
\label{sumbvrg7}
\end{align}
$\Fun(E_0)$ is in this way a non singular BV algebra. 

In our analysis of the BV RG, we shall need certain deformations of the BV
Laplacian $\varDelta_E$ and bracket $(-,-)_E$, which we introduce next. 

Let $A\in\End(E)$ with $A^\sim=-A$.

\begin{defi} \label{def:sumbvrg4}
The $A$--deformed BV Laplacian is 
\begin{equation}
\varDelta_Au=-\frac{1}{2}\grtr\big(A(x,\otimes \,(x^*,u)_E)_E\big) \vphantom{\bigg]}
\label{sumbvrg8}
\end{equation}
for $u\in\Fun(E_0)$. 
The $A$--deformed BV bracket is 
\begin{equation}
(u,v)_A=(u,x^*)_EA(x,v)_E
\label{sumbvrg9}
\end{equation}
with $u,v\in\Fun(E_0)$.
\end{defi}

\noindent
The coordinate expressions of $\varDelta_A$ and $(-,-)_A$ read
\begin{align}
&\varDelta_Au=\frac{1}{2}(-1)^{1+(|A|+1)\epsilon^i}A^i{}_j(x^j,(x_i,u)_E)_E,
\vphantom{\Big]}
\label{sumbvrg10}
\\
&(u,v)_A=(u,x_i)_EA^i{}_j(x^j,v)_E.
\vphantom{\Big]}
\label{sumbvrg11}
\end{align}
$\varDelta_A$ and $(-,-)_A$ are to be compared with their undeformed counterparts $\varDelta_E$ and $(-,-)_E$ 
given in \ceqref{sumbvrg6} and \ceqref{sumbvrg7}. 

Deformed Laplacians and brackets obey a host of relations forming an elegant algebraic 
structure. 
Here, however, we shall restrict ourselves to 
mention just the ones which turn out to be useful in the following RG analysis. 

\begin{prop} \label{prop:sumbvrg1/1}
The $A$--deformed BV Laplacian $\varDelta_A$ has the following properties,
\begin{align}
&|\varDelta_Au|=|u|+|A|+1,
\vphantom{\Big]}
\label{dbvalg1}
\\
&\varDelta_A(uvw)=-\varDelta_Auvw-(-1)^{(|A|+1)|u|}u\varDelta_Avw
\vphantom{\Big]}
\label{dbvalg2}
\\
&\hspace{2.3cm}-(-1)^{(|A|+1)(|u|+|v|)}uv\varDelta_Aw+\varDelta_A(uv)w
\vphantom{\Big]}  
\nonumber
\\
&\hspace{2.3cm}+(-1)^{(|u|+|A|+1)|v|}v\varDelta_A(uw)+(-1)^{(|A|+1)|u|}u\varDelta_A(vw)
\vphantom{\Big]}  
\nonumber
\end{align}
for $u,v,w\in \Fun(E_0)$ and \hphantom{xxxxxxxxxxxxxxxxxxxx}
\begin{equation}
\varDelta_A1_{\Fun(E_0)}=0.  
\vphantom{\Big]}
\label{dbvalg4}
\end{equation}
The $A$--deformed BV bracket $(-,-)_A$ can be expressed in terms of $\varDelta_A$ 
\begin{equation}
(u,v)_A=(-1)^{(|A|+1)|u|}\big(\varDelta_A(uv)-\varDelta_Auv-(-1)^{(|A|+1)|u|}u\varDelta_Av\big)
\vphantom{\bigg]}
\label{sumbvrg11/1}
\end{equation}
with $u,v\in\Fun(E_0)$. Furthermore, $(-,-)_A$  enjoys the properties 
\begin{align}
&|(u,v)_A|=|u|+|v|+|A|+1,
\vphantom{\Big]}
\label{dbvalg6}
\\
&(u,v)_A+(-1)^{|A|(|u|+|v|)+(|u|+1)(|v|+1)}(v,u)_A=0,
\vphantom{\Big]}
\label{dbvalg7}
\\
&(u,vw)_A-(u,v)_Aw-(-1)^{(|u|+|A|+1)|v|}v(u,w)_A=0, 
\vphantom{\Big]}
\label{dbvalg9}
\\
&(u,1_{\Fun(E_0)})_A=0
\vphantom{\Big]}
\label{dbvalg10}
\end{align}   
with $u,v,w\in\Fun(E_0)$. 
\end{prop}    

\begin{proof} 
For convenience, we first prove relations \ceqref{dbvalg6}--\ceqref{dbvalg10}
and then relations \ceqref{dbvalg1}--\ceqref{dbvalg4}.
 
For $z\in\Fun(E_0)$, we have $|(x,z)_E|=|z|+1$. Hence,  
$|\langle (u,x)_E,A(x,v)_E\rangle_E|$ $=|u|+1+|A|-1+|v|+1=|u|+|v|+|A|+1$.
By \ceqref{sumbvrg9}, this shows \ceqref{dbvalg6}. 

For $z\in\Fun(E_0)$, we have $|(x,z)_E|=|z|+1$, as previously observed,
and $(x,z)_E=(-1)^{|z|}(z,x)_E$. By \ceqref{sumbvrg-1}, we have so 
\begin{align}
(u,x^*)_E&A(x,v)_E
\vphantom{\Big]}
\label{pbvalg1}
\\
&=\langle (-1)^{(|u|+1)F}(u,x)_E,A(x,v)_E\rangle_E
\vphantom{\Big]}
\nonumber
\\
&=(-1)^{(|A|+1)(|u|+|v|)+(|u|+1)(|v|+1)}\langle (-1)^{(|v|+1)F}(x,v)_E,A^\sim(u,x)_E\rangle_E
\vphantom{\Big]}
\nonumber
\\
&=-(-1)^{|A|(|u|+|v|)+(|u|+1)(|v|+1)}\langle (-1)^{(|v|+1)F}(v,x)_E,A(x,u)_E\rangle_E.
\vphantom{\Big]}
\nonumber
\\
&=-(-1)^{|A|(|u|+|v|)+(|u|+1)(|v|+1)}(v,x^*)_EA(x,u)_E.
\vphantom{\Big]}
\nonumber
\end{align}
By \ceqref{sumbvrg9}, this shows \ceqref{dbvalg7}. 

Using \ceqref{bvalg9}, it is verified that $(x,vw)_E=(x,v)_Ew+(-1)^{|v||w|}(x,w)_Ev$. So, 
\begin{align}
&(u,x^*)_EA(x,vw)_E=(u,x^*)_EA(x,v)_E\rangle_Ew
\vphantom{\Big]}
\label{pbvalg2}
\\
&\hspace{5.3cm}+(-1)^{|v||w|}(u,x^*)_EA(x,w)_Ev.
\vphantom{\Big]}
\nonumber
\end{align}
By \ceqref{sumbvrg9},  using further \ceqref{dbvalg6}, this shows \ceqref{dbvalg9}.

Since $(x,1_{\Fun(E_0)})_E=0$, \ceqref{dbvalg10} follows trivially from \ceqref{sumbvrg9}.

For $z\in\Fun(E_0)$, we have $|(x,\otimes\,(x^*,z)_E)_E|=|z|+1$. It follows from this that 
$|\grtr((-1)^{(|u|+1)F}A(x,\otimes \,(x^*,u)_E)_E)|=|A|+|u|+1$.
By \ceqref{sumbvrg8}, this shows \ceqref{dbvalg1}. 

The proof of the remaining relations requires the identity
\begin{equation}
\varDelta_A(uv)=\varDelta_Auv+(-1)^{(|A|+1)|u|}u\varDelta_Av+(-1)^{(|A|+1)|u|}(u,v)_A. \vphantom{\bigg]}
\label{sumbvrg18}
\end{equation}
To show \ceqref{sumbvrg18}, we express the operator $\varDelta_A$ 
through \ceqref{sumbvrg8} and uses systematically the graded derivation property \ceqref{bvalg9} of 
the bracket $(-,-)_E$. Since $\varDelta_A$ involves a twice iterated bracket, four terms are produced in this way.
Two of these give immediately the first two terms in the right hand side of \ceqref{sumbvrg18}. 
Using that $A^\sim=-A$, the remaining two terms are found to be equal to half the third term of \ceqref{sumbvrg18}. 

Using \ceqref{sumbvrg18}, one finds 
\begin{align}
&\varDelta_A(uvw)=\varDelta_Auvw+(-1)^{(|A|+1)|u|}u\varDelta_A(vw) \hspace{2.5cm}
\vphantom{\Big]}
\label{pbvalg3}
\end{align}
\vspace{-1cm}\eject\noindent
\begin{align}
&\hspace{2.5cm}+(-1)^{(|A|+1)|u|}(u,v)_Aw+(-1)^{(|A|+1)(|u|+|v|)+|u||v|}v(u,w)_A.
\vphantom{\Big]}
\nonumber
\end{align}
Expressing  the brackets $(u,w)_A$, $(v,w)_A$ in the right hand side using \ceqref{sumbvrg18},
we obtain \ceqref{dbvalg2} straightforwardly.

Since $(x^*,1_{\Fun(E_0)})_E=0$, \ceqref{dbvalg4} follows trivially from \ceqref{sumbvrg8}.
Relation \ceqref{sumbvrg11/1} is a simple rearrangement of \ceqref{sumbvrg18}.
\end{proof}

Prop. \cref{prop:sumbvrg1/1} states that the deformed Laplacian $\varDelta_A$ 
enjoys properties generalizing the \ceqref{bvalg1}--\ceqref{bvalg4}
except for the nilpotence relation \ceqref{bvalg3}. Similarly, the deformed bracket 
$(-,-)_A$ exhibits properties generalizing the \ceqref{bvalg6}--\ceqref{bvalg10}
except for the graded Jacobi identity \ceqref{bvalg8}. For $\varDelta_A$  and $(-,-)_A$ to obey analogs of 
\ceqref{bvalg3} and \ceqref{bvalg8}, it is required that $A$ is in addition 
even graded. 

\begin{prop} \label{prop:sumbvrg1/2} 
If $|A|=0$ mod $2$, the Laplacian $\varDelta_A$ is nilpotent, 
\begin{equation}
\varDelta_A{}^2=0.
\label{dbvalg3}
\end{equation}
Moreover, the bracket $(-,-)_A$ satisfies the shifted graded Jacobi identity
\begin{align}
&(-1)^{(|w|+1)(|u|+1)}(u,(v,w)_A)_A
+(-1)^{(|u|+1)(|v|+1)}(v,(w,u)_A)_A
\vphantom{\Big]}
\label{dbvalg8}
\\
&\hspace{5cm}+(-1)^{(|v|+1)(|w|+1)}(w,(u,v)_A)_A=0  
\vphantom{\Big]}
\nonumber
\end{align}
with $u,v,w\in\Fun(E_0)$. 
\end{prop}

\begin{proof}
As  
the components of the brackets $(x,\otimes\, x)_E$, $(x,\otimes\, x^*)_E$ and $(x^*,\otimes\, x^*)_E$ 
are constant elements of $\Fun(E_0)$, 
$[\ad_E x,\otimes\,\ad_E x]=0$, $[\ad_E x,\otimes\,\ad_E x^*]=0$
and $[\ad_E x^*,\otimes\,\ad_E x^*]=0$. Therefore, when commuting $\varDelta_A{}^2$ using \ceqref{sumbvrg8}
we can bring the left most factor $\varDelta_A$ past the rightmost one producing at most signs.
Keeping carefully track of these, we find that $\varDelta_A{}^2=(-1)^{1+|A|}\varDelta_A{}^2$.
\ceqref{dbvalg8} follows. 

For the validity of \ceqref{dbvalg8}, only the reduction mod $2$ of the grading is relevant.
Since $|A|=0$ mod $2$, under the reduction the deformed Laplacian $\varDelta_A$ has the same formal properties as 
an ordinary BV Laplacian and $\varDelta_A$ and the bracket $(-,-)_A$ are related formally 
in the same way as in an ordinary BV algebra. The proof of the shifted graded Jacobi identity for a BV algebra \pagebreak 
therefore goes over unchanged in the present situation. \ceqref{dbvalg8} is so proven. 
\end{proof}

For $A=1_E$, $\varDelta_A$ and $(-,-)_A$ reproduce $\varDelta_E$ and $(-,-)_E$. 
More is true. 

\begin{prop} \label{prop:sumbvrg2}
If $|A|=0$, $\varDelta_A$ is a BV Laplacian over $\Fun(E_0)$ 
and $(-,-)_A$ is the associated BV bracket (cf. defs. \cref{def:bvalg} and \cref{def:bvbrac}).  
Further, if $A$ is invertible, $\varDelta_A$ is non singular.
\end{prop}

\begin{proof}
The first part of the proposition follows immediately from \cref{prop:sumbvrg1/1} and \cref{prop:sumbvrg1/2}
combined. From \ceqref{sumbvrg11}, it is evident that, if $A$ is invertible, $(u,-)_A$ vanishes 
only for $u\in\mathbb{R}1_{\Fun(E_0)}$. 
\end{proof}

Eq.  \ceqref{dbvalg3} is actually a particular case of the more general relation
\begin{equation}
[\varDelta_A,\varDelta_B]=0   
\label{sumbvrg12/0}
\end{equation}
valid for any two antisymmetric endomorphisms $A$, $B$ of arbitrary degree and shown in
the same way. 

\begin{prop} \label{prop:sumbvrg3}
Let $A,B\in\End(E)$ be such that  $A^\sim=-A$, $B^\sim=-B$. Then, for $u,v\in\Fun(E_0)$ one has 
\begin{equation}
\varDelta_A(u,v)_B=(\varDelta_Au,v)_B+(-1)^{(|A|+1)(|B|+|u|+1)}(u,\varDelta_Av)_B+[u,v]_{A,B},
\vphantom{\bigg]}
\label{sumbvrg12}
\end{equation}
where the bracket $[-,-]_{A,B}$ is defined as  
\begin{equation}
[u,v]_{A,B}=\grtr\big(A(x,\otimes\,(u,x^*)_E)_EB(x,\otimes\,(v,x^*)_E)_E\big). \vphantom{\bigg]}
\label{sumbvrg13}
\end{equation}
\end{prop}

\begin{proof}
Relation \ceqref{sumbvrg12} is the result of an explicit calculation. Hence, we provide only an outline of its proof. 

To begin with, one shows that 
\begin{equation}
\varDelta_A(u,v)_E=(\varDelta_Au,v)_E+(-1)^{(|A|+1)(|u|+1)}(u,\varDelta_Av)_E+[u,v]_A, 
\label{sumbvrg16}
\end{equation}
where the bracket $[-,-]_A$ is defined as  
\begin{equation}
[u,v]_A=\grtr\big(A((x,u)_E,\otimes\,(v,x^*)_E)_E\big).
\label{sumbvrg17}
\end{equation}
To prove \ceqref{sumbvrg16}, \ceqref{sumbvrg17}, one expresses the operator $\varDelta_A$ 
through \ceqref{sumbvrg8} and uses systematically the graded Jacobi identity \ceqref{bvalg8} 
obeyed by the bracket $(-,-)_E$. Since $\varDelta_A$ involves a twice iterated bracket, four terms are 
produced in this way. Two of these give readily the first two terms in the right hand side of \ceqref{sumbvrg16}. 
The remaining two terms combine to yield the third term of \ceqref{sumbvrg16}. 
At this stage the antisymmetry relation $A^\sim=-A$ must be used. 

Next, we compute $\varDelta_A(u,v)_B$. Using the defining expression \ceqref{sumbvrg9} of $(u,v)_B$ and the 
relation \ceqref{sumbvrg18}, we find 
\begin{align}
\varDelta_A(u,v)_B&=\varDelta_A(u,x^*)_EB(x,v)_E
\vphantom{\Big]}
\label{sumbvrg19}
\\
&+(-1)^{(|A|+1)(|B|+|u|)}(u,x^*)_EB\varDelta_A(x,v)_E
\vphantom{\Big]}
\nonumber
\\
&+(-1)^{(|A|+1)|u|}((u,x^*),B(x,v)_E)_A.
\vphantom{\Big]}
\nonumber
\end{align}
We now compute $\varDelta_A(u,x^*)_E$ through \ceqref{sumbvrg16} getting 
$\varDelta_A(u,x^*)_E=(\varDelta_Au,x^*)_E$. In the same way, one finds that 
$\varDelta_A(x,v)_E=(-1)^{(|A|+1)}(x,\varDelta_Av)_E$. Using these relations
in \ceqref{sumbvrg19} 
\begin{align}
&\varDelta_A(u,v)_B-(\varDelta_Au,v)_B-(-1)^{(|A|+1)(|B|+|u|+1)}(u,\varDelta_Av)_B
\vphantom{\Big]}
\label{sumbvrg20}
\\
&\hspace{5cm}=(-1)^{(|A|+1)|u|}((u,x^*),B(x,v)_E)_A.
\vphantom{\Big]}
\nonumber
\end{align}
To complete the proof, there remain to show that the right hand side of \ceqref{sumbvrg20} 
is precisely $[u,v]_{A,B}$. This can be done straightforwardly using  \ceqref{sumbvrg9}. 
\end{proof}

\noindent 
Written in terms of coordinates, $[u,v]_{A,B}$ is given by 
\begin{align}
&[u,v]_{A,B}=(-1)^{(|A|+|B|+|u|+|v|+1)\epsilon^i}A^i{}_j
\vphantom{\Big]}
\label{sumbvrg14}
\\
&\hspace{2.8cm}\times (x^j,(u,x_k)_E)_EB^k{}_l(x^l,(v,x_i)_E)_E. 
\vphantom{\Big]}
\nonumber
\end{align}

The following propositions will be used in the analysis of certain BV RG flows. 

\begin{prop} \label{cor:deltacom}
Let $A,B\in\End(E)$ such that $A^\sim=-A$, $|A|=0$ mod $2$ and $B^\sim$ $=-B$.
Then, one has 
\begin{equation}
[\langle x,B\ad_E x\rangle_E,\varDelta_A]=(-1)^{1+|B|}\varDelta_{AB+BA}. 
\label{sumbvrg15}
\end{equation}
\end{prop}

\begin{proof}
Let $u\in\Fun(E_0)$. Then, by \ceqref{sumbvrg18}
\begin{align}
\varDelta_A\langle x,B\ad_Ex \rangle_Eu
&=\varDelta_A(x^*B(x,u)_E)
\vphantom{\Big]}
\label{sumbvrg21}
\\
&=\varDelta_Ax^*B(x,u)_E+(-1)^{|B|+1}x^*B\varDelta_A(x,u)_E
\vphantom{\Big]}
\nonumber
\\
&\hspace{5.2cm}-(x^*,B(x,u)_E)_A.
\vphantom{\Big]}
\nonumber
\end{align}
Now, $\varDelta_Ax^*=0$ by \ceqref{sumbvrg8}. As $\varDelta_Ax=0$ also by \ceqref{sumbvrg8}, using \ceqref{sumbvrg16} 
we find further that $\varDelta_A(x,u)_E=-(x,\varDelta_Au)_E$. Thus, \ceqref{sumbvrg21} can be cast as 
\begin{align}
&\langle x,B\ad_Ex \rangle_E\varDelta_Au-(-1)^{|B|}\varDelta_A\langle x,B\ad_Ex \rangle_Eu
\vphantom{\Big]}
\label{sumbvrg22}
\\
&\hspace{4.7cm}=(-1)^{|B|}(x^*B,(x,u)_E)_A. 
\vphantom{\Big]}
\nonumber
\end{align}
Expressing the bracket $(-,-)_A$ through \ceqref{sumbvrg9}, the right hand side of \ceqref{sumbvrg22}
becomes after a few passages
\begin{align}
(-1)^{|B|}(x^*B,(x,u)_E)_A&=\frac{(-1)^{|B|}}{2}\grtr\big((AB+BA)(x,\otimes \,(x^*,u)_E)_E\big) 
\vphantom{\Big]}
\label{sumbvrg23}
\\
&=(-1)^{1+|B|}\varDelta_{AB+BA}u.   
\vphantom{\Big]}
\nonumber
\end{align}
The direct computation of the left hand side of \ceqref{sumbvrg23} yields actually the operator
$2(-1)^{1+|B|}\varDelta_{AB}$  given by eq. \ceqref{sumbvrg8} with $A$ replaced by $AB$, although 
the endomorphism $AB$ appearing here does not meet the antisymmetry requirement $(AB)^\sim=-AB$ 
in general. 
However, it can be verified that for a non antisymmetry endomorphism $C$, $\varDelta_C$ 
depends only on the antisymmetric part $(C-C^\sim)/2$ of $C$. Hence, it is possible to replace 
$AB$ by $(AB-(AB)^\sim)/2=(AB+BA)/2$. Substituting now 
\ceqref{sumbvrg23} into \ceqref{sumbvrg22}, we obtain finally \ceqref{sumbvrg15}. 
\end{proof}

\begin{prop} \label{cor:deltacom1}
Let $A,B\in\End(E)$ such that $A^\sim=-A$, $|A|=0$ mod $2$ and $B^\sim$ $=-B$.
Then, one has 
\begin{align}
&\langle x,B\ad_E x\rangle_E(u,v)_A=(\langle x,B\ad_E x\rangle_Eu,v)_A
\vphantom{\Big]}
\label{sumbvrg15/1}
\\
&\hspace{2.cm}+(-1)^{|B|(|u|+1)}(u,\langle x,B\ad_E x\rangle_Ev)_A
-(-1)^{|B|(|u|+1)}(u,v)_{AB+BA}
\vphantom{\Big]}
\nonumber
\end{align}
for $u,v\in\Fun(E_0)$. 
\end{prop}

\begin{proof}
Expressing $(u,v)_A$ with \ceqref{sumbvrg11/1} and using \ceqref{bvalg9}, we find
\begin{align}
\langle x,B\ad_E& x\rangle_E(u,v)_A
\vphantom{\Big]}
\label{sumbvrg15/2}
\\
&=(-1)^{|u|}\langle x,B\ad_E x\rangle_E\big(\varDelta_A(uv)-\varDelta_Auv-(-1)^{|u|}u\varDelta_Av\big)
\vphantom{\Big]}
\nonumber
\\
&=(-1)^{|u|}\big(\langle x,B\ad_E x\rangle_E\varDelta_A(uv)-\langle x,B\ad_E x\rangle_E\varDelta_Auv
\vphantom{\Big]}
\nonumber
\\
&\hspace{.5cm}-(-1)^{|B|(|u|+1)}\varDelta_Au\langle x,B\ad_E x\rangle_Ev-(-1)^{|u|}\langle x,B\ad_E x\rangle_Eu\varDelta_Av
\vphantom{\Big]}
\nonumber
\\
&\hspace{1cm}-(-1)^{(|B|+1)|u|}u\langle x,B\ad_E x\rangle_E\varDelta_Av\big).
\vphantom{\Big]}
\nonumber
\end{align}
By virtue of relation \ceqref{sumbvrg15}, using again \ceqref{sumbvrg11/1} and \ceqref{bvalg9}, we obtain 
\begin{align}
\langle x,B&\ad_E x\rangle_E(u,v)_A
\vphantom{\Big]}
\label{sumbvrg15/3}
\\
&=(-1)^{|B|+|u|+1}\big(\varDelta_{AB+BA}(uv)-\varDelta_{AB+BA}uv-(-1)^{|u|(|B|+1)}u\varDelta_{AB+BA}v\big)  
\vphantom{\Big]}
\nonumber
\\
&\hspace{.5cm}+(-1)^{|B|+|u|}\big[\varDelta_A\big(\langle x,B\ad_E x\rangle_Euv+(-1)^{|B||u|}u\langle x,B\ad_E x\rangle_Ev\big)
\vphantom{\Big]}
\nonumber
\\
&\hspace{.7cm}-\varDelta_A\langle x,B\ad_E x\rangle_Euv-(-1)^{|B||u|}\varDelta_Au\langle x,B\ad_E x\rangle_Ev
\vphantom{\Big]}
\nonumber
\\
&\hspace{.9cm}-(-1)^{|B|+|u|}\langle x,B\ad_E x\rangle_Eu\varDelta_Av
-(-1)^{(|B|+1)|u|}u\varDelta_A\langle x,B\ad_E x\rangle_Ev\big]
\vphantom{\Big]}
\nonumber
\\
&=-(-1)^{|B|(|u|+1)}(u,v)_{AB+BA}+(\langle x,B\ad_E x\rangle_Eu,v)_A
\vphantom{\Big]}
\nonumber
\\
&\hspace{.5cm}+(-1)^{|B|(|u|+1)}(u,\langle x,B\ad_E x\rangle_Ev)_A. 
\vphantom{\Big]}
\nonumber
\end{align}
This shows \ceqref{sumbvrg15/1}. 
\end{proof}

In the study of free BV EAs, the computation  below is often employed.

\begin{prop} 
Let $A,B,C\in\End(E)$ be endomorphism such that  $A^\sim=-A$, $B^\sim=+B$, $C^\sim=+C$. Then, 
\begin{align}
&\varDelta_A\langle x,Bx\rangle_E=\grtr(AB),
\vphantom{\Big]}
\label{sumbvrg24}
\\
&(\langle x,Bx\rangle_E,\langle x,Cx\rangle_E)_A=4\langle x,BACx\rangle_E.
\vphantom{\Big]}
\label{sumbvrg25}
\end{align}
\end{prop}

\begin{proof}
A simple calculation furnishes  
\begin{align}
&(x,\langle x,Cx\rangle_E)_E=2Cx,
\vphantom{\Big]}
\label{sumbvrg26}
\\
&(\langle x,Bx\rangle_E,x^*)_E=2x^*B
\vphantom{\Big]}
\label{sumbvrg27}
\end{align}
using which, one finds further 
\begin{equation}
(x,\otimes\,(x^*,\langle x,Bx\rangle_E)_E)_E=-2B.
\label{sumbvrg28}
\end{equation}
From \ceqref{sumbvrg8}, using \ceqref{sumbvrg28}, \ceqref{sumbvrg24} follows immediately. 
From \ceqref{sumbvrg9}, using \ceqref{sumbvrg26}, \ceqref{sumbvrg27}, we obtain \ceqref{sumbvrg25}
readily. 
\end{proof}

\begin{prop}
if $A,B\in\End(E)$ are endomorphism such that $A^\sim=-A$ and $B^\sim=+B$, then 
\begin{equation}
\langle x,A\ad_E x\rangle_E\langle x,Bx\rangle_E
=\langle x,(AB+(-1)^{|A||B|}BA)x\rangle_E.
\label{sumbvrg26/1}
\end{equation}
\end{prop}

\begin{proof}
This follows form an explicit computation, exploiting that $A^\sim=-A$. 
\end{proof}


\subsection{\textcolor{blue}{\sffamily Symplectic gl(1$|$1) structures}}\label{subsec:n2struct}  

$\mathfrak{gl}(1|1)$ structures are recurrent in differential geometry, supersymmetric quantum mechanics 
and topological sigma models. A $\mathfrak{gl}(1|1)$ structure is also one of the constitutive
elements of the symplectic set--up of ref.  \ccite{Costello:2007ei}.

Let $\mathcal{E}$ be a degree $-1$ symplectic vector space and $E$ be its prolongation.

\begin{defi} \label{def:n2struct1}
A $\mathfrak{gl}(1|1)$ structure on $E$ is a set of four endomorphisms
$Q,\overline{Q},H,F$ $\in\End(E)$ of degrees $|Q|=1$, $|\overline{Q}|=-1$, $|H|=0$, $|F|=0$
obeying the graded commutation relations
\begin{align}
&[Q,Q]=0, \qquad [\overline{Q},\overline{Q}]=0,
\vphantom{\Big]}
\label{n2struct4}
\\
&[Q,\overline{Q}]=H,
\vphantom{\Big]}
\label{n2struct5}
\\
&[Q,H]=0, \qquad [\overline{Q},H]=0,
\vphantom{\Big]}
\label{n2struct6}
\\
&[F,Q]=Q,\qquad [F,\overline{Q}]=-\overline{Q},
\vphantom{\Big]}
\label{n2struct8}
\\
&[F,H]=0
\vphantom{\Big]}
\label{n2struct9}
\end{align}
and the transposition conditions
\begin{align}
&Q^\sim=Q,
\vphantom{\Big]}
\label{n2struct1}
\end{align}
\begin{align}
&\overline{Q}^\sim=-\overline{Q},
\vphantom{\Big]}
\label{n2struct2}
\\
&H^\sim=-H,
\vphantom{\Big]}
\label{n2struct3}
\\
&F^\sim=F+1_E.
\vphantom{\Big]}
\label{n2struct3/0}
\end{align}
\end{defi}
\noindent
It is possible to check that the \ceqref{n2struct1}--\ceqref{n2struct3/0} 
and the \ceqref{n2struct4}--\ceqref{n2struct9} are compatible thanks to the 
basic identity $[A,B]^\sim=[A^\sim,B^\sim]$. 

A $\mathfrak{gl}(1|1)$ structure on $E$ is called such because the \ceqref{n2struct4}--\ceqref{n2struct9} 
are the Lie brackets of a standard set of generators of the $\mathfrak{gl}(1|1)$ Lie superalgebra. 
In physical parlance, $Q$, $\overline{Q}$, $H$ and $F$ are named respectively supercharge, conjugate supercharge, 
Hamiltonian and Fermion number. $F$ can be identified with the Euler endomorphism introduced earlier
in subsect. \cref{subsec:sumbvrg} (cf. eq. \ceqref{sumbvrg-4}).

Relations \ceqref{n2struct5}, \ceqref{n2struct8} imply that 
\begin{align}
\grtr Q=0,
\vphantom{\Big]}
\label{n2struct7}
\\
\grtr\overline{Q}=0,
\vphantom{\Big]}
\label{n2struct7/1}
\\
\grtr H=0.
\vphantom{\Big]}
\label{n2struct7/2}
\end{align}
The last two identities follow also from \ceqref{n2struct2}, \ceqref{n2struct3}. $Q,\overline{Q},H$ define 
in this way an $\mathfrak{sl}(1|1)$ structure on $E$. \ceqref{n2struct4}--\ceqref{n2struct6} are indeed 
the basic Lie brackets of the $\mathfrak{sl}(1|1)$ Lie superalgebra. 

In finite dimension, one can  construct an endomorphism $\varphi(H)\in\End(E)$
for any function $\varphi:\mathbb{R}\rightarrow \mathbb{R}$ analytic and of infinite 
convergence radius at $0$, simply by replacing the variable $x\in\mathbb{R}$ by $H$
in the Taylor series of $\varphi(x)$ at $0$. It is easy to see using \ceqref{n2struct3} 
and \ceqref{n2struct6} that $|\varphi(H)|=0$,
\begin{equation}
\varphi(H)^\sim=-\varphi(H)
\label{n2struct10}
\end{equation}
and \hphantom{xxxxxxxxxxxxxxxxxxxxxxxx}
\begin{equation}
[Q,\varphi(H)]=[\overline{Q},\varphi(H)]=0.
\label{n2struct11}
\end{equation}

$\mathfrak{gl}(1|1)$ structures can be defined, and in fact are most useful, in an infinite dimensional context, when the 
vector space $\mathcal{E}$ is equipped also with a Hilbert space structure, such that $H$, $F$ are selfadjoint and 
$Q$, $\overline{Q}$ are reciprocally adjoint. In such a case, $Q$, $\overline{Q}$, $H$ and $F$ are generally 
unbounded differential operators and domain issues arise. Functions of $H$ may be defined using the spectral 
theorem.


\subsection{\textcolor{blue}{\sffamily Free models of BV RG}}\label{subsec:freemod}  

The basic datum required for the construction of free models of the BV RG 
is a degree $-1$ symplectic vector space $\mathcal{E}$ together with a $\mathfrak{gl}(1|1)$ structure
$Q,\overline{Q},H,F$ on the prolongation $E$ of $\mathcal{E}$. 

An application of prop. \cref{prop:bvdbuild} yields the following result. 

\begin{prop} \label{prop:basicrgst}
The operators 
\begin{equation}
\varDelta^0{}_t=\varDelta_{\ee^{-tH}}   
\label{freemod1}
\end{equation}
defined according to \ceqref{sumbvrg8} form a family of non singular 
BV Laplacians on $\Fun(E_0)$ over $\mathbb{R}$. The automorphisms of $\Fun(E_0)$ 
\begin{equation}
\chi^0{}_{t,s}=\exp\bigg(\frac{t-s}{2}\langle x,H\ad_E x\rangle_E\bigg)
\label{freemod2}
\end{equation}
form a BV flow along $\mathbb{R}$ relative to it with zero logarithmic Jacobian family 
\begin{equation}
r_{\chi^0 t,s}=0. 
\label{freemod3}
\end{equation}
\end{prop}

\begin{proof}
As anticipated, the proof is an application of prop. \cref{prop:bvdbuild}. 
We set $o=0$ and $\varDelta_o=\varDelta_E$. We then define $\chi^0{}_{t,s}$ and $r_{\chi^0 t,s}$ 
by \ceqref{freemod2} and \ceqref{freemod3}, respectively.
Since $\frac{1}{2}\langle x,H\ad_E x\rangle_E$ is a degree $0$ derivation of $\Fun(E_0)$, 
$\chi^0{}_{t,s}$ is an automorphism of $\Fun(E_0)$ obeying \ceqref{bvflow1}. 
The vanishing of $r_{\chi^0 t,s}$ makes further \ceqref{bvflow1/1} and \ceqref{bvflow6} with $t=0$
trivially satisfied. It follows that the operators 
\begin{equation}
\varDelta_{\chi^0 t}=\exp\bigg(\frac{t}{2}\langle x,H\ad_E x\rangle_E\bigg)\varDelta_E
\exp\bigg(-\frac{t}{2}\langle x,H\ad_E x\rangle_E\bigg)
\label{freemod4}
\end{equation}
defined in accordance to \ceqref{xbvflow13} constitute a family of non singular BV Laplacians 
on $\Fun(E_0)$ over $\mathbb{R}$ and that $\chi^0{}_{t,s}$ is a BV flow along $\mathbb{R}$ relative to it with $r_{\chi^0 t,s}$ 
as logarithmic Jacobian family.  The only thing left to do is showing that $\varDelta_{\chi^0 t}=\varDelta^0{}_t$ is given by
\ceqref{freemod1}. 

Expanding the right hand side of \ceqref{freemod4}, we have 
\begin{equation}
\varDelta_{\chi^0 t}
=\sum_{n=0}^\infty\frac{t^n}{2^nn!}\big(\ad_{\Fun(E_0)}(\langle x,H\ad_E x\rangle_E)\big)^n\varDelta_E.
\label{freemod5}
\end{equation}
Using prop. \cref{cor:deltacom} with $A=H^n$ with integer $n\geq 0$ and $B=H$, it can be easily 
proven by induction that 
\begin{equation}
\big(\ad_{\Fun(E_0)}(\langle x,H\ad_E x\rangle_E)\big)^n\varDelta_E=(-2)^n\varDelta_{H^n}. \vphantom{\bigg]}
\label{freemod6}
\end{equation}
It follows that \hphantom{xxxxxxxxxxxxxx}
\begin{equation}
\varDelta_{\chi^0 t}=\sss_{n=0}^\infty\frac{(-t)^n}{n!}\varDelta_{H^n}=\varDelta_{\ee^{-tH}} \vphantom{\bigg]}
\label{freemod7}
\end{equation}
as claimed. 
\end{proof}


By prop. \cref{prop:sumbvrg2}, the BV bracket $(-,-)^0{}_t$ associated with $\varDelta^0{}_t$ is 
\begin{equation}
(u,v)^0{}_t=(u,v)_{\ee^{-tH}} \vphantom{\bigg]}
\label{freemod7/1}
\end{equation}
with $u,v\in \Fun(E_0)$, where the bracket $(-,-)_{\ee^{-tH}}$ is 
defined in \ceqref{sumbvrg9}.

A BV Laplacian and bracket 
\ceqref{freemod1}, \ceqref{freemod7/1} were first considered  
in ref. \ccite{Costello:2007ei}.

Next, we consider the family of free actions $S^0{}_t\in\Fun(E_0)$, where 
\begin{equation}
S^0{}_t=-\frac{1}{2}\langle x,Q \ee^{tH}x\rangle_E. \vphantom{\bigg]} 
\label{freemod8}
\end{equation}

\begin{prop} \label{rop:basicrgea}
$S^0{}_t$ obeys eqs. \ceqref{bvrenflow1} and \ceqref{bvrenflow2}
and is therefore a BV RG EA along $\mathbb{R}$. 
\end{prop}

\begin{proof}
We first prove that $S^0{}_t$ obeys the BV ME \ceqref{bvrenflow1}. 
A simple application of relations \ceqref{sumbvrg24}, \ceqref{sumbvrg25} and \ceqref{n2struct11} yields 
\begin{align}
\varDelta^0{}_tS^0{}_t+\,&\frac{1}{2}(S^0{}_t,S^0{}_t)^0{}_t
\vphantom{\Big]}
\label{freemod9}
\\
&=-\varDelta_{\ee^{-tH}}\frac{1}{2}\langle x,Q \ee^{tH}x\rangle_E
+\frac{1}{8}(\langle x,Q \ee^{tH}x\rangle_E,\langle x,Q \ee^{tH}x\rangle_E)_{\ee^{-tH}}
\vphantom{\Big]}
\nonumber
\\
&=-\frac{1}{2}\grtr(Q)+\frac{1}{2}\langle x,Q^2\ee^{tH}x\rangle_E=0,
\vphantom{\Big]}
\nonumber
\end{align}
where in the last step we used that $\grtr Q=0$ and $Q^2=0$ by the first commutation relation \ceqref{n2struct4}
and \ceqref{n2struct7}. Thus, $S^0{}_t$ does indeed satisfy \ceqref{bvrenflow1}. \pagebreak 

Next, we prove that $S^0{}_t$ evolves in accordance with  \ceqref{bvrenflow2}. Using \ceqref{freemod2}, 
\ceqref{freemod3}, we find 
\begin{align}
\hat\chi^0{}_{t,s}S^0{}_s&
=-\frac{1}{2}\exp\bigg(\frac{t-s}{2}\langle x,H\ad_E x\rangle_E\bigg)\langle x,Q \ee^{sH}x\rangle_E
\vphantom{\Big]}
\label{freemod11}
\\
&=-\frac{1}{2}\sum_{n=0}^\infty\frac{(t-s)^n}{2^nn!}\big(\langle x,H\ad_E x\rangle_E\big)^n\langle x,Q \ee^{sH}x\rangle_E.
\vphantom{\Big]}
\nonumber
\end{align}
Using \ceqref{sumbvrg26/1}, it is straightforward to show by induction that 
\begin{equation}
\big(\langle x,H\ad_E x\rangle_E\big)^n\langle x,Q \ee^{sH}x\rangle_E
=\langle x,(2H)^nQ\ee^{sH}x\rangle_E. \vphantom{\bigg]}
\label{freemod12}
\end{equation}
From \ceqref{freemod12},  using \ceqref{n2struct11}, it follows then that \hphantom{xxxxxxxxxxxxxxxx}
\begin{align}
\chi^0{}_{t,s}S^0{}_s
&=-\frac{1}{2}\sum_{n=0}^\infty\frac{(t-s)^n}{n!}\langle x,H^nQ\ee^{sH}x\rangle_E
\vphantom{\Big]}
\label{freemod13}
\\
&=-\frac{1}{2}\langle x,\ee^{(t-s)H}Q\ee^{sH}x\rangle_E=S^0{}_t,
\vphantom{\Big]}
\nonumber
\end{align}
verifying \ceqref{bvrenflow2}. 
\end{proof}

The RGE obeyed by $S^0{}_t$ has therefore the form 
\begin{equation}
\frac{dS^0{}_t}{dt}=\frac{1}{2}\langle x,H(x,S^0{}_t)_E\rangle_E. \vphantom{\bigg]} 
\label{freemod14}
\end{equation}
The Hamiltonian $H$ thus drives the RG flow. 

The RG set--up considered above is evidently of the basic type discussed in subsect.
\cref{subsec:bveffact}. It is natural to wonder whether there exists an analogous construction in the 
extended set--up. The answer is affirmative as we show next. 


\begin{prop} \label{prop:extrgst}
The operators 
\begin{equation}
\varDelta^0{}_{t\theta}=\varDelta_{\ee^{-tH}}+\theta\varDelta_{\overline{Q}\ee^{-tH}}  \vphantom{\bigg]} 
\label{freemod15}
\end{equation}
defined according to \ceqref{sumbvrg8} form a family of non singular 
BV Laplacians on $\Fun(E_0)$ over $T[1]\mathbb{R}$. The automorphisms of $\Fun(E_0)$ 
\begin{align}
&\chi^0{}_{t\theta,s\zeta}=\bigg(\id_{\Fun(E_0)}+\theta\frac{1}{2}\langle x,\overline{Q}\ad_E x\rangle_E\bigg)
\vphantom{\Big]}
\label{freemod16}
\\
&\hspace{3cm}\exp\bigg(\frac{t-s}{2}\langle x,H\ad_E x\rangle_E\bigg)
\bigg(\id_{\Fun(E_0)}-\zeta\frac{1}{2}\langle x,\overline{Q}\ad_E x\rangle_E\bigg)
\vphantom{\Big]}
\nonumber
\end{align}
\vspace{-.9cm}\eject\noindent
form a BV flow along $T[1]\mathbb{R}$ relative to it with zero 
logarithmic Jacobian family
\begin{equation}
r_{\chi^0 t\theta,s\zeta}=0. 
\label{freemod17}
\end{equation}
\end{prop}

\begin{proof} The proof is a straightforward generalization of that of prop. \cref{prop:basicrgst}
exploiting again the results of prop. \cref{prop:bvdbuild}. We set $o=00$ and $\varDelta_o=\varDelta_E$. We then define 
$\chi^0{}_{t\theta,s\zeta}$ and $r_{\chi^0 t\theta,s\zeta}$ 
by \ceqref{freemod16} and \ceqref{freemod17}, respectively.
Since $\frac{1}{2}\langle x,H\ad_E x\rangle_E$ and $\frac{1}{2}\langle x,\overline{Q}\ad_E x\rangle_E$ are degree $0$ and $-1$
derivations of $\Fun(E_0)$, $\chi^0{}_{t\theta,s\zeta}$ is an automorphism of $\Fun(E_0)$.  The verification of \ceqref{bvflow1}
is wieldy. The vanishing of $r_{\chi^0 t,s}$ makes further \ceqref{bvflow1/1} and \ceqref{bvflow6} with $t=00$
trivially satisfied. It follows that the operators 
\begin{align}
\varDelta_{\chi^0 t\theta}&
=\bigg(\id_{\Fun(E_0)}+\theta\frac{1}{2}\langle x,\overline{Q}\ad_E x\rangle_E\bigg)
\exp\bigg(\frac{t}{2}\langle x,H\ad_E x\rangle_E\bigg)
\vphantom{\Big]}
\label{freemod18}
\\
&\hspace{2cm}\times \varDelta_E
\exp\bigg(-\frac{t}{2}\langle x,H\ad_E x\rangle_E\bigg)\bigg(\id_{\Fun(E_0)}
-\theta\frac{1}{2}\langle x,\overline{Q}\ad_E x\rangle_E\bigg)
\vphantom{\Big]}
\nonumber
\end{align}
defined according to \ceqref{xbvflow13} constitute a family of non singular BV Laplacians 
on $\Fun(E_0)$ over $T[1]\mathbb{R}$ and that $\chi^0{}_{t\theta,s\zeta}$ is a BV flow 
along $\mathbb{R}$ relative to it with $r_{\chi^0 t\theta,s\zeta}$ as logarithmic Jacobian. 
To complete the proof, we have to show only that $\varDelta_{\chi^0 t\theta}=\varDelta^0{}_{t\theta}$ is given by
\ceqref{freemod15}. 

The product of the three central factors in the right hand side of \ceqref{freemod18} was computed in the 
paragraph of eqs. \ceqref{freemod5}--\ceqref{freemod7} above. \ceqref{freemod18} so becomes
\begin{align}
\varDelta_{\chi^0 t\theta}&
=\bigg(\id_{\Fun(E_0)}+\theta\frac{1}{2}\langle x,\overline{Q}\ad_E x\rangle_E\bigg)
\vphantom{\Big]}
\label{freemod19}
\\
&\hspace{2cm}\times \varDelta_{\ee^{-tH}}\bigg(\id_{\Fun(E_0)}-\theta\frac{1}{2}\langle x,\overline{Q}\ad_E x\rangle_E\bigg)
\vphantom{\Big]}
\nonumber
\\
&=\varDelta_{\ee^{-tH}}+\theta\frac{1}{2}\big[\langle x,\overline{Q}\ad_E x\rangle_E,\varDelta_{\ee^{-tH}}\big].
\vphantom{\Big]}
\nonumber
\end{align}
Using cor. \cref{cor:deltacom} with $A=\ee^{-tH}$ and $B=\overline{Q}$, we find
\begin{equation}
\big[\langle x,\overline{Q}\ad_E x\rangle_E,\varDelta_{\ee^{-tH}}\big]=2\varDelta_{\overline{Q}\ee^{-tH}}.
\label{freemod20}
\end{equation}
Substituting \ceqref{freemod20} into \ceqref{freemod19}, we get \ceqref{freemod15} immediately.
\end{proof}

An application of props. \cref{prop:sumbvrg1/1} and \cref{prop:sumbvrg2} \pagebreak 
shows that the BV bracket $(-,-)^0{}_{t\theta}$ associated with 
$\varDelta^0{}_{t\theta}$ is given by 
\begin{equation}
(u,v)^0{}_{t\theta}=(u,v)_{\ee^{-tH}}+\theta(-1)^{|u|}(u,v)_{\overline{Q}\ee^{-tH}}
\label{freemod20/1}
\end{equation}
with $u,v\in \Fun(E_0)$, where the brackets $(-,-)_{\ee^{-tH}}$, $(-,-)_{\overline{Q}\ee^{-tH}}$
are defined in accordance with \ceqref{sumbvrg9}.

Next, we consider the family of free actions $S^0{}_t\in\Fun(E_0)$ where 
\begin{equation}
S^0{}_{t\theta}=-\frac{1}{2}\langle x,Q \ee^{tH}x\rangle_E-\theta\frac{1}{4}\langle x,\{\overline{Q},Q\} \ee^{tH}x\rangle_E,
\label{freemod21}
\end{equation}
where $\{\overline{Q},Q\}=\overline{Q}Q-Q\overline{Q}$ is the graded anticommutator of $\overline{Q}$, $Q$. 

\begin{prop} \label{prop:extrgea}
$S^0{}_{t\theta}$ obeys eqs. \ceqref{bvrenflow1} and \ceqref{bvrenflow2}
and is therefore a BV RG EA along $T[1]\mathbb{R}$. 
\end{prop}

\begin{proof}
We first prove that $S^0{}_{t\theta}$ obeys the BV ME \ceqref{bvrenflow1}. 
To organize the computations in a manageable and efficient manner, we let $\varDelta^0{}_t=\varDelta_{\ee^{-tH}}$, 
$\varDelta^{0\star}{}_t=\varDelta_{\overline{Q}\ee^{-tH}}$ 
be the two Laplacians appearing in the right hand side \ceqref{freemod15}.
Similarly, we let $(-,-)^0{}_t=(-,-)_{\ee^{-tH}}$, $(-,-)^{0\star}{}_t=(-,-)_{\overline{Q}\ee^{-tH}}$
be the two brackets in the right hand side of \ceqref{freemod20/1}. Finally, we denote by
$S^0{}_t=-\frac{1}{2}\langle x,Q \ee^{tH}x\rangle_E$, $S^{0\star }{}_t=-\frac{1}{4}\langle x,\{\overline{Q},Q\} \ee^{tH}x\rangle_E$
the two terms in the right hand side of \ceqref{freemod21}. The ME splits in this way into two equations.
The first, $\varDelta^0{}_tS^0{}_t+\frac{1}{2}(S^0{}_t,S^0{}_t)^0{}_t=0$,  was verified earlier on through the computation 
\ceqref{freemod9}. The second, 
$\varDelta^{0\star}{}_tS^0{}_t-\varDelta^0{}_tS^{0\star }{}_t
+\frac{1}{2}(S^0{}_t,S^0{}_t)^{0\star}{}_t-(S^0{}_t,S^{0\star }{}_t)^0{}_t=0$,  is shown similarly 
using systematically the identities \ceqref{sumbvrg24}, \ceqref{sumbvrg25} and \ceqref{n2struct11}, 
\begin{align}
\varDelta^{0\star}{}_tS^0{}_t-\,&\varDelta^0{}_tS^{0\star }{}_t+\frac{1}{2}(S^0{}_t,S^0{}_t)^{0\star}{}_t-(S^0{}_t,S^{0\star }{}_t)^0{}_t
\vphantom{\Big]}
\label{freemod22}
\\
&=-\frac{1}{2}\varDelta_{\overline{Q}\ee^{-tH}}\langle x,Q \ee^{tH}x\rangle_E
+\frac{1}{4}\varDelta_{\ee^{-tH}}\langle x,\{\overline{Q},Q\} \ee^{tH}x\rangle_E
\vphantom{\Big]}
\nonumber
\\
&\hspace{1.5cm}+\frac{1}{8}(\langle x,Q \ee^{tH}x\rangle_E,\langle x,Q \ee^{tH}x\rangle_E)_{\overline{Q}\ee^{-tH}}
\vphantom{\Big]}
\nonumber
\\
&\hspace{3cm}-\frac{1}{8}(\langle x,Q \ee^{tH}x\rangle_E,\langle x,\{\overline{Q},Q\} \ee^{tH}x\rangle_E)_{\ee^{-tH}}
\vphantom{\Big]}
\nonumber
\\
&=-\frac{1}{2}\grtr(\overline{Q}Q)+\frac{1}{4}\grtr(\{\overline{Q},Q\})
\vphantom{\Big]}
\nonumber
\\
&\hspace{2.5cm}+\frac{1}{2}\langle x, Q\overline{Q}Q\ee^{tH}x \rangle_E
-\frac{1}{2}\langle x, Q\{\overline{Q},Q\} \ee^{tH}x \rangle_E=0,
\vphantom{\Big]}
\nonumber
\end{align}
where in the last step we used that $\grtr(\overline{Q}Q)=-\grtr(Q\overline{Q})$ and $Q^2=0$ 
by the first commutation relation \ceqref{n2struct4}. Thus, $S^0{}_{t\theta}$ satisfies \ceqref{bvrenflow1} 
as required. 

Next, we prove that $S^0{}_{t\theta}$ evolves in accordance with \ceqref{bvrenflow2}. Let us write 
the flow maps $\chi^0{}_{t\theta,s\zeta}$ given in \ceqref{freemod16} as 
\begin{equation}
\chi^0{}_{t\theta,s\zeta}
=(1_{\Fun(E_0)}+\theta\chi^{0\star })\chi^0{}_{t,s}(1_{\Fun(E_0)}-\zeta\chi^{0\star }),
\label{freemod23}
\end{equation}
where $\chi^{0\star }=\frac{1}{2}\langle x,\overline{Q}\ad_E x\rangle_E$ and $\chi^0{}_{t,s}=
\exp\big(\frac{t-s}{2}\langle x,H\ad_E x\rangle_E\big)$. Then, 
\begin{equation}
\chi^0{}_{t\theta,s\zeta}S^0{}_{s\zeta}
=(1_{\Fun(E_0)}+\theta\chi^{0\star })\chi^0{}_{t,s}\big(S^0{}_s+\zeta(S^{0\star }{}_s-\chi^{0\star }S^0{}_s)\big).
\label{freemod24}
\end{equation}
Now, using \ceqref{sumbvrg26/1}, we find 
\begin{align}
\chi^{0\star }S^0{}_s
&=-\frac{1}{4}\langle x,\overline{Q}\ad_E x\rangle_E\langle x,Q\ee^{sH}x\rangle_E
\vphantom{\Big]}
\label{freemod25}
\\
&=-\frac{1}{2}\langle x, \overline{Q}Q\ee^{sH}x\rangle_E
\vphantom{\Big]}
\nonumber
\\
&=-\frac{1}{4}\langle x, \{\overline{Q},Q\}\ee^{sH}x\rangle_E=S^{0\star }{}_s. 
\vphantom{\Big]}
\nonumber
\end{align}
In the last step, we employed the relation $(\overline{Q}Q\ee^{sH})^\sim=-Q\overline{Q}\ee^{sH}$, which follows from
\ceqref{n2struct1}, \ceqref{n2struct2} and \ceqref{n2struct10}. Further, 
\begin{equation}
\chi^0{}_{t,s}S^0{}_s=S^0{}_t
\label{freemod26}
\end{equation}
by the same calculation as that of the paragraph of eqs. \ceqref{freemod11}--\ceqref{freemod13}. 
On account of \ceqref{freemod25}, \ceqref{freemod26}, relation 
\ceqref{freemod24} simplifies as 
\begin{align}
\chi^0{}_{t\theta,s\zeta}S^0{}_{s\zeta}
&=(1_{\Fun(E_0)}+\theta\chi^{0\star })\chi^0{}_{t,s}S^0{}_s
\vphantom{\Big]}
\label{freemod27}
\\
&=S^0{}_t+\theta\chi^{0\star } S^0{}_t
\vphantom{\Big]}
\nonumber
\\
&=S^0{}_t+\theta S^{0\star }{}_t=S^0{}_{t\theta}.
\vphantom{\Big]}
\nonumber
\end{align}
\ceqref{bvrenflow2} is so verified. 
\end{proof}

From \ceqref{freemod16} and \ceqref{freemod21}, it is evident the BV flow $\chi^0{}_{t,s}=\chi^0{}_{t0,s0}$
and EA $S^0{}_t=S^0{}_{t0}$ are those of the basic RG set--up given by \ceqref{freemod2} and 
\ceqref{freemod8}, respectively. In  particular, \pagebreak 
the RGE obeyed by $S^0{}_t$ is again \ceqref{freemod14}.
The RG set--up considered above is thus of the extended type discussed in subsect.
\cref{subsec:bveffact} with $\chi^0{}_{t,s}$ and $S^0{}_t$ being the basic projection of 
$\chi^0{}_{t0,s0}$ and $S^0{}_{t\theta}$. 

\begin{prop}
Under basic projection, the RGE can be put in the form
\begin{equation}
\frac{d S^0{}_t}{dt}=-\varDelta_{\overline{Q}\ee^{-tH}}S^0{}_t-\frac{1}{2}(S^0{}_t,S^0{}_t)_{\overline{Q}\ee^{-tH}}
-\frac{1}{2}\grtr(\overline{Q}Q).
\label{freemod28}
\end{equation}
\end{prop}

\begin{proof}
A simple application of the identities \ceqref{sumbvrg24}, \ceqref{sumbvrg25}, \ceqref{n2struct4} and \ceqref{n2struct11}
yields 
\begin{align}
-\frac{1}{2}\varDelta_{\overline{Q}\ee^{-tH}}\langle x,Q\ee^{tH} x\rangle_E
&+\frac{1}{8}(\langle x,Q\ee^{tH} x\rangle_E,\langle x,Q\ee^{tH} x\rangle_E)_{\overline{Q}\ee^{-tH}}
\vphantom{\Big]}
\label{freemod29}
\\
&=-\frac{1}{2}\grtr(\overline{Q}Q)+\frac{1}{2}\langle x,Q\overline{Q}Q\ee^{tH} x\rangle_E
\vphantom{\Big]}
\nonumber
\\
&=-\frac{1}{2}\grtr(\overline{Q}Q)+\frac{1}{2}\langle x,QH\ee^{tH} x\rangle_E
\vphantom{\Big]}
\nonumber
\\
&=-\frac{1}{2}\grtr(\overline{Q}Q)+\frac{d}{dt}\frac{1}{2}\langle x,Q\ee^{tH} x\rangle_E.
\vphantom{\Big]}
\nonumber
\end{align}
Recalling \ceqref{freemod8}, \ceqref{freemod29} can be cast in the form \ceqref{freemod28}. 
\end{proof}


By the results of subsect. \cref{subsec:bveffact}, we expect that 
the RGE \ceqref{freemod14} can be cast in the form \ceqref{bveffact27}.
Eq. \ceqref{freemod28} apparently deviates from \ceqref{bveffact27} by 
the sign of the first two terms. This mismatch is however only apparent.
As we shall argue in subsect. \cref{subsec:permod} below, 
the reduced infinitesimal generator and Jacobian of the BV flow,
$\bar\chi^{0\bcdot}{}_t$ and $\bar r^\bcdot{}_{\chi^0 t}$,  contain those very same terms with
coefficients such to produce the result shown.  

In this subsection, we illustrated two non trivial examples of BV RG flow. These flows are special:
their logarithmic Jacobian vanishes. In the associated RGEs, the quantum term is therefore absent.
In subsect. \cref{subsec:permod} next, we shall encounter instances  
of non unimodular flows and study the corresponding RGEs.


\subsection{\textcolor{blue}{\sffamily Towards perturbative BV RG}}\label{subsec:permod}  

In this final partly conjectural subsection, we consider the implications of the results 
found in subsect. \cref{subsec:freemod} for the perturbative BV RG. 

In perturbation theory, one promotes the function algebra $\Fun(E_0)$ to the formal power series algebra
$\Fun(E_0)[[\hbar]]$ of the parameter $\hbar$ over $\Fun(E_0)$. Further, the full action $S_t\in\Fun(E_0)[[\hbar]]$ 
is rescaled by $\hbar^{-1}$. It is further assumed that $S_t$ satisfies 
the same BV quantum ME as the free action $S^0{}_t$, 
\begin{equation}
\hbar\varDelta_{\ee^{-tH}}S_t+\frac{1}{2}(S_t,S_t)_{\ee^{-tH}}=0.
\label{permod1}
\end{equation}
Aiming to a perturbative analysis of the BV RG, one splits the action $S_t$ as 
\begin{equation}
S_t=S^0{}_t+I_t,
\label{permod2}
\end{equation}
where $S^0{}_t$ 
is the free action \ceqref{freemod8} and 
$I_t\in\Fun(E_0)[[\hbar]]$ is an interaction term at least cubic in $x$ mod $\hbar$. 

\begin{prop}
$I_t$ obeys the BV quantum ME
\begin{equation}
\hbar\varDelta_{\ee^{-tH}}I_t-\mathcal{Q}I_t+\frac{1}{2}(I_t,I_t)_{\ee^{-tH}}=0,
\label{permod3}
\end{equation}
where $\mathcal{Q}$ is the degree $1$ first order differential operator defined by 
\begin{equation}
\mathcal{Q}u=\langle x,Q(x,u)_E\rangle_E
\label{permod4}
\end{equation}
acting on $\Fun(E_0)[[\hbar]]$. 
\end{prop}

\begin{proof} 
The ME \ceqref{permod1} for $S_t$ together with its counterpart for $S^0{}_t$ yield by virtue of \ceqref{permod2}
the equation 
\begin{equation}
\hbar\varDelta_{\ee^{-tH}}I_t+(S^0{}_t,I_t)_{\ee^{-tH}}+\frac{1}{2}(I_t,I_t)_{\ee^{-tH}}=0.
\label{permod5}
\end{equation}
From \ceqref{freemod8}, by \ceqref{sumbvrg9}, \ceqref{sumbvrg27},  
the second term in the right hand side is
\begin{align}
(S^0{}_t,I_t)_{\ee^{-tH}}
&=-\frac{1}{2}\langle(\langle x,Q\ee^{tH} x\rangle_E,x)_E,\ee^{-tH}(x,I_t)_E\rangle_E 
\vphantom{\Big]}
\label{permod6}
\\
&=-\langle x,Q(x,I_t)_E\rangle_E=-\mathcal{Q}I_t. 
\vphantom{\Big]}
\nonumber
\end{align}
We reach thus \ceqref{permod3}. 
\end{proof}

\noindent
Eq. \ceqref{permod3} is of the same form as the ME for $I_t$ obtained by Costello in 
ref. \ccite{Costello:2007ei}. 

The next problem we have to address is determining the BV RG flow determining the 
RGE of the full action $S_t$. We assume as a working hypothesis that 
this is structured as the Polchinski's RGE with seed \pagebreak 
action $S^0{}_t$ reducing to the free RGE \ceqref{freemod28} in the limit of vanishing $I_t$ is 
\begin{equation}
\frac{dS_t}{dt}=\hbar\varDelta_{\overline{Q}\ee^{-tH}}(S_t-2S^0{}_t)
+\frac{1}{2}(S_t,S_t-2S^0{}_t)_{\overline{Q}\ee^{-tH}}-\frac{\hbar}{2}\grtr(\overline{Q}Q).
\label{permod7}
\end{equation}
The last term may be absorbed by adding a constant $\frac{t}{2}\grtr(\overline{Q}Q)$ to $S_t$
and for this reason is usually neglected, but we shall keep it nevertheless. 
There are several indications that \ceqref{permod7} is a likely to be the correct full RGE. 

\begin{prop}
The RGE \ceqref{permod7} can be cast in the form 
\begin{equation}
\frac{dI_t}{dt}=\hbar\varDelta_{\overline{Q}\ee^{-tH}}I_t+\frac{1}{2}(I_t,I_t)_{\overline{Q}\ee^{-tH}}.
\label{permod8}
\end{equation}
\end{prop}

\begin{proof}
Inserting \ceqref{permod2} into \ceqref{permod7}, we obtain
\begin{align}
\frac{dI_t}{dt}
&=\hbar\varDelta_{\overline{Q}\ee^{-tH}}I_t+\frac{1}{2}(I_t,I_t)_{\overline{Q}\ee^{-tH}}
\vphantom{\Big]}
\label{permod9}
\\
&\hspace{2cm}-\frac{d S^0{}_t}{dt}-\hbar\varDelta_{\overline{Q}\ee^{-tH}}S^0{}_t-
\frac{1}{2}(S^0{}_t,S^0{}_t)_{\overline{Q}\ee^{-tH}}-\frac{\hbar}{2}\grtr(\overline{Q}Q).
\vphantom{\Big]}
\nonumber
\end{align}
The second line of the right hand side vanishes on account of \ceqref{freemod28}.
\end{proof}

Eq. \ceqref{permod8} reads compactly as 
\begin{equation}
\frac{dI_t}{dt}=\hbar\varDelta_{\overline{Q}\ee^{-tH}}e^{I_t/\hbar} 
\label{permod10}
\end{equation}
by \ceqref{freemod15}, \ceqref{freemod20/1} and the 
standard identity $\ee^{-u}\varDelta\ee^u=\varDelta u+\frac{1}{2}(u,u)$ of BV theory. 
The formal solution of \ceqref{permod10} is 
\begin{equation}
\ee^{I_t/\hbar}=\exp\bigg(\hbar\int_s^t d\tau \varDelta_{\overline{Q}\ee^{-\tau H}}\bigg)\ee^{I_s/\hbar}.
\label{permod11}
\end{equation}
This relation is the characterizing property of $I_t$ in the analysis of ref. \ccite{Costello:2007ei}. 

The RGE \ceqref{permod7} can be written as 
\begin{equation}
\frac{dS_t}{dt}=\hbar\varDelta_{\overline{Q}\ee^{-tH}}S_t+\frac{1}{2}(S_t,S_t)_{\overline{Q}\ee^{-tH}}
+\bar\chi^\bcdot{}_tS_t+\hbar \bar r^\bcdot{}_{\chi t},
\label{permod12}
\end{equation}
where 
$\bar\chi^\bcdot{}_t$ and $\bar r^\bcdot{}_t,$ are the degree $0$ derivation and element of $\Fun(E_0)$
\begin{align}
&\bar\chi^\bcdot{}_t=-\ad_{\overline{Q}\ee^{-tH}}S^0{}_t,
\vphantom{\Big]}
\label{permod13}
\\
&\bar r^\bcdot{}_{\chi t}=-2\varDelta_{\overline{Q}\ee^{-tH}}S^0{}_t-\frac{1}{2}\grtr(\overline{Q}Q).
\vphantom{\Big]}
\label{permod14}
\end{align}
Eq. \ceqref{permod12} is formally identical to the RGE \ceqref{bveffact27} of the extended RG set--up 
of subsect. \cref{subsec:bveffact}. By itself this is not enough to show that 
our RGE fits in the BV RG framework elaborated in this paper. However, it provides useful indications about 
what the underlying BV RG flow, if there indeed is one as we suppose, may look like.

In the extended RG set--up, the full BV RG EA $S_t$ has degree $-1$ partner
$S^\star {}_t$ and $S_t$ and $S^\star {}_t$ obey eq. \ceqref{bveffact26}. Let us assume that such a partner does 
indeed exist in the present framework. Eq. \ceqref{bveffact26} would then read 
\begin{equation}
\hbar \varDelta_{\ee^{-tH}}S^\star {}_t+(S_t,S^\star {}_t)_{\ee^{-tH}}
=\hbar \varDelta_{\overline{Q}\ee^{-tH}}S_t+\frac{1}{2}(S_t,S_t)_{\overline{Q}\ee^{-tH}}.
\label{permod14/1} 
\end{equation}   
Correspondingly,  the RGE \ceqref{permod12} would take the form 
\begin{equation}
\frac{dS_t}{dt}=\chi^\bcdot{}_tS_t+\hbar r^\bcdot{}_{\chi t},
\label{permod14/2}
\end{equation}
where $\chi^\bcdot{}_t$ and $r^\bcdot{}_t$ are the degree $0$ derivation and element of $\Fun(E_0)$
\begin{align}
&\chi^\bcdot{}_t=-\ad_{\ee^{-tH}}S^\star {}_t+\bar\chi^\bcdot{}_t,
\vphantom{\Big]}
\label{permod16}
\\
&r^\bcdot{}_{\chi t}=\varDelta_{\ee^{-tH}}S^\star {}_t+\bar r^\bcdot{}_{\chi t}.
\vphantom{\Big]}
\label{permod17}
\end{align}
If we manage to show that $\chi^\bcdot{}_t$ and $r^\bcdot{}_t$ meet the hypotheses of prop. \cref{prop:genres1},
then $\chi^\bcdot{}_t$ and $r^\bcdot{}_t$ will be the infinitesimal generator and logarithmic Jacobian 
of a BV RG flow $\chi_{t,s}$. Further, by prop. \cref{prop:genres2}, $S_t$ will be a BV RG EA with respect to $\chi_{t,s}$. 


It is reasonable to hypothesize that, in analogy to $S_t$, $S^\star {}_t$ splits as 
\begin{equation}
S^\star {}_t=S^{0\star }{}_t+I^\star {}_t, 
\label{permod15}
\end{equation}
where $S^{0\star }{}_t$ 
is the partner of the free action $S^0{}_t$ given by the second term in the right hand side of
\ceqref{freemod21} and $I^\star {}_t\in\Fun(E_0)[[\hbar]]$ is similarly the partner cubic in $x$ mod $\hbar$
of the interaction term $I_t$. 

\begin{prop} Under the above assumptions, $\chi^\bcdot{}_t$ and $r^\bcdot{}_{\chi t}$ are given by 
\begin{align}
&\chi^\bcdot{}_t=\frac{1}{2}\langle x, H\ad_E x\rangle_E -\ad_{\ee^{-tH}}I^\star {}_t,
\vphantom{\Big]}
\label{permod18}
\\
&r^\bcdot{}_{\chi t}=\varDelta_{\ee^{-tH}}I^\star {}_t.
\vphantom{\Big]}
\label{permod19}
\end{align}
\end{prop}

\begin{proof} To show the statement, we first substitute relations \ceqref{permod13}, \ceqref{permod14}
and \ceqref{permod15} into \ceqref{permod16}, \ceqref{permod17} and then insert the explicit expressions
of $S^0{}_t$ and $S^{0\star }{}_t$ appearing in \ceqref{freemod21} into the resulting identities. We obtain in this way 
\begin{align}
&\chi^\bcdot{}_t=\frac{1}{4}\ad_{\ee^{-tH}}(\langle x,\{\overline{Q},Q\}\ee^{tH}x\rangle_E)
\vphantom{\Big]}
\label{permod20}
\\
&\hspace{2cm}
+\frac{1}{2}\ad_{\overline{Q}\ee^{-tH}}(\langle x,Q\ee^{tH}x\rangle_E)-\ad_{\ee^{-tH}}I^\star {}_t,
\vphantom{\Big]}
\nonumber
\\
&r^\bcdot{}_{\chi t}=-\frac{1}{4}\varDelta_{\ee^{-tH}}\langle x,\{\overline{Q},Q\}\ee^{tH}x\rangle_E
\vphantom{\Big]}
\label{permod21}
\\
&\hspace{2cm}
+\varDelta_{\overline{Q}\ee^{-tH}}\langle x,Q\ee^{tH}x\rangle_E-\frac{1}{2}\grtr(\overline{Q}Q)
+\varDelta_{\ee^{-tH}}I^\star {}_t. 
\vphantom{\Big]}
\nonumber
\end{align}
Combined application of \ceqref{sumbvrg9}, \ceqref{sumbvrg27} and \ceqref{n2struct11} furnishes
\begin{align}
\frac{1}{4}(\langle x,\{\overline{Q},Q\}\ee^{tH}x\rangle_E,-)_{\ee^{-tH}}
+&\frac{1}{2}(\langle x,Q\ee^{tH}x\rangle_E,-)_{\overline{Q}\ee^{-tH}}
\vphantom{\Big]}
\label{permod22}
\\
&=\frac{1}{2}\langle x,(\{\overline{Q},Q\}+2Q\overline{Q})(x,-)_E\rangle_E
\vphantom{\Big]}
\nonumber
\\
&=\frac{1}{2}\langle x,H(x,-)_E\rangle_E,
\vphantom{\Big]}
\nonumber
\end{align}
where \ceqref{n2struct5} was used in the last step. Inserting \ceqref{permod22} into \ceqref{permod20}, we obtain 
\ceqref{permod18}. Using \ceqref{sumbvrg24} and \ceqref{n2struct11} yields further 
\begin{align}
-\frac{1}{4}\varDelta_{\ee^{-tH}}\langle x,\{\overline{Q},Q\}\ee^{tH}x\rangle_E
&+\varDelta_{\overline{Q}\ee^{-tH}}\langle x,Q\ee^{tH}x\rangle_E
\vphantom{\Big]}
\label{permod23}
\\
&=\frac{1}{4}\grtr(-\{\overline{Q},Q\}+4\overline{Q}Q)
\vphantom{\Big]}
\nonumber
\\
&=\frac{1}{2}\grtr(\overline{Q}Q).
\vphantom{\Big]}
\nonumber
\end{align}
Inserting \ceqref{permod23} into \ceqref{permod21}, we obtain 
\ceqref{permod19}.
\end{proof}

\begin{prop}
Under the assumptions made, 
$\chi^\bcdot{}_t$ and $r^\bcdot{}_t$ are the infinitesimal generator and logarithmic Jacobian 
of a BV RG flow $\chi_{t,s}$. Further, $S_t$ is a BV RG EA with respect to $\chi_{t,s}$. 
\end{prop}

\begin{proof}
We have to show that $\chi^\bcdot{}_t$, $r^\bcdot{}_t$ meet the hypotheses of prop. \cref{prop:genres1}, that is 
that they satisfy eqs. \ceqref{genres1} and \ceqref{genres2}. Let $u,v\in\Fun(E_0)$. By properties 
\ceqref{dbvalg8} and \ceqref{sumbvrg15/1}, we have \pagebreak 
\begin{align}
\chi^\bcdot{}_t(u,v)_{\ee^{-tH}}
&=\langle x,H\ad_E x\rangle_E(u,v)_{\ee^{-tH}}-(I^\star {}_t,(u,v)_{\ee^{-tH}})_{\ee^{-tH}}
\vphantom{\Big]}
\label{}
\\
&=\frac{1}{2}(\langle x,H\ad_E x\rangle_Eu,v)_{\ee^{-tH}}+\frac{1}{2}(u,\langle x,H\ad_E x\rangle_Ev)_{\ee^{-tH}}
\vphantom{\Big]}
\nonumber
\\
&\hspace{.5cm}-(u,v)_{H\ee^{-tH}}-((I^\star {}_t,u)_{\ee^{-tH}},v)_{\ee^{-tH}}-(u,(I^\star {}_t,v)_{\ee^{-tH}})_{\ee^{-tH}}
\vphantom{\Big]}
\nonumber
\\
&=(\chi^\bcdot{}_tu,v)_{\ee^{-tH}}+(u,\chi^\bcdot{}_tv)_{\ee^{-tH}}+\frac{d}{dt}(u,v)_{\ee^{-tH}}.
\vphantom{\Big]}
\nonumber
\end{align}
\ceqref{genres1} thus holds. Next, by \ceqref{sumbvrg12} (with $A=B$ and thus $[-,-]_{A,B}=0$) and 
\ceqref{sumbvrg15}, we have 
\begin{align}
[\chi^\bcdot{}_t,\varDelta_{\ee^{-tH}}]
&=\frac{1}{2}[\langle x,H\ad_E x\rangle_E,\varDelta_{\ee^{-tH}}]-[\ad_{\ee^{-tH}}I^\star {}_t,\varDelta_{\ee^{-tH}}]
\vphantom{\Big]}
\label{}
\\
&=-\varDelta_{H\ee^{-tH}}+\ad_{\ee^{-tH}}\varDelta_{\ee^{-tH}}I^\star {}_t
\vphantom{\Big]}
\nonumber
\\
&=\frac{d}{dt}\varDelta_{\ee^{-tH}}+\ad_{\ee^{-tH}}r^\bcdot{}_{\chi t}. 
\vphantom{\Big]}
\nonumber
\end{align}
Therefore, also \ceqref{genres2} holds. It follows in this way that 
$\chi^\bcdot{}_t$ and $r^\bcdot{}_t$ are the infinitesimal generator and logarithmic Jacobian 
of a BV RG flow $\chi_{t,s}$. Since $S_t$ satisfies eq. \ceqref{permod14/2}, $S_t$ is a BV RG EA
relative to $\chi_{t,s}$ by  prop. \cref{prop:genres2}.  
\end{proof}

The infinitesimal generator and Jacobian $\chi^{0\bcdot}{}_t$ and $r^\bcdot{}_{\chi^0 t}$ 
of the special free BV RG flow $\chi^0{}_{t,s}$ of eq. \ceqref{freemod2} are \hphantom{xxxxxxxx}
\begin{align}
&\chi^{0\bcdot}{}_t=\frac{1}{2}\langle x, H\ad_E x\rangle_E, 
\vphantom{\Big]}
\label{permod24}
\\
&r^\bcdot{}_{\chi^0 t}=0. 
\vphantom{\Big]}
\label{permod25}
\end{align}
Comparison of \ceqref{permod18}, \ceqref{permod19} with \ceqref{permod24}, \ceqref{permod25} shows that 
the full BV RG flow $\chi_{t,s}$ differs form its free counterpart $\chi^0{}_{t,s}$ 
by an amount set by the interacting term partner $I^\star {}_t$. As this latter, $\chi_{t,s}$ has thus 
a perturbative expansion. 

The importance of the BV RG EA partner's interaction term $I^\star {}_t$ should by now be evident, 
although presently we cannot prove the partner's existence in general. $I^\star {}_t$ is related 
to BV RG EA interaction term $I_t$ through eq. \ceqref{permod14/1}. 

\begin{prop} The interaction terms $I_t$ and $I^\star {}_t$ obey
\begin{equation}
\hbar \varDelta_{\ee^{-tH}}I^\star {}_t-\mathcal{Q}I^\star {}_t+(I_t,I^\star {}_t)_{\ee^{-tH}}
=\hbar \varDelta_{\overline{Q}\ee^{-tH}}I_t-\frac{1}{2}\mathcal{H}I_t+\frac{1}{2}(I_t,I_t)_{\overline{Q}\ee^{-tH}},
\label{permod3/1}
\end{equation} 
where $\mathcal{Q}$ is defined in \ceqref{permod4} and $\mathcal{H}$ is the degree $0$ first order differential operator
\begin{equation}
\mathcal{H}u=\langle x,H(x,u)_E\rangle_E.
\label{permod4/1}
\end{equation}
\end{prop}

\begin{proof}
The full action and action partner $S_t$ and $S^\star {}_t$ satisfy eq. \ceqref{permod14/1}. 
By \ceqref{freemod22}, the free action and action partner $S^0{}_t$ and $S^{0\star }{}_t$
obey eq. \ceqref{permod14/1} as well. Insertion of the expansions \ceqref{permod2}, \ceqref{permod15} 
into \ceqref{permod14/1} then yields 
\begin{align}
&\hbar \varDelta_{\ee^{-tH}}I^\star {}_t+(S^0{}_t,I^\star {}_t)_{\ee^{-tH}}+(I_t,I^\star {}_t)_{\ee^{-tH}}
\vphantom{\Big]}       
\label{permod3/2}
\\
&\hspace{2cm}=\hbar \varDelta_{\overline{Q}\ee^{-tH}}I_t+(S^0{}_t,I_t)_{\overline{Q}\ee^{-tH}}
-(I_t,S^{0\star }{}_t)_{\ee^{-tH}}+\frac{1}{2}(I_t,I_t)_{\overline{Q}\ee^{-tH}}.
\vphantom{\Big]}
\nonumber
\end{align}
By a calculation identical to \ceqref{permod6}, one can show that 
\begin{equation}
(S^0{}_t,I^\star {}_t)_{\ee^{-tH}}=-\mathcal{Q}I^\star {}_t.
\label{permod4/2}
\end{equation}
By a similar calculation using \ceqref{n2struct5} and \ceqref{n2struct11}, one finds 
\begin{align}
(S^0{}_t,I_t)_{\overline{Q}\ee^{-tH}}-&(I_t,S^{0\star }{}_t)_{\ee^{-tH}}
\vphantom{\Big]}
\label{}
\\
&=-\frac{1}{2}\langle(\langle x,Q\ee^{tH} x\rangle_E,x)_E,\overline{Q}\ee^{-tH}(x,I_t)_E\rangle_E 
\vphantom{\Big]}
\nonumber
\\
&\hspace{.5cm}-\frac{1}{4}\langle(\langle x,\{\overline{Q},Q\}\ee^{tH} x\rangle_E,x)_E,\ee^{-tH}(x,I_t)_E\rangle_E 
\vphantom{\Big]}
\nonumber
\\
&=-\frac{1}{2}\langle x,(2Q\overline{Q}+\{\overline{Q},Q\})(x,I_t)_E\rangle_E 
\vphantom{\Big]}
\nonumber
\\
&=-\frac{1}{2}\langle x,H(x,u)_E\rangle_E=-\frac{1}{2}\mathcal{H}I_t.
\vphantom{\Big]}
\nonumber
\end{align}
\ceqref{permod3/1} follows. 
\end{proof}


We have collected in this way a number of clues suggesting that Polchinski's equation \ceqref{permod7}
supplemented by eq. \ceqref{permod14/1} may indeed be the ones describing perturbatively 
the BV RG flow of the full interacting theory. We hope to further substantiate this 
claim in future work. 
\vfil\eject

\vspace{.4cm}
\noindent
\textcolor{blue}{\sf Acknowledgements}. 
The author wishes to thank G. Velo for useful discussions.
He also acknowledges financial support from INFN High Energy Physics 
Research Agency under the provisions of the agreement between 
Bologna University and INFN. Finally, he thanks for kind hospitality 
the Instituto Superior T\'ecnico of Lisbon, where a part of this work was done
in occasion of the conference ``Higher Structures 2017''.

\vfil\eject

\end{document}